\documentclass[aps,showpacs,amsmath,amssymb,twocolumn,prx,longbibliography,superscriptaddress,notitlepage,floatfix,10pt]{revtex4-2}

\usepackage[dvips]{graphicx}
\usepackage{bm, bbm,ascmac,amsmath,amssymb,amsthm,mathrsfs,amsfonts,dsfont}
\usepackage{epsfig}
\usepackage{braket}
\usepackage{enumerate}
\usepackage[dvipsnames]{xcolor}
\usepackage{float}
\usepackage{algorithm}
\usepackage{algorithmic}
\usepackage{comment}
\usepackage{appendix}
\usepackage{here}
\usepackage{tabularx}
\usepackage{dcolumn}
\usepackage[caption=false]{subfig}
\usepackage{adjustbox}
\usepackage[hidelinks]{hyperref}
\usepackage{bookmark}
\usepackage{physics}
\usepackage{enumitem}
\usepackage{makecell}

\usepackage{scalerel}

\iftrue
  \usepackage{marginnote}
\fi

\theoremstyle{plain}
\newtheorem{theorem}{Theorem}
\newtheorem{lemma}{Lemma}
\newtheorem{corollary}{Corollary}

\theoremstyle{definition}
\newtheorem{definition}{Definition}

\begin{document}

\newfloat{protocol}{htbp}{idf}
\floatname{protocol}{Protocol~}

\newfloat{resource}{htbp}{idf}
\floatname{resource}{Resource~}

\newfloat{simulator}{htbp}{idf}
\floatname{simulator}{Simulator~}

\newfloat{problem}{htbp}{idf}
\floatname{problem}{Problem~}

\title{Verifiable blind probabilistic error cancellation}

\author{Bo Yang}
\email{Bo.Yang@lip6.fr}
\affiliation{LIP6, Sorbonne Université, CNRS, 4 place Jussieu, 75005 Paris, France}

\author{Elham Kashefi}
\email{Elham.Kashefi@lip6.fr}
\affiliation{LIP6, Sorbonne Université, CNRS, 4 place Jussieu, 75005 Paris, France}
\affiliation{School of Informatics, University of Edinburgh, 10 Crichton Street, EH8 9AB Edinburgh, United Kingdom}

\author{Harold Ollivier}
\email{harold.ollivier@inria.fr}
\affiliation{QAT, DIENS, École Normale Supérieure, PSL University, CNRS, INRIA, 45 rue d'Ulm, Paris 75005, France}

\begin{abstract} 
  Quantum error mitigation (QEM) is an essential tool for mitigating hardware noise without incurring space overhead.
  Yet, its reliability depends on modeling, calibration, and implementation, leaving end-to-end security on untrusted quantum hardware unresolved.
  We address this problem by introducing verifiable blind probabilistic error cancellation (VBPEC), the first secure verification protocol against a fully malicious adversary that integrates QEM.
  VBPEC brings probabilistic error cancellation (PEC), a widely studied QEM technique, within the scope of composable security by formalizing delegated mitigation as a cryptographic resource in the abstract cryptography framework.
  The protocol performs PEC with perfect blindness and an exponentially small security error.
  VBPEC retains the absence of quantum-space overhead from recent statistically-secure verified quantum computation protocols and from PEC.
  The only overhead takes the form of additional repetitions due to the QEM procedure.
  To achieve this, we extend trap-based verification from deterministic pass/fail checks to statistical tests that benefit from QEM and develop a new proof technique that integrates the corresponding additional deviation sources.
  Rather than merely tolerating honest noise below a fixed threshold, VBPEC actively cancels it, enabling correctly mitigated estimates to be accepted with high probability without compromising security.
  Our framework thus establishes an essential route towards secure, reliable, and practical delegated quantum computation on near-future quantum hardware: VBPEC fundamentally improves the practicality of verification.
\end{abstract}

\maketitle

\section{Introduction\label{sec:Introduction}}

Quantum error mitigation (QEM)~\cite{Endo2021Hybrid, Cai2023Quantum, Li2017Efficient, Temme2017Error, Endo2018Practical, Koczor2021Exponential, Huggins2021Virtual, Bravyi2021Mitigating, Yoshioka2022Generalized, Yang2022Efficient} has emerged as a foundational approach for handling the impact of hardware noise while retaining all or almost all of the available physical qubits for computation, thereby enabling useful computations before the fully fault-tolerant regime.
It has become central to recent real-device experiments~\cite{Kandala2019Error, Kim2023Evidence, Yoshioka2025Krylov, Abanin2025Observation} and is expected to remain an integral component of practical quantum computation as the devices scale up.
However, the reliability of such error-mitigated computations typically rests on hardware-specific assumptions, leaving their end-to-end correctness without a rigorous operational guarantee.

A prominent example is probabilistic error cancellation (PEC)~\cite{Temme2017Error, Endo2018Practical}, which suppresses estimation bias by applying the inverse of a noise channel by decomposing it using quasi-probabilities.
Even if PEC underlies several major QEM techniques~\cite{Piveteau2021Error, Suzuki2022Quantum, Mari2021Extending, Kim2023Evidence}, it requires accurate noise characterization and faithful implementation of the error cancellation operations sampled via the quasi-probability distribution.
While recent works have improved noise modeling and error bounding~\cite{VandenBerg2023Probabilistic, VandenBerg2024Techniques, Gupta2023Probabilistic, Govia2025Bounding, Chen2025Disambiguating}, stabilized the device noise through hardware-level control~\cite{Kim2025Error}, and developed model-agnostic execution strategies~\cite{Xie2026Noise}, these advances do not by themselves provide an end-to-end performance guarantee.
This limitation becomes particularly acute for cloud-based access where users may delegate their computation to an untrusted and actively malicious remote server.

To this end, cryptographic protocols for delegated ``Verifiable Blind Quantum Computation (VBQC)''~\cite{Fitzsimons2017Unconditionally, Leichtle2021Verifying, Kapourniotis2024Unifying, Kapourniotis2025Plugging, Yang2025Verifiable} offer a principled route to such guarantees without relying on computational hardness assumptions or on the server's trusted behavior and noise assumptions.
Built on measurement-based quantum computation (MBQC)~\cite{Gottesman1999Demonstrating, Raussendorf2001A, Knill2001A, Danos2007The, Briegel2009Measurement}, these protocols allow a client with single-qubit prepare-and-send capabilities to delegate a computation while preserving perfect blindness and achieving verifiability with exponentially small security error.
Moreover, within the abstract cryptography framework~\cite{Maurer2011Abstract}, they are composably secure: the security guarantees are preserved when the protocol is used as a component of a larger cryptographic system~\cite{Dunjko2014Composable, Kapourniotis2024Unifying}.
The recently proposed ``Verifiable Blind Observable Estimation (VBOE)'' protocol~\cite{Yang2025Verifiable} further extends this paradigm to expectation-value estimation, a computational primitive underlying a wide range of key quantum algorithms.

However, these guarantees do not yet extend to the error-mitigated execution pipelines required by noisy hardware---that is with gate-level noise.
Existing verification protocols can tolerate bounded deviations or reject unreliable executions, but they do not actively mitigate physical noise to improve the acceptance probability beyond some maximum tolerable circuit-level noise~\cite{Leichtle2021Verifying, Kapourniotis2024Unifying}.
On the other hand, QEM can improve the accuracy of estimates, but it is not formulated as a cryptographic resource and does not by itself protect the execution pipeline against adversarial behavior.
Indeed, QEM generally increases the attack surface available to a malicious player as it can not only deviate while executing rounds but also inaccurately report the noise map used by the client to devise its mitigation strategy.
This leaves a fundamental gap between practical noise-resilient quantum computation and its composably secure delegation.

We close this gap by introducing ``Verifiable Blind Probabilistic Error Cancellation (VBPEC)'', obtained by integrating PEC into VBOE.
Concretely, VBPEC enables a client to perform PEC blindly and verifiably on an untrusted quantum server, without relying on trusted hardware behavior or faithful implementation of the mitigation operations.

VBPEC can simultaneously achieve a small security error and a high probability of accepting correctly mitigated estimates.
When the actual deviation matches the reference PEC noise model, the acceptance probability converges exponentially to $1$ as the number of rounds increases. 
The same exponential convergence also persists under moderate model mismatch, demonstrating robustness against inaccuracies in the reference noise model.

Integrating PEC into VBOE is not a trivial extension.
We formalize the corresponding ideal functionality as the ``Secure Delegated Probabilistic Error Cancellation (SDPEC)'' resource, which captures the desired behavior of PEC under untrusted execution: the server learns only public information, and every accepted output is guaranteed to lie within the prescribed error of the ideal expectation value; and whenever the noise map used by the client to perform the error mitigation corresponds to that of the device, it is guaranteed to accept with high probability. 
We prove that VBPEC constructs this resource in the abstract cryptography framework, achieving composable security with perfect blindness and exponentially small security error.
This gives the first end-to-end cryptographic reliability guarantee for a QEM procedure, making PEC a composably secure primitive for delegated quantum computation.

These security guarantees and robustness improvements are enabled by three technical contributions.

First, we show that PEC can be implemented blindly within MBQC.
The key observation is that the quantum one-time pad (QOTP)~\cite{Knill2004Fault, Kern2005Quantum, Dankert2009Exact, Geller2013Efficient, Wallman2016Noise} used for blindness twirls any server deviation into an effective stochastic Pauli channel.
More precisely, the action of the Server can always be represented as a correct implementation of the unitary part of the protocol (i.e., performing the $CZ$ gates creating $G$ and the single-qubit rotations required to implement the measurements), the effective Pauli channel, and then the $\mathsf{X}$ basis used to drive the MBQC computations.
This implies that the harmful part of an effective Pauli deviation is its $\mathsf{Z}$ component, which acts precisely as a flip of the corresponding measurement outcome.
This allows the Pauli cancellation operation defined for PEC to be absorbed entirely into the client’s classical post-processing, thereby hiding the selected cancellation operation from the server along with the target computation.

Second, we extend the use of general trap patterns~\cite{Kapourniotis2024Unifying} from detecting malicious deviations to effectively verifying the noise model required for PEC.
We show that the probability distribution of the server’s effective stochastic Pauli deviation can be reconstructed from trap-output statistics via the Walsh--Hadamard transform.
By incorporating the Walsh--Hadamard transform into the client's post-processing of the test-round data, the client obtains unbiased estimators of the Pauli-deviation probabilities appearing in the PEC noise model.
These estimators allow VBPEC to quantify the mismatch between the server’s actual deviation and the reference PEC noise model, and henceforth to translate this mismatch into a bound on the residual estimation bias.
This extends conventional deterministic verification into a form of statistical conditioning, broadening the role of traps from deviation detection to end-to-end deviation benchmarking.

Third, we develop a proof technique for verification protocols with non-deterministic, estimator-valued test outcomes.
In VBPEC, due to the above test-round construction and the noise that is applied to it even in the honest scenario, they do not yield deterministic round-wise outcomes anymore.
Hence, proof techniques based on conventional round-wise pass/fail criteria~\cite{Leichtle2021Verifying} are no longer directly applicable.
We therefore prove security by simultaneously controlling the finite-shot fluctuations of the mitigated estimator for the target computation and of the trap-based noise-model estimators.
This yields a flexible security analysis for verification protocols whose correctness is formulated at the level of expectation values rather than deterministic classical outcomes.

\section{Preliminaries\label{sec:Preliminaries}}

\subsection{Observable estimation problems}

A generic observable estimation problem $\mathsf{C}$, in a class of computation $\mathfrak{C}$, consists of estimating $\operatorname{Tr}[\rho_{\mathsf{C}}O_{\mathsf{C}}]$, where $\rho_{\mathsf{C}}$ is a state on a Hilbert space $\mathcal{H}_{\mathsf{C}}$ and $O_{\mathsf{C}}\in\mathcal{L}(\mathcal{H}_{\mathsf{C}})$ is a Hermitian observable whose associated measurement is efficiently implementable within the computation class $\mathfrak{C}$.
Without loss of generality, we assume $-\mathbb{I}_{\mathsf{C}} \leq O_{\mathsf{C}} \leq \mathbb{I}_{\mathsf{C}}$, where $\mathbb{I}_{\mathsf{C}}$ is the identity operator on $\mathcal{H}_{\mathsf{C}}$.
Moreover, without loss of generality, we reduce the task to the estimation of a binary observable, whose implementation may require adding one ancillary qubit.
Then the two-outcome POVM $\{M_{+},M_{-}\}$, with $M_{\pm}=(\mathbb{I}_{\mathsf{C}}\pm O_{\mathsf{C}})/2$, defines a binary random variable with outcomes $\pm1$ whose expectation value is $\operatorname{Tr}[\rho_{\mathsf{C}}O_{\mathsf{C}}]$.
By Naimark dilation, this POVM can be realized as a projective binary measurement on an enlarged Hilbert space.
Absorbing this dilation and the final basis change into the computation, we denote the resulting output state by $\rho$ and the corresponding binary observable by $O = |+_{O}\rangle\langle+_{O}| - |-_{O}\rangle\langle-_{O}|$.
Later on, we adopt this observable estimation task to estimate $\operatorname{Tr}[\rho O] = \operatorname{Tr}[\rho_{\mathsf{C}}O_{\mathsf{C}}]$.

Based on this simplification, we first recall the general procedure for observable estimation.
For a binary observable $O$ with outcomes $\pm 1$, measuring $\rho$ in the eigenbasis of $O$ gives the $+1$ outcome with probability
\begin{equation}
  \begin{split}
    p
    = \operatorname{Tr}\left[\rho |+_{O}\rangle\langle+_{O}|\right]
    = \frac{1+\operatorname{Tr}\left[\rho O\right]}{2}.
  \end{split}
\end{equation}
We denote by $y_{i}\in\{-1,1\}$ the outcome of the $i$-th measurement, so that $\operatorname{Pr}[y_{i}=1]=p$ and $\operatorname{Pr}[y_{i}=-1]=1-p$, and hence $\mathbb{E}[y_{i}]=\operatorname{Tr}\left[\rho O\right]$.
The empirical estimator from $N_{c}$ samples is
\begin{equation}
  \mu = \frac{1}{N_{c}}\sum_{i=1}^{N_{c}} y_{i}.
\end{equation}
By Hoeffding's inequality,
\begin{equation}\label{eq:hoeffding_observable_estimation}
  \operatorname{Pr}\left[|\mu-\operatorname{Tr}\left[\rho O\right]|\geq\epsilon\right]
  \leq 2\exp\left(-\frac{\epsilon^{2}}{2}
  N_{c}\right).
\end{equation}
It states that for a fixed $\epsilon$, the probability of the estimator being further away than $\epsilon$ from the true value $\operatorname{Tr}\left[\rho O\right]$ is negligible in $N_{c}$, the number of collected samples.
This motivates the following definition:
\begin{definition}[$\left(\epsilon, \delta\right)$-Observable Estimation]
  Given an observable $O$, a reference state $\rho$, a protocol $(\epsilon, \delta)$-estimates $\operatorname{Tr}\left[\rho O\right]$ if the protocol outputs an estimator $o$ that satisfies
  \begin{equation}\label{eq:estimation_epsilon_delta}
    \operatorname{Pr}\left[|o - \operatorname{Tr}\left[\rho O\right]| \geq \epsilon\right]\leq \delta.
  \end{equation}
  Here, $\epsilon>0$ is the allowed estimation error and $\delta>0$ is an upper bound on the probability of obtaining an estimate outside this error tolerance $\epsilon$.
\end{definition}

\subsection{Probabilistic error cancellation (PEC)\label{sec:pec}}

Probabilistic error cancellation (PEC)~\cite{Temme2017Error, Endo2018Practical} is a QEM technique that estimates the ideal expectation value of an observable by virtually implementing the inverse of a reference noise channel via a quantum-classical hybrid procedure.
Its performance critically relies on the accuracy of the reference noise model, which has motivated the development of various noise-learning methods for characterizing device noise~\cite{VandenBerg2023Probabilistic, VandenBerg2024Techniques, Gupta2023Probabilistic, Govia2025Bounding, Chen2025Disambiguating, Kim2025Error}.
Since our aim here is to securely perform PEC, we focus on the noise-inversion step and treat the reference CPTP map $\mathcal{E}_{\mathrm{PEC}}$ as given, independently of how it is obtained.
In fact, our protocol will eliminate the need to trust the noise model and instead assess whether deviations occurring during the delegated computation correspond to that of the given $\mathcal{E}_{\mathrm{PEC}}$ or to a possibly malicious behavior.

\begin{resource*}[htbp]
  \caption{\raggedright Blind Delegated Probabilistic Error Cancellation (BDPEC)}
  \label{resource:bdpec}
  \begin{algorithmic}[0]
    \STATE \textbf{Public Information:}
    Public information $(\mathfrak{C}, G, f_{G}, N_{c}, \mathcal{E}_{\mathrm{PEC}})$, where $\mathcal{E}_{\mathrm{PEC}}$ is a classical description of the PEC noise model.

    \STATE \textbf{Client's interface:}
    The target computation $\mathsf{C}\in \mathfrak{C}$ and its associated measurement angles $\{\phi_{v}\}_{v\in V}$ to produce the state $\rho$ to be measured by $O$.

    \STATE \textbf{Server's interface:}
    \begin{enumerate}
      \item The interface is filtered so that when $e=0$, the interface does not send any information nor take inputs.
      \item For $e=1$, the Resource receives a quantum state $\sigma$ and $F$, a list of instructions so that the resource produces $s \in \mathbb R$.
    \end{enumerate}

    \STATE \textbf{Processing by the Resource:}
    \begin{enumerate}
      \item If $e=0$, it sets $\displaystyle o = \frac{1}{N_{c}} \sum_{i = 1}^{N_{c}}\left\|\vec{q}\right\|_{1}\operatorname{sign}\left(\vec{q}(\mathsf{P}_{i})\right)y_{i}$ with $y_{i}=2Y_{i}-1$ where $Y_{i} \sim \mathcal{B}(p)$, sampled from a Bernoulli distribution $\mathcal{B}(p)$ with $p = \operatorname{Tr}\left[\left(\mathsf{P}_{i}\circ\mathcal{E}_{\mathrm{PEC}}\right)\left(\rho\right) |+_{O}\rangle\langle+_{O}|\right]$, where $\mathsf{P}_{i}\in\{\mathsf{I},\mathsf{Z}\}^{|V|}$ is sampled with probability $|\vec{q}(\mathsf{P}_{i})|/\left\|\vec{q}\right\|_{1}$.
      \item If $e=1$, it computes $o$ using the transmitted state $\sigma$ and $F$.
      \item It forwards $o$ to the Client.
    \end{enumerate}
  \end{algorithmic}
\end{resource*}

Since the inverse map $\mathcal{E}_{\mathrm{PEC}}^{-1}$ is generally not CPTP, it cannot be implemented directly as a physical quantum operation in general.
PEC circumvents this obstruction by expressing $\mathcal{E}_{\mathrm{PEC}}^{-1}$ as a linear combination of implementable CPTP maps.
Namely, one chooses a set of CPTP maps $\{\mathcal{F}_{l}\}_{l=1}^{L}$ and real coefficients $\{q_{l}\}_{l=1}^{L}$ as a quasi-probability distribution, such that
\begin{equation}
  \begin{split}
    \mathcal{E}_{\mathrm{PEC}}^{-1}
    = \sum_{l=1}^{L} q_{l}\mathcal{F}_{l}.
  \end{split}
\end{equation}
This converts the expectation value of an observable $O$ for a state $\rho$ to the following form,
\begin{equation}
  \begin{split}
    \operatorname{Tr}\left[\rho O\right]
     & = \operatorname{Tr}\left[\left(\mathcal{E}_{\mathrm{PEC}}^{-1} \circ \mathcal{E}_{\mathrm{PEC}}\right) \left(\rho\right) O\right]                                                                                                            \\
     & = \sum_{l=1}^{L} q_{l} \operatorname{Tr}\left[\left(\mathcal{F}_{l} \circ \mathcal{E}_{\mathrm{PEC}}\right) \left(\rho\right) O\right]                                                                                                       \\
     & = \sum_{l=1}^{L} \left\|\vec{q}\right\|_{1}\operatorname{sign}\left(q_{l}\right) \frac{|q_{l}|}{\left\|\vec{q}\right\|_{1}} \operatorname{Tr}\left[\left(\mathcal{F}_{l} \circ \mathcal{E}_{\mathrm{PEC}}\right) \left(\rho\right) O\right],
  \end{split} \label{eq:pec}
\end{equation}
where $\displaystyle \left\|\vec{q}\right\|_{1} = \sum_{l=1}^{L} \left|q_{l}\right|$ is the scaling factor so that $\{|q_{l}|/\left\|\vec{q}\right\|_{1}\}_{l=1}^{L}$ becomes a valid probability distribution.

Suppose that we run $N$ samples during PEC.
For the $i$-th circuit execution, the PEC procedure follows first by choosing $\mathcal{F}_{l_{i}}$ from the Monte-Carlo sampling according to the probability $|q_{l_{i}}|/\left\|\vec{q}\right\|_{1}$.
The cancellation operation $\mathcal{F}_{l_{i}}$ is then inserted into the original circuit, yielding a measurement outcome $\tilde{y}_{i}\in\{-1,1\}$.
One estimates $\operatorname{Tr}\left[\left(\mathcal{F}_{l_{i}} \circ \mathcal{E}_{\mathrm{PEC}}\left(\rho\right)\right) O\right]$ by averaging over the outcomes of the rounds satisfying $l_i = l$.
The PEC-ed estimator of $\operatorname{Tr}\left[\rho O\right]$ is then obtained by averaging the whole measurement results with the weight $\left\|\vec{q}\right\|_{1} \operatorname{sign}\left(q_{l_{i}}\right)$ through classical post-processing.
Thus, the PEC-ed estimator for an $N$-round protocol is obtained by
\begin{equation}
  \begin{split}
    \frac{1}{N}\sum_{i=1}^{N} \left\|\vec{q}\right\|_{1}\operatorname{sign}\left(q_{l_{i}}\right) \tilde{y}_{i}. \label{eq:pec_classical_processing}
  \end{split}
\end{equation}
The schematic illustration of PEC is depicted in Fig.~\ref{fig:pec}.
\begin{figure}[htbp]
  \centering
  \includegraphics[width=\linewidth]{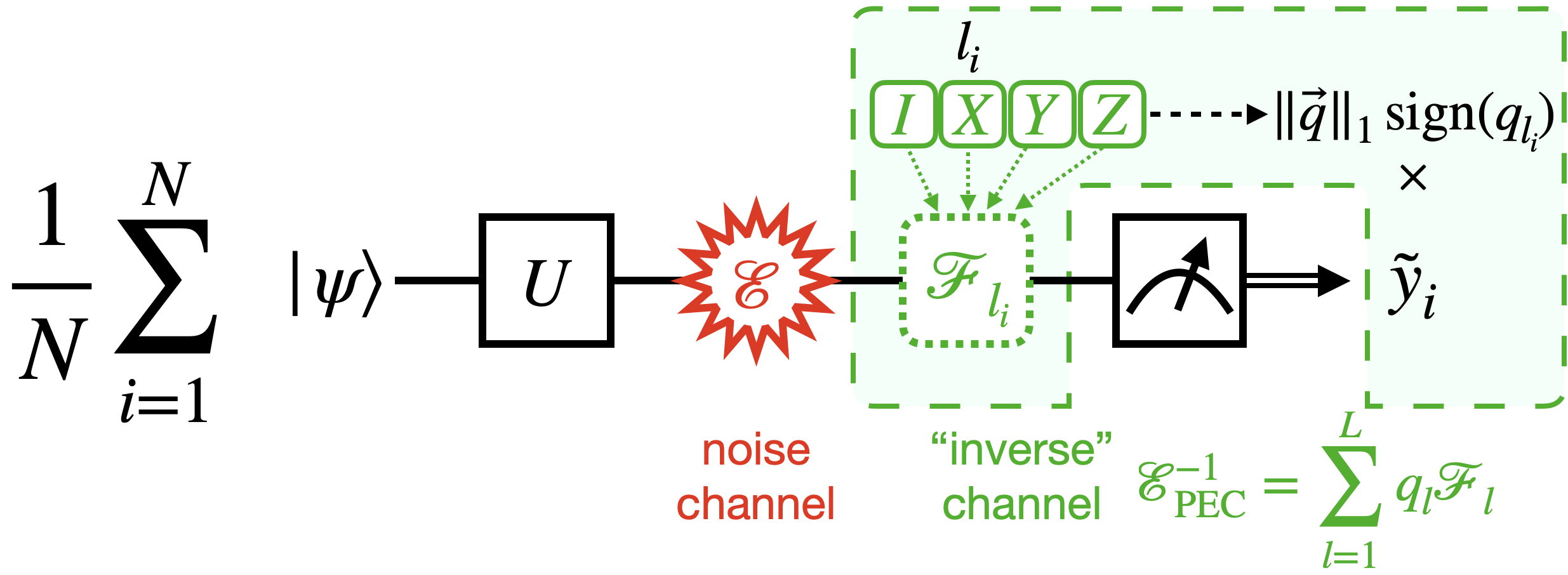}
  \caption{
    Schematic illustration of the PEC process for a single-qubit operation affected by a stochastic Pauli noise $\mathcal{E}_{\mathrm{PEC}}$.
    For each sample $i\in[N]$, the cancellation operation is chosen from $\mathcal{F}_{l_{i}}\in\{\mathsf{I},\mathsf{X},\mathsf{Y},\mathsf{Z}\}$ via the Monte-Carlo sampling.
    More generally, PEC can be implemented locally by inserting noise-cancellation operators at selected locations, based on the corresponding local noise models.
    After completing the total $N$ rounds, the classical post-processing following Eq.~\eqref{eq:pec_classical_processing} yields the mitigated estimate by PEC.
  }
  \label{fig:pec}
\end{figure}

For ease of presentation, we assume throughout the paper that the reference $\mathcal{E}_{\mathrm{PEC}}$ is a stochastic Pauli channel and is invertible in the Pauli transfer representation.
This ensures that the quasi-probability vector $\vec{q}$ is uniquely determined by the noise model.
Note that in the protocols introduced below, any deviation from the server is effectively transformed by QOTP into a stochastic Pauli channel, i.e., a convex combination of Pauli operators, and the assumption of an invertible Pauli transfer matrix can be easily removed by specifying $\vec{q}$ directly, together with $\mathcal{E}_{\mathrm{PEC}}$.

\section{Universal blind probabilistic error cancellation (UBPEC)\label{sec:ubpec}}

We first make PEC blind by incorporating it with the ``Universal Blind Quantum Computation (UBQC)'' protocol~\cite{Broadbent2009Universal}, which relies on the measurement-based quantum computation (MBQC) model~\cite{Gottesman1999Demonstrating, Raussendorf2001A, Knill2001A, Danos2007The, Briegel2009Measurement}.
In UBQC, the Client can prepare and sequentially send single qubits to the Server.
In a delegation scenario, all operations performed by the Server are untrusted: in particular, the Server may introduce arbitrary deviations.
On the Client side, we assume only that possible preparation imperfections can be modeled as a single deviation independent of secret parameters used by the protocol.
This allows us, without loss of generality, to model the Client as ideal and to absorb the imperfections into the deviation acting outside the Client.

\begin{protocol*}[htbp]
  \caption{\raggedright Universal Blind Probabilistic Error Cancellation (UBPEC) --- every single round}
  \label{protocol:ubpec_single}
  \begin{algorithmic}[0]
    \STATE \textbf{Public Information:}
    Public information $(\mathfrak{C}, G, f_{G}, \mathcal{E}_{\mathrm{PEC}})$.
    \STATE \textbf{Inputs from Client:}
    The target computation $\mathsf{C}\in \mathfrak{C}$ that estimates $\operatorname{Tr}\left[\rho O\right]$ and the error cancellation operation $\mathsf{P}\in\{\mathsf{I},\mathsf{Z}\}^{|V|}$.
    
    \STATE \textbf{Protocol:}
    \begin{enumerate}

      \item The Client generates secret parameters:
            \begin{enumerate}
              \item ($\mathsf{X}$ randomization) The Client chooses a random bit $\mathrm{a}_{v}^{\mathrm{init}}\in\{0,1\}$ for $v\in \mathrm{I}$, the set of input qubits, and sets $\mathrm{a}_{v}^{\mathrm{init}}=0$ for $v\in V\setminus \mathrm{I}$.
                    The Client also computes $\displaystyle \mathrm{a}_{v}^{\mathrm{prop}} = \bigoplus_{j\in N_{G}(v)}\mathrm{a}_{j}^{\mathrm{init}} \in\{0,1\}$ for all $v\in V$.
              \item ($\mathsf{Z}$ randomization) The Client chooses a random bit $\mathrm{r}_{v}\in\{0,1\}$ for all $v\in V$.
              \item (randomization for blindness) The Client chooses a random $\theta_{v}\in\Theta$ for all $v \in V$.
              \item ($\mathsf{Z}$-flip for PEC) For every $v\in V$, the Client sets $\mathrm{z}_{v}=1$ if $\mathsf{P}(v) = \mathsf{Z}$, or $\mathrm{z}_{v}=0$ otherwise.
            \end{enumerate}

      \item The Client prepares and sends to the Server all single qubits corresponding to $v\in V$.
            For $v\in \mathrm{I}$, the Client sequentially sends each qubit in $\displaystyle \left(\bigotimes_{v\in\mathrm{I}}\mathsf{Rz}_{v}(\theta_{v})\mathsf{X}_{v}^{\mathrm{a}_{v}^{\mathrm{init}}}\right)\left[\rho_{\mathrm{init}}\right]$.
            For $v\in V\setminus\mathrm{I}$, the Client sends $\displaystyle |+_{\theta_{v}}\rangle$.

      \item The Server applies the entangling operation $\displaystyle \mathsf{G}=\prod_{(u,v)\in E}\mathsf{CZ}_{u,v}$ to prepare the resource state according to the graph $G$.

      \item For each $v\in V$, the Client and Server interactively perform the MBQC process by exchanging classical messages specifying measurement angles and reporting measurement outcomes.
            Once the Client receives the measurement outcome $\mathrm{b}_{j}\in\{0,1\}$ for all $j\in S_{X,v}\cup S_{Z,v}$, where $S_{X,v} = f_{G}^{-1}\left(v\right), S_{Z,v} = \left\{j~|~v\in N_{G}\left(f_{G}\left(j\right)\right)\right\}$, the Client computes the adaptive angle update $\phi_{v}^{\prime}$.
            The Client then computes the measurement angle $\delta_{v}$:
            \begin{equation}\label{eq:delta_UBPEC}
              \begin{split}
                s_{X,v} = \bigoplus_{j \in S_{X,v}} \mathrm{b}_{j} \oplus \mathrm{r}_{j} \oplus \mathrm{z}_{j}, &
                \quad s_{Z,v} = \bigoplus_{j \in S_{Z,v}} \mathrm{b}_{j} \oplus \mathrm{r}_{j} \oplus \mathrm{z}_{j}, \\
                \phi_{v}^{\prime} = (-1)^{s_{X,v}} \phi_{v}+s_{Z,v} \pi,                                        &
                \quad \delta_{v} = \left(-1\right)^{\mathrm{a}_{v}^{\mathrm{init}}} \phi_{v}^{\prime} + \theta_{v}+\left(\mathrm{r}_{v} + \mathrm{a}_{v}^{\mathrm{prop}}\right) \pi.
              \end{split}
            \end{equation}
            Note that $s_{X,v} = s_{Z,v} = 0$ for $v\in \mathrm{I}$.
            The Client sends to the Server the angle $\delta_{v}$ and the Server returns to the Client a bit $\mathrm{b}_{v}\in\{0,1\}$ as a measurement result of qubit $v$ with basis $\{|+_{\delta_{v}}\rangle, |-_{\delta_{v}}\rangle\}$.

      \item The Client then outputs $\tilde{y}\in\{-1,1\}$ corresponding to the decoded measurement outcome for the output qubit.
    \end{enumerate}
  \end{algorithmic}
\end{protocol*}

\begin{protocol*}[htbp]
  \caption{\raggedright Universal Blind Probabilistic Error Cancellation (UBPEC) --- the whole protocol}
  \label{protocol:ubpec_whole}
  \begin{algorithmic}[0]
    \STATE \textbf{Public Information:}
    Public information $(\mathfrak{C}, G, f_{G}, N_{c}, \mathcal{E}_{\mathrm{PEC}})$.
    \STATE \textbf{Inputs from Client:}
    The target computation $\mathsf{C}\in \mathfrak{C}$ that estimates $\operatorname{Tr}\left[\rho O\right]$.
    \STATE \textbf{Protocol:}
    \begin{enumerate}
      \item The Client chooses the total number of rounds $N_{c}$.
      \item For every round indexed by $i \in [N_{c}]$, the Client chooses $\mathsf{P}_{i}\in\{\mathsf{I},\mathsf{Z}\}^{|V|}$ with probability $|\vec{q}(\mathsf{P}_{i})|/\left\|\vec{q}\right\|_{1}$, and delegates the target computation via the UBPEC protocol (Protocol~\ref{protocol:ubpec_single}), with inputs $\mathsf{C}$ and $\mathsf{P}_{i}$.
      \item Upon receiving and decoding the result of the computation round $i \in [N_{c}]$, the Client assigns the result to $\tilde{y}_{i}$.
            The Client then sets $\displaystyle \tilde{o} = \frac{1}{N_{c}}\sum_{i\in [N_{c}]} \left\|\vec{q}\right\|_{1}\operatorname{sign}\left(\vec{q}(\mathsf{P}_{i})\right)\tilde{y}_{i}$ and returns $\tilde{o}$ as the final result.
    \end{enumerate}
  \end{algorithmic}
\end{protocol*}

Building on UBQC, we introduce the ``Universal Blind Observable Estimation (UBOE)'' protocol by adding classical post-processing and using UBQC as a subroutine.
Namely, we let the UBOE protocol perform multiple rounds of UBQC and average the results to estimate the observable $O$ for a reference state $\rho$.
Since this averaging process is performed on the Client's side and UBQC is perfectly blind and composably secure, the UBOE protocol is also perfectly blind and composably secure.
A detailed explanation of MBQC, UBOE, and its security analysis within abstract cryptography can be found in Appendix~\ref{appendix:bdqc}.
Below, we integrate PEC with the UBOE protocol and show that it is also composably secure with respect to blindness.

We first specify the behavior of the ideal ``Blind Delegated Probabilistic Error Cancellation (BDPEC)'' resource which the protocol aims to construct.
With an honest Server accessing the filtered interface (i.e., $e=0$), the ideal resource should on the client's interface the unbiased, error-mitigated estimator of the ideal expectation value $\operatorname{Tr}[\rho O]$ as one would obtain running PEC using the noise map $\mathcal{E}_{\mathrm{PEC}}$
Otherwise, for a malicious Server accessing its unfiltered interface (i.e., $e=1$), it can provide an arbitrary deviation to the ideal functionality which outputs the correspondingly affected biased estimator.
The full behavior is formally described in Resource~\ref{resource:bdpec}.

In order to construct the BDPEC resource, note that PEC can be naturally applied to UBQC via the Client's classical processing, in a manner similar to the Pauli frame~\cite{Fowler2012Surface, Suzuki2022Quantum} and Clifford propagation~\cite{Scheiber2025Reduced}.
Indeed, the quantum one-time pad (QOTP) used in UBQC twirls any Server's deviation into a stochastic Pauli channel. More precisely, the action of the Server can always be represented as a correct implementation of the unitary part of the protocol (i.e., performing the $CZ$ gates creating $G$ and the single-qubit rotations required to implement the measurements), the effective Pauli channel, and then the $\mathsf{X}$ basis used to drive the MBQC computations.
As such, the effective noise model $\mathcal{E}_{\mathrm{PEC}}$ can be reduced to having only $\mathsf{I}$ and $\mathsf{Z}$ terms. This is because $\mathsf{X}$ errors are harmless when they occur just before $\mathsf{X}$ measurements. As a consequence, the inversion operators $\mathcal{F}_l$ of Eq.~\ref{eq:pec} can be chosen as tensor products of $\mathsf{I}$ and $\mathsf{Z}$ only. In turn, this can be implemented either by adding a $\pi$ rotation to the measurement angle or, equivalently, by performing a classical bit flip on the measurement result received from the Server.

To describe the Server's deviation and the PEC noise model, we introduce some notations that will be used throughout the analysis that follows.
Let $\Omega_{V} = \{0,1\}^{|V|}$.
We then let $\vec{p}\in\mathbb{R}_{\geq0}^{|\Omega_{V}|}$ denote the coefficients of the convex combination of $\mathsf{Z}$-Pauli strings describing the Server's deviation empirically averaged over the $N_{c}$ rounds of the protocol, and let $\vec{q}\in\mathbb{R}^{|\Omega_{V}|}$ denote the quasi-probability distribution of the inverse of the noise model for PEC, that is, 
\begin{equation}
  \begin{split}
    \mathcal{E}_{\mathrm{PEC}}^{-1} = \sum_{\mathsf{P}\in\{\mathsf{I}, \mathsf{Z}\}^{|V|}}\vec{q}(\mathsf{P})\mathsf{P}(\cdot)\mathsf{P},
  \end{split}
\end{equation}
where $\vec{q}(\mathsf{P})$ are the elements of $\vec{q}$ indexed by $\mathsf{P}\in\{\mathsf{I}, \mathsf{Z}\}^{|V|}$.
We also let $\mathsf{P}(v)\in\{\mathsf{I}, \mathsf{Z}\}$ denote the Pauli operation of $\mathsf{P}$ at qubit $v$.

Without loss of generality, we identify $\{\mathsf{I}, \mathsf{Z}\}^{|V|}$ and $\Omega_{V}$ by matching $\mathsf{I}$ with $0$ and $\mathsf{Z}$ with $1$.
We also define the probability distribution of the noise model for PEC as $\vec{p}_{\mathrm{PEC}}\in\mathbb{R}_{\geq0}^{|\Omega_{V}|}$, and its support as
\begin{equation}
  \begin{split}
    \Lambda:=\mathrm{supp}(\vec{p}_{\mathrm{PEC}})
    = \left\{\mathsf{P} \in\{\mathsf{I}, \mathsf{Z}\}^{|V|}~\Big|~\vec{p}_{\mathrm{PEC}}(\mathsf{P}) > 0\right\}.
  \end{split}
\end{equation}

Based on the above, we design the ``Universal Blind Probabilistic Error Cancellation (UBPEC)'' protocol that takes $N_{c}$ rounds in total.

  The whole protocol follows the procedure in Protocol~\ref{protocol:ubpec_whole}, where the Client additionally chooses $\mathsf{P}_{i}\in\{\mathsf{I},\mathsf{Z}\}^{|V|}$ with probability $|\vec{q}(\mathsf{P})|/\left\|\vec{q}\right\|_{1}$ for the error cancellation operation of each single round $i\in[N_{c}]$.
  Taking this error cancellation operation, each single round of UBPEC then follows the procedure in Protocol~\ref{protocol:ubpec_single}, which applies $\mathsf{Z}$-flip for PEC via Client's classical processing.

With $\tilde{y}_{i}\in\{-1,1\}$ being the output of the $i$-th round, the final estimator is computed by

  \begin{equation}
    \begin{split}
      \tilde{o} = \frac{1}{N_{c}}\sum_{i=1}^{N_{c}} \left\|\vec{q}\right\|_{1}\operatorname{sign}\left(\vec{q}(\mathsf{P}_{i})\right) \tilde{y}_{i},
    \end{split}
  \end{equation}

where $\operatorname{sign}\left(\vec{q}(\mathsf{P}_{i})\right)$ is the sign of the quasi-probability $\vec{q}(\mathsf{P}_{i})$ for the $i$-th round.
The schematic illustration of a single round of UBPEC is depicted in Fig.~\ref{fig:ubpec}.

\begin{figure}
  \centering
  \includegraphics[width=0.9\linewidth]{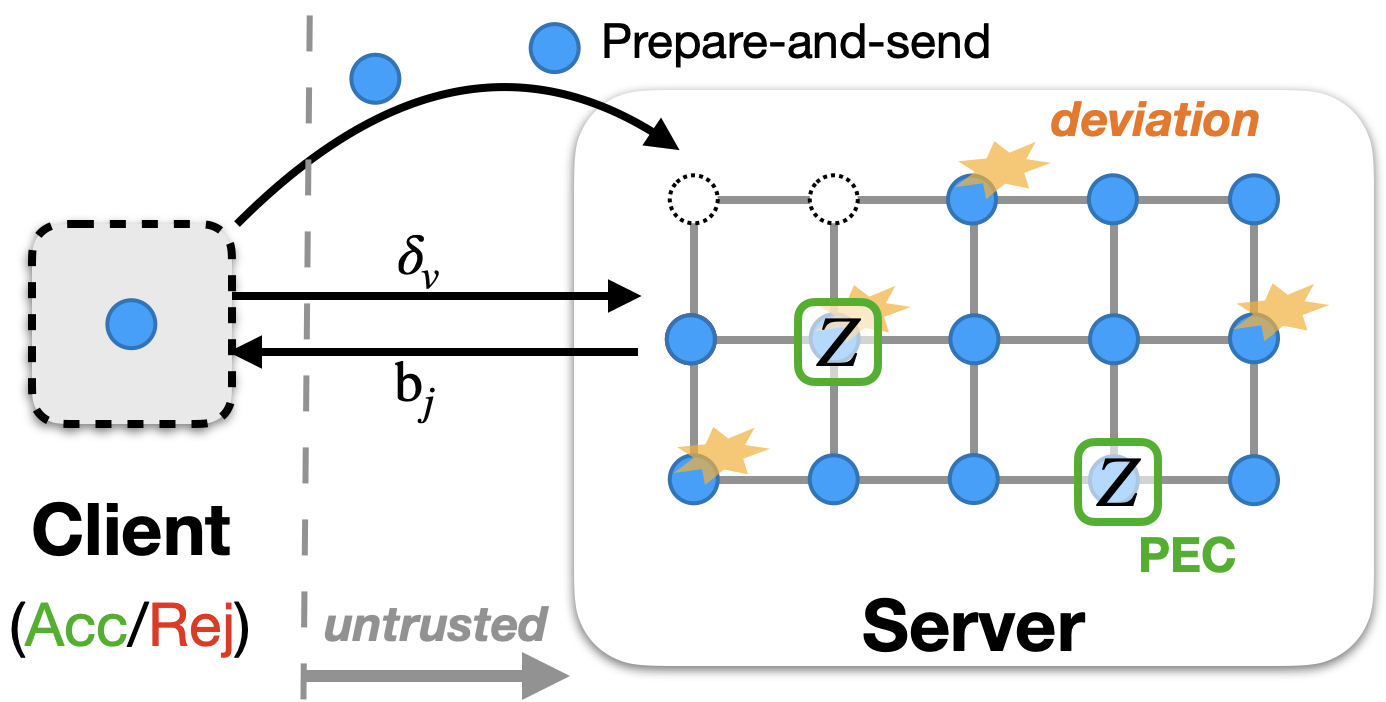}
  \caption{
    Schematic illustration of a single round of the UBPEC protocol.
    All quantum systems outside the client's trusted domain are considered untrusted, meaning deviations in this region are not assumed to follow any particular structure or noise model.
  }
  \label{fig:ubpec}
\end{figure}

We prove the security of UBPEC within the abstract cryptography framework by establishing indistinguishability between the behavior of the ideal resource attached to a simulator on its malicious interfaces and the actual protocol.
Namely, we show that for any distinguisher, UBPEC is perfectly indistinguishable from the BDPEC resource augmented by a simulator that has only access to the prescribed public leakage and the server's interface of the BDPEC.
The formal conditions and procedures for proving security are given in Definitions~\ref{definition:Statistical_Indistinguishability_of_Resources} and~\ref{definition:Construction_of_resources}, respectively, in Appendix~\ref{appendix:ac}.

To understand intuitively the security proof, one can rely on that of UBOE (Protocol~\ref{protocol:uboe_whole}).
Indeed, reordering and grouping the rounds of UBPEC for each choice of $\mathsf{P}$, one realizes that it amounts to applying successively UBOE for a computation that is altered by the corresponding $\mathsf{P}$, and combining the results using the quasi-probability distribution $\vec{q}(\mathsf{P})$.
The composability of Protocol~\ref{protocol:uboe_whole} then gives the result.
The above reordering is made possible only because all altered computations belong to the same class that is protected by blindness, i.e., the bit flips are masked perfectly by $\mathrm{r}_{v}\in\{0,1\}$ for $\mathsf{Z}$ randomization of QOTP.
This implies that the distinguisher cannot tell apart the transcripts of BDPEC together with a simulator and UBPEC.
\begin{theorem}[Security of UBPEC]\label{theorem:UBPEC}
  The UBPEC protocol perfectly constructs the BDPEC resource, leaking only public information $\left(\mathfrak{C}, G, f_{G}, N_{c}, \mathcal{E}_{\mathrm{PEC}}\right)$.
\end{theorem}

We remark that the blind insertion of Pauli operations via the Client's measurement instructions is applicable to other QEM methods, such as probabilistic error amplification (PEA)~\cite{Mari2021Extending, Kim2023Evidence}, which is widely adopted as an essential practical QEM building block on real quantum hardware nowadays~\cite{Yoshioka2025Krylov}.

\begin{resource*}[htbp]
  \caption{\raggedright Secure Delegated Probabilistic Error Cancellation (SDPEC)}
  \label{resource:sdpec}
  \begin{algorithmic}[0]
    \STATE \textbf{Public Information:}
    Public information $(\mathfrak{C}, G, f_{G}, N_{c}, N_{t}, \mathcal{E}_{\mathrm{PEC}})$; a security parameter $\epsilon_{t}>0$; and the allowed bias $\epsilon>0$.

    \STATE \textbf{Client's interface:}
    The target computation $\mathsf{C}\in \mathfrak{C}$ to produce the state $\rho$ to be measured by $O$.

    \STATE \textbf{Server's interface:}
    \begin{enumerate}
      \item The interface is filtered so that when $e=0$, the interface does not send any information nor take inputs.
      \item For $e=1$, the Resource receives a quantum state $\sigma$ and $F$, a list of instructions so that the resource produces $s \in \mathbb R \cup \{\mathsf{Abort}\}$.
    \end{enumerate}

    \STATE \textbf{Processing by the Resource:}
    \begin{enumerate}
      \item If $e=0$, it sets $\displaystyle s = \frac{1}{N_{c}} \sum_{i = 1}^{N_{c}}\left\|\vec{q}\right\|_{1}\operatorname{sign}\left(\vec{q}(\mathsf{P}_{i})\right)y_{i}$ with $y_{i}=2Y_{i}-1$ where $Y_{i} \sim \mathcal{B}(p)$, sampled from a Bernoulli distribution $\mathcal{B}(p)$ with $p = \operatorname{Tr}\left[\left(\mathsf{P}_{i}\circ\mathcal{E}_{\mathrm{PEC}}\right)\left(\rho\right) |+_{O}\rangle\langle+_{O}|\right]$, where $\mathsf{P}_{i}\in\{\mathsf{I}, \mathsf{Z}\}^{|V|}$ is sampled with probability $|\vec{q}(\mathsf{P}_{i})|/\left\|\vec{q}\right\|_{1}$.
      \item If $e=1$, it computes $s$ using the transmitted state $\sigma$ and $F$.
      \item If $s =  \mathsf{Abort}$ it forwards $o=\mathsf{Abort}$ to the Client.
      \item If $|s - \operatorname{Tr}\left[\rho O\right]| \geq \epsilon$ it sets $o = \mathsf{Abort}$ and forwards it to the Client.
      \item Otherwise it directly forwards $s$ to the Client.
    \end{enumerate}
  \end{algorithmic}
\end{resource*}

\section{Verifiable blind probabilistic error cancellation (VBPEC)\label{sec:vbpec}}

We next endow UBPEC with verifiability, that is we ensure that the Server's deviation follows closely the noise model $\mathcal{E}_{\mathrm{PEC}}$ which in turn guarantees that the PEC-ed target computation estimates $\operatorname{Tr}[\rho O]$ with little bias.

Below, we first define the ``Secure Delegated Probabilistic Error Cancellation (SDPEC)'' resource (Resource~\ref{resource:sdpec}) to formalize the ideal behavior of verifiability and blindness of PEC, and then show that it can be constructed by the proposed ``Verifiable Blind Probabilistic Error Cancellation (VBPEC)'' protocol (Protocol~\ref{protocol:vbpec}) to negligible error within the abstract cryptography framework.

\begin{figure}[htbp]
  \centering
  \includegraphics[width=\linewidth]{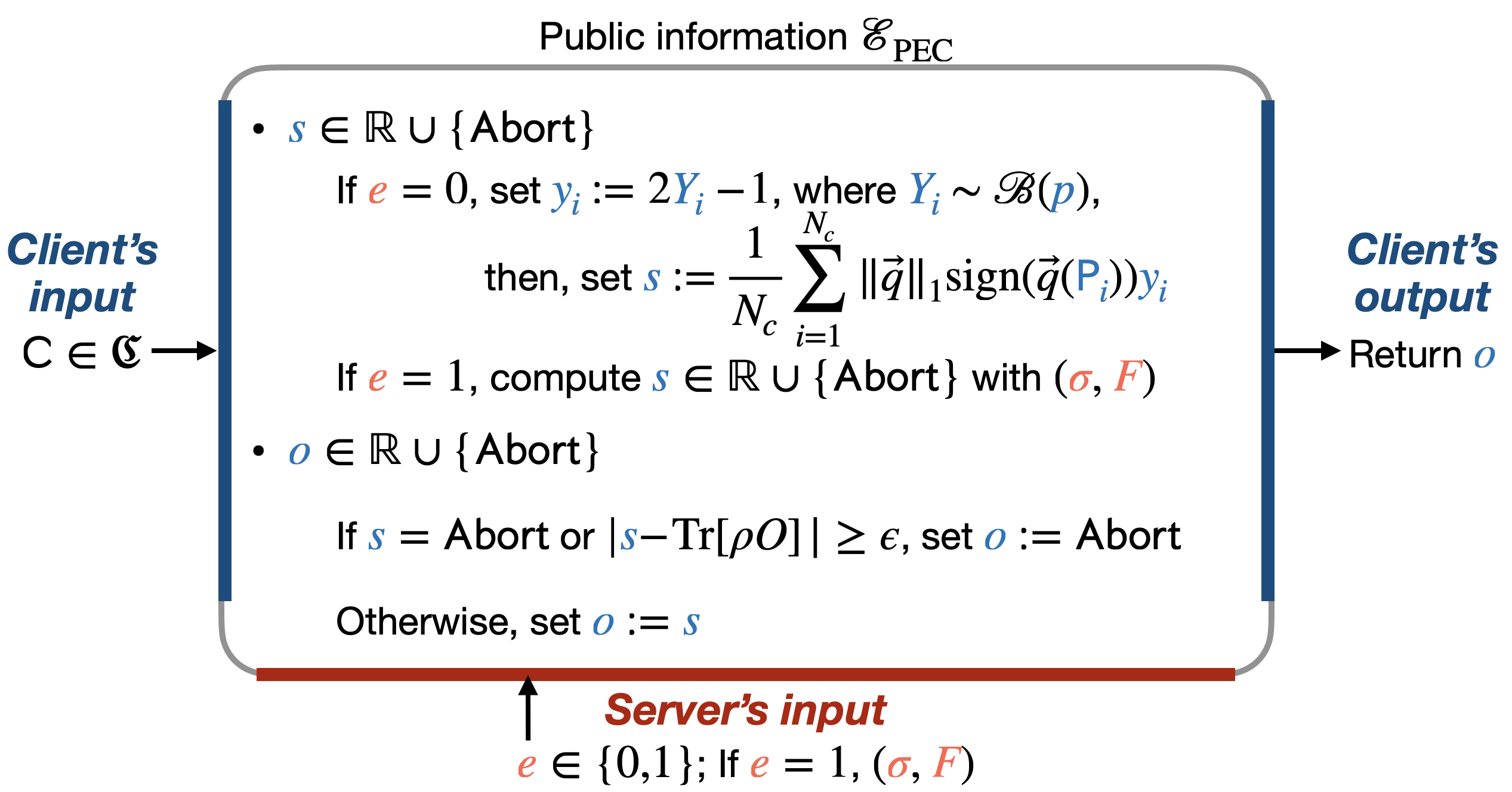}
  \caption{
    Schematic illustration of the SDPEC resource, where $p = \operatorname{Tr}\left[\left(\mathsf{P}_{i}\circ\mathcal{E}_{\mathrm{PEC}}\right)\left(\rho\right) |+_{O}\rangle\langle+_{O}|\right]$ and $\mathsf{P}_{i}\in\{\mathsf{I}, \mathsf{Z}\}^{|V|}$ is sampled with probability $|\vec{q}(\mathsf{P}_{i})|/\left\|\vec{q}\right\|_{1}$.
  }
  \label{fig:sdpec}
\end{figure}

\subsection{Secure delegated probabilistic error cancellation (SDPEC) resource\label{sec:Secure_delegated_probabilistic_error_cancellation}}

Resource~\ref{resource:sdpec} provides perfect blindness and either aborts or returns a real-valued output within error $\epsilon$ of $\operatorname{Tr}\left[\rho O\right]$.
To model this, $e\in\{0,1\}$ is a flag controlling whether the Server's interface filter is activated or not.
When $e=0$, Resource~\ref{resource:sdpec} performs the same procedure as Resource~\ref{resource:bdpec}, sampling for each round from the distribution under the deviation $\mathcal{E}_{\mathrm{PEC}}$ together with the error cancellation operation $P_{i}$.
It then takes the empirical average of the returned $\tilde y_{i})$ to produce the unbiased estimator of $\operatorname{Tr}[\rho O]$.
When the Server asks for full access ($e=1$), it receives the permitted leakage, i.e., $\mathfrak{C}$ corresponding to all the observable estimation problems that the resource can handle.
The Server also receives the parameters $N_{c}$, $N_{t}$, $\epsilon_{t}$ and $\epsilon$.
The Server is then allowed to send a deviation to be applied by the resource, which takes the form of a quantum state $\sigma$ and a classical list of quantum and classical instructions that produce either a scalar or $\mathsf{Abort}$.
If the produced scalar is within $\epsilon$ of $\operatorname{Tr}\left[\rho O\right]$, then the resource sends the scalar to the Client.
Otherwise, it sends $\mathsf{Abort}$.

\begin{protocol*}[htbp]
  \caption{\raggedright Single-round general trap pattern (Kapourniotis et al.~\cite{Kapourniotis2024Unifying})}
  \label{protocol:trap_general_single}
  \begin{algorithmic}[0]
    \STATE \textbf{Public Information:}
    Public information $(\mathfrak{C}, G, f_{G})$.

    \STATE \textbf{Protocol:}
    \begin{enumerate}
      \item The Client chooses uniformly at random a subset $S\subseteq V$ to specify the trap pattern.

      \item The Client sends qubits to the Server.
            \begin{itemize}
              \item If $v\in S$, the Client chooses $\theta_{v}\in\Theta$ at random and sends the state $|+_{\theta_{v}}\rangle$.
              \item If $v\in V\setminus S$, the Client chooses a bit $\mathrm{d}_{v} \in \{0,1\}$ uniformly at random and sends the state $|d_{v}\rangle$.
            \end{itemize}

      \item The Server applies the entangling operation $\displaystyle \mathsf{G}=\prod_{(u,v)\in E}\mathsf{CZ}_{u,v}$ to prepare the resource state.

      \item For $v \in V$, the Client sends a measurement angle $\delta_{v}\in\Theta$, the Server measures the appropriate corresponding qubit in the basis $\{|+_{\delta_{v}}\rangle, |-_{\delta_{v}}\rangle\}$, returning outcome $\mathrm{b}_{v}$ to the Client.
            The angle $\delta_{v}$ is defined as follows:
            \begin{itemize}
              \item If $v\in T_{\mathrm{even}}$: the Client chooses $\mathrm{r}_{v} \in\{0,1\}$ uniformly at random and sets $\delta_{v} = \theta_{v} + \mathrm{r}_{v} \pi$.
              \item If $v\in T_{\mathrm{odd}}$: the Client chooses $\mathrm{r}_{v} \in\{0,1\}$ uniformly at random and sets $\delta_{v} = \theta_{v} + (\mathrm{r}_{v} + 1/2) \pi$.
              \item If $v\in V\setminus S$: the Client chooses $\delta_{v}$ from $\Theta$ uniformly at random.
            \end{itemize}

      \item The Client returns $\tilde{y}=1$ if $\displaystyle \left(\bigoplus_{v \in N_{G}^{\mathrm{odd}}(S)} d_{v}\right)\oplus\left(\bigoplus_{v \in S}\mathrm{b}_{v}\oplus\mathrm{r}_{v}\right) = 0$.
            Otherwise, it returns $\tilde{y}=-1$.

    \end{enumerate}
  \end{algorithmic}
\end{protocol*}

\begin{protocol*}[htbp]
  \caption{\raggedright Verifiable Blind Probabilistic Error Cancellation (VBPEC)}
  \label{protocol:vbpec}
  \begin{algorithmic}[0]
    \STATE \textbf{Public Information:}
    Public information $(\mathfrak{C}, G, f_{G}, N_{c}, N_{t}, \mathcal{E}_{\mathrm{PEC}})$; a security parameter $\epsilon_{t}>0$; and the allowed bias $\epsilon>0$.

    \STATE \textbf{Inputs from Client:}
    The target computation $\mathsf{C}\in \mathfrak{C}$ that estimates $\operatorname{Tr}\left[\rho O\right]$.

    \STATE \textbf{Protocol:}
    \begin{enumerate}
      \item The Client randomly samples indices in $[N]$ for $N = N_{c} + N_{t}$ to indicate the locations of $N_{t}$ test rounds and $N_{c}$ computation rounds.
            Let $\mathrm{S}_{\mathsf{C}}$ be the index set of computation rounds.
            Let $\mathrm{S}_{\mathsf{T}}$ be the index set of test rounds with general trap patterns $\{S_{i}\}_{i\in\mathrm{S}_{\mathsf{T}}}$ over $V$.

      \item Each round is then sequentially delegated to the Server.
            For $i \in \mathrm{S}_{\mathsf{C}}$, the Client chooses $\mathsf{P}_{i}\in\{\mathsf{I},\mathsf{Z}\}^{|V|}$ with probability $|\vec{q}(\mathsf{P}_{i})|/\left\|\vec{q}\right\|_{1}$, and executes the single-round UBPEC protocol (Protocol~\ref{protocol:ubpec_single}).
            For $i \in \mathrm{S}_{\mathsf{T}}$, the Client executes the single-round general trap protocol (Protocol~\ref{protocol:trap_general_single}).

      \item Upon receiving and decoding the result of each round $i \in [N]$, the Client computes
            \begin{equation}\label{eq:vbpec_process_classical}
              \begin{split}
                \tilde{o}_{\mathsf{C}}
                 & = \frac{1}{N_{c}}\sum_{i\in \mathrm{S}_{\mathsf{C}}} \left\|\vec{q}\right\|_{1}\operatorname{sign}\left(\vec{q}(\mathsf{P}_{i})\right)\tilde{y}_{i}, \\
                \tilde{\Delta}
                 & = \sum_{\mathsf{P}\in\Lambda} \left|\vec{p}_{\mathrm{PEC}}(\mathsf{P}) - \tilde{p}_{\mathsf{T}}(\mathsf{P})\right|
                + \left|1 - \sum_{\mathsf{P}\in\Lambda} \tilde{p}_{\mathsf{T}}(\mathsf{P})\right|,
                \quad \text{where}\quad \tilde{p}_{\mathsf{T}}(\mathsf{P}) = \frac{1}{N_{t}}\sum_{i\in \mathrm{S}_{\mathsf{T}}} (-1)^{\langle \mathrm{set}(\mathsf{P}), S_{i}\rangle} \tilde{y}_{i},
              \end{split}
            \end{equation}
            where $\mathrm{set}(\mathsf{P}) = \{v\in V~\mid~\mathsf{P}(v)=\mathsf{Z}\}$, and $\langle S, T\rangle = |S\cap T| \pmod{2}$ for $S, T\subseteq V$.

      \item
            The Client checks $\tilde{\Delta}\leq \epsilon_{t}$.
            If satisfied, the Client returns $\tilde{o} = \tilde{o}_{\mathsf{C}}$, otherwise the Client returns $\tilde{o} = \mathsf{Abort}$.
    \end{enumerate}
  \end{algorithmic}
\end{protocol*}

This definition captures the intuitive notion of secure delegated probabilistic error cancellation: a malicious Server learns only the public information, namely the class of target computations supported by the resource.
Moreover, whenever the resource returns an accepted real-valued output, the Server can influence this output only within $\epsilon$ of the true expectation value; otherwise, the Server inevitably forces the resource to return $\mathsf{Abort}$.
The schematic illustration of the SDPEC resource is depicted in Fig.~\ref{fig:sdpec}.

\subsection{Verifiable blind probabilistic error cancellation (VBPEC) protocol}

Considering the SDPEC resource defined above, we propose the following VBPEC protocol.
Our design principle is to use test rounds to benchmark the deviation of the server and compare it to the expected noise channel $\mathcal{E}_{\mathrm{PEC}}$ used to perform the PEC.
To this end, we use the general trap patterns shown as Protocol~\ref{protocol:trap_general_single} introduced by~\cite{Kapourniotis2024Unifying}.
We then perform a classical post-processing of the general trap results to evaluate how close the Server's deviation is to the noise model $\vec{p}_{\mathrm{PEC}}$.

Let the empirically averaged server's deviation on the test rounds be $\vec{p}_{\mathsf{T}}(\mathsf{P})$.
  Indeed, $\vec{p}_{\mathsf{T}}(\mathsf{P})$ can be estimated through the inverse Walsh--Hadamard transform by post-processing the outputs from randomly sampled general trap patterns for each $\mathsf{P}\in\{\mathsf{I},\mathsf{Z}\}^{|V|}$.
More precisely, let $\{\tilde{y}_{i}\}_{i\in\mathrm{S}_{\mathsf{T}}}$ be the outputs of such random general traps.
Then, for each $\mathsf{P}\in\{\mathsf{I},\mathsf{Z}\}^{|V|}$, an unbiased estimator of $\vec{p}_{\mathsf{T}}(\mathsf{P})$ is given by (Lemma~\ref{lemma:Unbiased_estimation_of_Pauli_deviation})
\begin{equation}
  \begin{split}
    \tilde{p}_{\mathsf{T}}(\mathsf{P}) = \frac{1}{N_{t}}\sum_{i\in \mathrm{S}_{\mathsf{T}}} (-1)^{\langle \mathrm{set}(\mathsf{P}), S_{i}\rangle} \tilde{y}_{i},
  \end{split}
\end{equation}
where $\mathrm{set}(\mathsf{P}) = \{v\in V~\mid~\mathsf{P}(v)=\mathsf{Z}\}$, and $\langle S, T\rangle = |S\cap T| \pmod{2}$ for $S, T\subseteq V$.

Based on this idea, we design VBPEC as Protocol~\ref{protocol:vbpec}, which randomly interleaves $N_{c}$ computation rounds and $N_{t}$ test rounds.
In each test round, the Client samples a general trap pattern at random and records the corresponding decoded trap output.
The collection of $N_{t}$ trap outputs is then used as a common dataset for estimating $\vec{p}_{\mathsf{T}}(\mathsf{P})$ where $\mathsf{P}\in\Lambda$ is the relevant support of the expected deviation $\mathcal{E}_{\mathrm{PEC}}$.
\footnote{By construction $\Lambda$ is polynomial in $|V|$ so that the evaluation can be performed efficiently in spite of the exponential length of $\vec{p}$.}
The Client then classically post-processes this dataset by weighting the decoded trap outputs with the Walsh--Hadamard signs associated with $\mathsf{P}$.
Averaging the resulting signed outputs gives the estimator $\tilde{p}_{\mathsf{T}}(\mathsf{P})$.
Applying this post-processing for all relevant $\mathsf{P}$ yields an estimated deviation distribution, which is compared with the PEC noise model through the statistic $\tilde{\Delta}$.
In this sense, the test rounds are not accepted or rejected individually; rather, they provide statistical samples from which the model mismatch is estimated.

Reflecting the above, a schematic illustration of the VBPEC protocol is shown in Fig.~\ref{fig:vbpec}.
\begin{figure*}[htbp]
  \centering
  \includegraphics[width=0.7\linewidth]{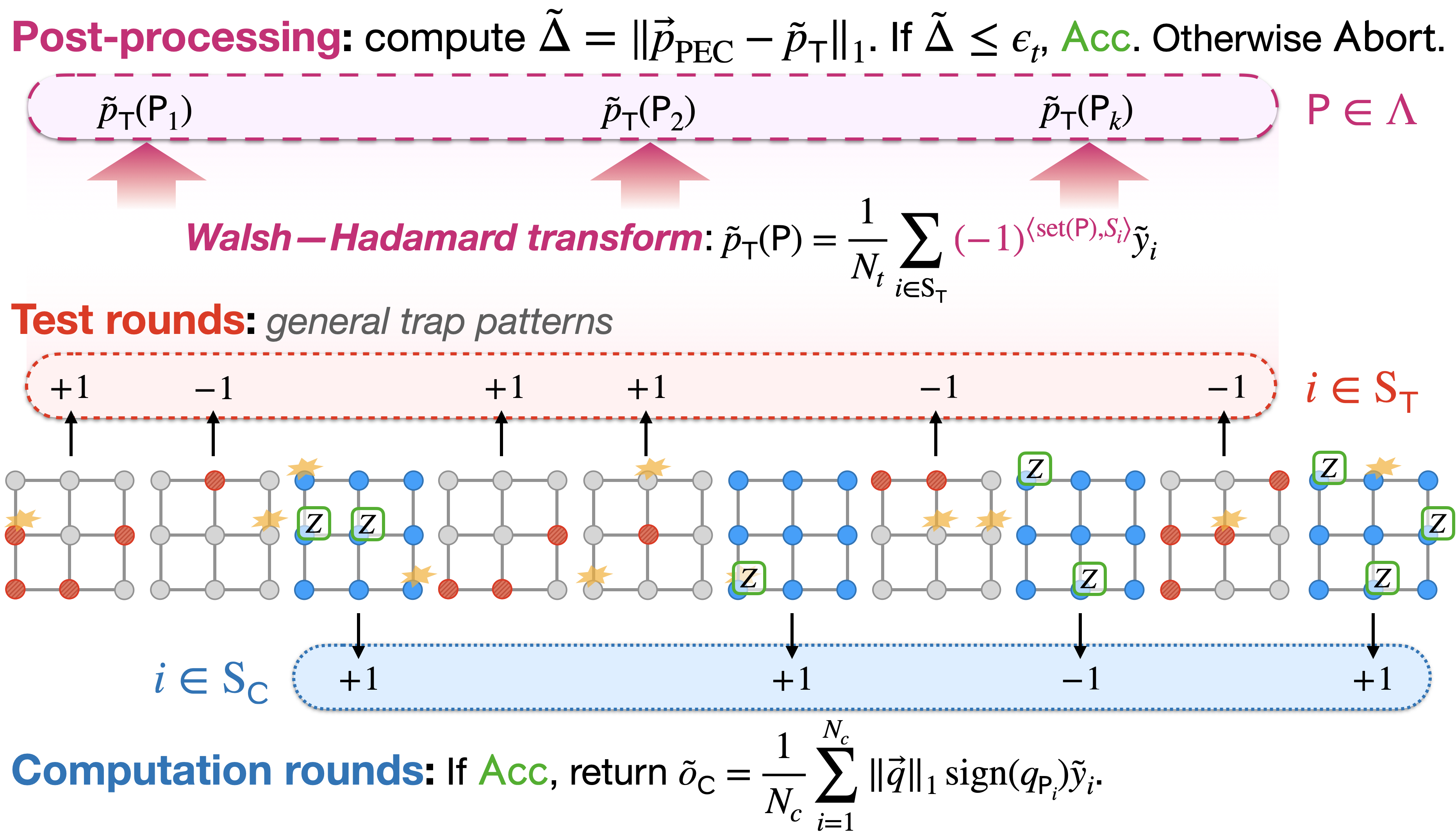}
  \caption{
    Schematic illustration of the VBPEC protocol.
    Each small graph represents one delegated MBQC instance.
    Blue vertices denote computation vertices used in computation rounds $i\in\mathsf{S}_{\mathsf{C}}$, while red vertices denote the vertices selected by the test setting $S_{i}$ in test rounds $i\in\mathsf{S}_{\mathsf{T}}$, and the grey vertices are qubits outside $S_{i}$.
    Yellow stars represent the noise or deviation occurring on the graph.
    The PEC cancellation operation is applied only on the computation rounds, shown as green boxes labelled $Z$, according to the probability decomposition $\vec{q}$.
  }
  \label{fig:vbpec}
\end{figure*}

\subsection{Constructing SDPEC with VBPEC\label{sec:Proof_of_composability}}

We can now state the main result of this paper:
\begin{theorem}[Composable Security of VBPEC]
  \label{theorem:vbpec}
  Let $\mathfrak{C}$ be a class of observable estimation problems that can be estimated using an MBQC pattern on a fixed graph $G$ with a given flow $f_{G}$.
  Let $N_{c}, N_{t}\in \mathbb{N}$, let $\epsilon, \epsilon_{t}$ be constants such that $\displaystyle 0 \leq \epsilon_{t} < \epsilon / \left\|\vec{q}\right\|_{1}$.
  Let $\mathcal{E}_{\mathrm{PEC}}$ be the noise model for PEC with $\Lambda$ whose size is polynomial in $|V|$.
  Then, the VBPEC protocol (Protocol~\ref{protocol:vbpec}), with $N_{c}$ computation rounds and $N_{t}$ test rounds, $\delta$-constructs the SDPEC resource.
  For $N_{t}\geq O(1/\epsilon_{t}^{2})$, $\delta$ is negligible in $N_{c}$ and $N_{t}$.
\end{theorem}
If the tested deviation statistics are close to the PEC reference model, then every accepted PEC estimate is close to the ideal expectation value, except with exponentially small probability.

Following abstract cryptography, to prove this theorem, we need to upper-bound the distinguishing advantage between the VBPEC protocol and the SDPEC resource in the honest (Correctness proof) and malicious (Security proof) settings (see Definition~\ref{definition:Construction_of_resources}).
To proceed with the proofs, let us first define the following conditions,
\begin{equation}
  \begin{split}
    \mathsf{Good}
     & \iff \left|\tilde{o}_{\mathsf{C}}-\operatorname{Tr}[\rho O]\right|\leq \epsilon, \\
    \mathsf{Bad}
     & \iff \neg\mathsf{Good}
    \iff \left|\tilde{o}_{\mathsf{C}}-\operatorname{Tr}[\rho O]\right|> \epsilon,       \\
    \mathsf{Acc}
     & \iff \tilde{\Delta} \leq \epsilon_{t},                                           \\
    \mathsf{Rej}
     & \iff \neg\mathsf{Acc}
    \iff \tilde{\Delta} > \epsilon_{t}.
  \end{split}
\end{equation}

We also introduce a deviation measure for two vectors $\vec{a},\vec{b}\in\mathbb{R}^{|\Omega_{V}|}$, defined by
\begin{equation}\label{eq:definition_D_Lambda}
  \begin{split}
    \mathrm{D}_{\Lambda}(\vec{a},\vec{b})
     & := \sum_{\mathsf{P}\in\Lambda}\left|\vec{a}(\mathsf{P})-\vec{b}(\mathsf{P})\right|    \\
     & \quad\quad\quad + \left|\left(1-\sum_{\mathsf{P}\in\Lambda}\vec{a}(\mathsf{P})\right)
    - \left(1-\sum_{\mathsf{P}\in\Lambda}\vec{b}(\mathsf{P})\right)\right|                   \\
     & = \sum_{\mathsf{P}\in\Lambda}\left|\vec{a}(\mathsf{P})-\vec{b}(\mathsf{P})\right|
    + \left|\sum_{\mathsf{P}\in\Lambda}(\vec{a}(\mathsf{P}) - \vec{b}(\mathsf{P}))\right|.
  \end{split}
\end{equation}
Operationally, $\mathrm{D}_{\Lambda}$ compares the probabilities of the Pauli deviations in $\Lambda$ and treats all deviations outside $\Lambda$ as one coarse-grained event.
We here remark that applying the triangle inequality on the second term of the right-hand side of Eq.~\eqref{eq:definition_D_Lambda} immediately implies
\begin{equation}\label{eq:inequality_triangle_D_Lambda}
  \begin{split}
    \mathrm{D}_{\Lambda}(\vec{a},\vec{b})
     & \leq \sum_{\mathsf{P}\in\Lambda}\left|\vec{a}(\mathsf{P})-\vec{b}(\mathsf{P})\right|
    + \sum_{\mathsf{P}\in\Lambda}\left|\vec{a}(\mathsf{P}) - \vec{b}(\mathsf{P})\right|     \\
     & = 2\sum_{\mathsf{P}\in\Lambda}\left|\vec{a}(\mathsf{P})-\vec{b}(\mathsf{P})\right|.
  \end{split}
\end{equation}

Besides, since $\vec{p}_{\mathrm{PEC}}$ is supported on $\Lambda$, the deviation of a probability distribution $\vec{a}$ from $\vec{p}_{\mathrm{PEC}}$ is equivalent to the 1-norm between $\vec{a}$ and $\vec{p}_{\mathrm{PEC}}$,
\begin{equation}\label{eq:definition_D_Lambda_against_pec}
  \begin{split}
    \mathrm{D}_{\Lambda}\left(\vec{a},\vec{p}_{\mathrm{PEC}}\right)
     & = \sum_{\mathsf{P}\in\Lambda}\left|\vec{a}(\mathsf{P})-\vec{p}_{\mathrm{PEC}}(\mathsf{P})\right|
    + \left|1-\sum_{\mathsf{P}\in\Lambda}\vec{a}(\mathsf{P})\right|                                     \\
     & = \left\|\vec{a} - \vec{p}_{\mathrm{PEC}}\right\|_{1}.
  \end{split}
\end{equation}
We also note that, comparing Eq.~\eqref{eq:vbpec_process_classical} and Eq.~\eqref{eq:definition_D_Lambda}, the test statistic can be equivalently written as $\tilde{\Delta} = \mathrm{D}_{\Lambda}\left(\tilde{p}_{\mathsf{T}},\vec{p}_{\mathrm{PEC}}\right)$.

\begin{proof}[\textbf{Proof of correctness}]

  The proof of correctness relies on the composability of the UBOE protocol.
  As apparent in Protocol~\ref{protocol:vbpec}, each round is delegated to the server using UBOE.
  Because UBOE perfectly constructs the BDOE resource, we can instead prove correctness using a hybrid protocol in which each UBOE instantiation is replaced by a call to BDOE.

  If the Server is honest, its deviation follows exactly the PEC noise model $\mathcal{E}_{\mathrm{PEC}}$.
  As a result, each outcome $y_{i}$ for the computation rounds $i\in \mathrm{S}_{\mathsf{C}}$ is processed through $y_{i}=2Y_{i}-1$, where $Y_{i}$ is sampled from the Bernoulli distribution $\mathcal{B}(p)$ with probability
  \begin{equation}
    \begin{split}
      p = \frac{1 + \operatorname{Tr}\left[\left(\mathsf{P}_{i}\circ\mathcal{E}_{\mathrm{PEC}}\right)\left(\rho\right) O\right]}{2},
    \end{split}
  \end{equation}
  where $\mathsf{P}_{i}\in\{\mathsf{I},\mathsf{Z}\}^{|V|}$ is sampled with probability $|\vec{q}(\mathsf{P}_{i})|/\left\|\vec{q}\right\|_{1}$.
  This ensures that the produced empirical average
  
    \begin{equation}
      \begin{split}
        \frac{1}{N_{c}}\sum_{i\in \mathrm{S}_{\mathsf{C}}} \left\|\vec{q}\right\|_{1}\operatorname{sign}\left(\vec{q}(\mathsf{P}_{i})\right) \tilde{y}_{i}
      \end{split}
    \end{equation}
  
  is drawn from the same probability distribution as that used to define the ideal resource.
  Hence, whenever the ideal resource and the protocol both output $\mathsf{Abort}$, their outputs coincide, and whenever they both output the estimated value, their probability distributions also coincide.
  Consequently, the only distinguishing advantage is the probability that the two setups have different decision of $\mathsf{Acc}$ or $\mathsf{Abort}$.
  Using the final output $o$ (resp.
  $\tilde{o}$) from the SDPEC resource (resp. the VBPEC protocol), this probability can be described by
  \begin{equation}
    \begin{split}
       & \operatorname{Pr}\left[(o=\mathsf{Acc})\wedge(\tilde{o}=\mathsf{Abort})\right]          \\
       & \quad + \operatorname{Pr}\left[(o=\mathsf{Abort})\wedge(\tilde{o}=\mathsf{Acc})\right],
    \end{split}
  \end{equation}
  which is upper-bounded by
  \begin{equation}
    \begin{split}
      \operatorname{Pr}\left[\tilde{o}=\mathsf{Abort}\right]
      + \operatorname{Pr}\left[o=\mathsf{Abort}\right].
    \end{split}
  \end{equation}
  We upper-bound each of these two $\mathsf{Abort}$ probabilities, respectively.

  First, the ideal resource can return $\mathsf{Abort}$ whenever the estimate is further away than $\epsilon$ from $\operatorname{Tr}\left[\rho O\right]$.
  The probability of such an event happening is upper-bounded, using Hoeffding's bound as well, by
  \begin{equation}\label{eq:bound_honest_sdpec_reject}
    \begin{split}
      \operatorname{Pr}\left[\left|o-\operatorname{Tr}\left[\rho O\right]\right|\geq\epsilon\right]
       & \leq 2\exp\left(-\frac{\epsilon^{2}}{2\left\|\vec{q}\right\|_{1}^{2}}
      N_{c}\right).
    \end{split}
  \end{equation}

  Next, to analyze the $\mathsf{Abort}$ probability by the test rounds, we introduce
  \begin{equation}\label{eq:definition_Delta}
    \begin{split}
      \Delta
       & := \mathbb{E}[\tilde{\Delta}]
      = \mathbb{E}\left[\mathrm{D}_{\Lambda}\left(\tilde{p}_{\mathsf{T}},\vec{p}_{\mathrm{PEC}}\right)\right]. \\
    \end{split}
  \end{equation}
  The $\mathsf{Abort}$ event is given by $\{ \tilde{\Delta} > \epsilon_{t} \}$.
  While $\tilde{\Delta}$ is the unbiased estimator of $\Delta$, $\Delta$ itself is not necessarily zero in general, i.e., $\Delta\geq 0$.
  Here we introduce $\epsilon_{\Delta}$ such that $0\leq \epsilon_{\Delta}\leq \epsilon_{t}$, to evaluate $\Delta$ with $\epsilon_{t}$ and $\epsilon_{\Delta}$.
  If both $|\tilde{\Delta} - \Delta|\leq \epsilon_{\Delta}$ and $\Delta \le \epsilon_{t}-\epsilon_{\Delta}$, it holds that
  \begin{equation}
    \begin{split}
      \tilde{\Delta}
       & \leq |\tilde{\Delta} - \Delta| + \Delta
      \leq \epsilon_{\Delta} + \left( \epsilon_{t}-\epsilon_{\Delta} \right)
      = \epsilon_{t}.
    \end{split}
  \end{equation}
  Hence, by contraposition,
  \begin{equation}
    \begin{split}
      \{ \tilde{\Delta} > \epsilon_{t} \}
       & \subseteq \{
      | \tilde{\Delta} - \Delta |
      > \epsilon_{\Delta}
      \} \cup \{
      \Delta > \epsilon_{t}-\epsilon_{\Delta}
      \}.
    \end{split}
  \end{equation}
  Applying the union bound, we obtain
  \begin{equation}\label{eq:honest_reject_split}
    \begin{split}
      \operatorname{Pr}[\tilde{\Delta} > \epsilon_{t}]
       & \leq \operatorname{Pr}[|\tilde{\Delta} - \Delta| > \epsilon_{\Delta}]
      + \operatorname{Pr}\left[\Delta > \epsilon_{t}-\epsilon_{\Delta}\right].
    \end{split}
  \end{equation}
  Note that $\Delta$, $\epsilon_{t}$, and $\epsilon_{\Delta}$ are all fixed parameters rather than estimators.
  We therefore upper-bound the first term on the right-hand side of Eq.~\eqref{eq:honest_reject_split}, and derive the condition under which the second term vanishes.

  First, we bound the term $\operatorname{Pr}[|\tilde{\Delta}-\Delta|>\epsilon_{\Delta}]$ using McDiarmid's inequality by evaluating the bounded difference.
  Suppose that we modify the pair $(S_{j},\tilde{y}_{j})$ to $(S_{j}^{\prime},\tilde{y}_{j}^{\prime})$ for an arbitrary test round $j\in \mathrm{S}_{\mathsf{T}}$, which changes $\tilde{p}_{\mathsf{T}}$ and $\tilde{\Delta}$ into $\tilde{p}_{\mathsf{T}}^{\prime}$ and $\tilde{\Delta}^{\prime}$, respectively.
  Since $|\tilde{y}_{i}|=1$ for $\tilde{y}_{i}\in\{-1,1\}$, changing one test round changes each component of the estimator by at most $\left|\tilde{p}_{\mathsf{T}}(\mathsf{P}) - \tilde{p}_{\mathsf{T}}^{\prime}(\mathsf{P}) \right|\leq 2/ N_{t}$ for every $\mathsf{P}\in\Lambda$.
  By the reverse triangle inequality for $\mathrm{D}_{\Lambda}$ and Eq.~\eqref{eq:inequality_triangle_D_Lambda}, we have
  \begin{equation}
    \begin{split}
      \left|\tilde{\Delta}-\tilde{\Delta}^{\prime}\right|
       & = \left|\mathrm{D}_{\Lambda} \left(\tilde{p}_{\mathsf{T}},\vec p_{\mathrm{PEC}}\right) - \mathrm{D}_{\Lambda} \left(\tilde{p}_{\mathsf{T}}^{\prime},\vec p_{\mathrm{PEC}}\right) \right| \\
       & \leq \mathrm{D}_{\Lambda} \left(\tilde{p}_{\mathsf{T}}, \tilde{p}_{\mathsf{T}}^{\prime}\right)
      \leq 2\sum_{\mathsf{P}\in\Lambda} \left| \tilde{p}_{\mathsf{T}}(\mathsf{P}) - \tilde{p}_{\mathsf{T}}^{\prime}(\mathsf{P}) \right|
      \leq \frac{4|\Lambda|}{N_{t}}.
    \end{split}
  \end{equation}

  Therefore, viewing $\tilde{\Delta}$ as a function of the $N_{t}$ independent test-round variables $\{(S_{i},\tilde{y}_{i})\}_{i\in\mathrm{S}_{\mathsf{T}}}$, we apply McDiarmid's inequality (Lemma~\ref{lemma:ineq_mcdiarmid} in Appendix~\ref{appendix:concentration_inequalities}) by defining $c_{i}=4|\Lambda|/N_{t}$ for all $i\in \mathrm{S}_{\mathsf{T}}$.
  Noticing that the sum of squares of $\{c_{i}\}_{i\in S_{\mathsf{T}}}$ is
  \begin{equation}
    \begin{split}
      \sum_{i\in \mathrm{S}_{\mathsf{T}}} c_{i}^{2}
       & = N_{t}\left(\frac{4|\Lambda|}{N_{t}}\right)^{2}
      =\frac{16|\Lambda|^{2}}{N_{t}},
    \end{split}
  \end{equation}
  McDiarmid's inequality yields
  \begin{equation}\label{eq:mcdiarmid_first_term_correctness}
    \begin{split}
      \operatorname{Pr}[|\tilde{\Delta}-\Delta| >\epsilon_{\Delta}]
       & \leq 2\exp\left(-\frac{2\epsilon_{\Delta}^{2}}{\sum_{i\in \mathrm{S}_{\mathsf{T}}} c_{i}^{2}}\right) \\
       & = 2\exp\left(-\frac{\epsilon_{\Delta}^{2}}{8|\Lambda|^{2}}
      N_{t}\right).
    \end{split}
  \end{equation}

  For the second term of Eq.~\eqref{eq:honest_reject_split}, we show that $\Delta$ can be reduced below $\epsilon_{t}-\epsilon_{\Delta}$ with an efficiently scaling choice of $N_{t}$.
  Again, since $|\tilde{y}_{i}| = 1$ for $\tilde{y}_{i}\in\{-1,1\}$, $\mathrm{Var}[\tilde{p}_{\mathsf{T}}(\mathsf{P})]\leq 1 / N_{t}$.
  Under honest execution when the reference noise model coincides with the actual noise, each $\tilde{p}_{\mathsf{T}}(\mathsf{P})$ is an unbiased estimator of the noise level of $\mathsf{P}$ in the noise model $\mathcal{E}_{\mathrm{PEC}}$, i.e., $\mathbb{E}[\tilde{p}_{\mathsf{T}}(\mathsf{P})] = \vec{p}_{\mathrm{PEC}}(\mathsf{P})$.
  Using the Cauchy-Schwarz inequality $\mathbb{E}\left[\left|X\right|\right]^{2} \leq \mathbb{E}\left[X^{2}\right]$ for expectation values, we obtain
  \begin{equation}\label{eq:mean_abs_gap_P}
    \begin{split}
      \mathbb{E}\left[\left|\vec{p}_{\mathrm{PEC}}(\mathsf{P}) - \tilde{p}_{\mathsf{T}}(\mathsf{P})\right|\right]^{2}
       & \leq \mathbb{E}\left[\left(\vec{p}_{\mathrm{PEC}}(\mathsf{P}) - \tilde{p}_{\mathsf{T}}(\mathsf{P})\right)^{2}\right]                     \\
       & = \mathbb{E}\left[\left(\mathbb{E}\left[\tilde{p}_{\mathsf{T}}(\mathsf{P})\right] - \tilde{p}_{\mathsf{T}}(\mathsf{P})\right)^{2}\right] \\
       & = \operatorname{Var}\left[\tilde{p}_{\mathsf{T}}(\mathsf{P})\right]                                                                      \\
       & \leq \frac{1}{N_{t}}.
    \end{split}
  \end{equation}
  Using Eq.~\eqref{eq:inequality_triangle_D_Lambda}, we have
  \begin{equation}\label{eq:mean_honest_delta_correctness}
    \begin{split}
      \Delta
       & = \mathbb{E}\left[\mathrm{D}_{\Lambda}\left(\tilde{p}_{\mathsf{T}},\vec{p}_{\mathrm{PEC}}\right)\right]                                       \\
       & \leq 2\sum_{\mathsf{P}\in\Lambda} \mathbb{E}\left[\left|\tilde{p}_{\mathsf{T}}(\mathsf{P}) - \vec{p}_{\mathrm{PEC}}(\mathsf{P})\right|\right]
      \leq \frac{2|\Lambda|}{\sqrt{N_{t}}}.
    \end{split}
  \end{equation}

  Therefore, so that $\operatorname{Pr}\left[\Delta > \epsilon_{t}-\epsilon_{\Delta}\right]=0$ holds, it is sufficient that $\Delta \leq\epsilon_{t}-\epsilon_{\Delta}$ holds, which is then guaranteed by the condition
  \begin{equation}\label{eq:condition_second_term_zero_general}
    \begin{split}
      N_{t} \geq \frac {4|\Lambda|^{2}} {\left(\epsilon_{t}-\epsilon_{\Delta}\right)^{2}}.
    \end{split}
  \end{equation}
  Finally, choosing $\epsilon_{\Delta}=\epsilon_{t}/2$, Eq.~\eqref{eq:condition_second_term_zero_general} becomes
  \begin{equation}\label{eq:condition_honest_pec_shots}
    \begin{split}
      N_{t} \geq \frac {16|\Lambda|^{2}} {\epsilon_{t}^{2}}.
    \end{split}
  \end{equation}
  Under this condition, the second term in Eq.~\eqref{eq:honest_reject_split} vanishes, and Eq.~\eqref{eq:mcdiarmid_first_term_correctness} gives
  \begin{equation}\label{eq:bound_honest_vbpec_reject}
    \begin{split}
      \operatorname{Pr}[\tilde{\Delta} > \epsilon_{t}]
       & \leq 2\exp\left(- \frac{\epsilon_{t}^{2}}{32|\Lambda|^{2}}
      N_{t}\right).
    \end{split}
  \end{equation}

  Combining Eq.~\eqref{eq:bound_honest_sdpec_reject} and Eq.~\eqref{eq:bound_honest_vbpec_reject}, we can finally upper-bound the distinguishability for correctness proof by
  \begin{equation}\label{eq:bound_correctness}
    \begin{split}
      2\exp\left(-\frac{\epsilon^{2}}{2\left\|\vec{q}\right\|_{1}^{2}} N_{c}\right)
      + 2\exp\left(- \frac{\epsilon_{t}^{2}}{32|\Lambda|^{2}}N_{t}\right),
    \end{split}
  \end{equation}
  conditioned on Eq.~\eqref{eq:condition_honest_pec_shots}.
  We can thus conclude that, for $\epsilon$ fixed and taking $\epsilon_{t} = \epsilon / 2$, the distinguishing advantage in the correctness under the honest Server is a negligible function of both $N_{c}$ and $N_{t}$.

\end{proof}

\begin{proof}[\textbf{Proof of security}]
  The security proof uses the composability of UBQC.
  We construct a simulator, attached to the Server interface of the SDPEC resource, that emulates the Server-side transcript of Protocol~\ref{protocol:vbpec} so that the SDPEC resource produces a transcript that is indistinguishable from the concrete protocol.
  This simulator can be easily constructed from the simulator designed to prove the security of UBQC, as described below.

  \begin{enumerate}
    \item The Simulator sets $e=1$.
    \item The Simulator prepares EPR pairs for each qubit that the Server is supposed to receive in Protocol~\ref{protocol:vbpec}.
    \item The Simulator sends all half-EPR pairs to the Server and instructs the Server to perform measurements at random angles, and retrieves the alleged measurement outcomes.
    \item The Simulator forwards the second half of the EPR pairs, the chosen angles, and the received bits to the SDPEC resource.
    \item The Simulator samples at random indices within $\left[N_{c} + N_{t}\right]$ to define the sets $\mathrm{S}_{\mathsf{C}}$, and $\mathrm{S}_{\mathsf{T}}$.
    \item The Simulator decides the trap patterns $\{S_{i}\}_{i\in\mathrm{S}_{\mathsf{T}}}$ associated with each test round in $\mathrm{S}_{\mathsf{T}}$ and passes this to the SDPEC resource.
  \end{enumerate}

  Following the security proof of UBQC, the information passed per round, together with the type of each round, is sufficient for the SDPEC resource to generate per-round measurement results following the same probability distributions as the ones obtained when running Protocol~\ref{protocol:vbpec}.
  The SDPEC resource then computes the PEC-ed estimator $o$ from computation rounds and the test output $\tilde{\Delta}$ to check whether $\tilde{\Delta}\leq \epsilon_{t}$.
  If this is not satisfied, the resource sets $s = \mathsf{Abort}$.
  Using the computed value $s$, the ideal resource then performs steps 3, 4, and 5 in the SDPEC resource, which further check whether $s$ is within $\epsilon$ of $\operatorname{Tr}\left[\rho O\right]$.

  We remark that the perfect blindness of UBQC ensures that the value $s$ computed by using the quantum state $\sigma$ provided by the simulator, i.e. the half EPR pairs, and the classical instructions, follows the same distribution as $\tilde{o}_{\mathsf{C}}$ in Protocol~\ref{protocol:vbpec}.
  The difference arises only from the additional check $|s-\operatorname{Tr}[\rho O]|>\epsilon$ by the resource: the ideal resource may reject in cases where the real protocol accepts.

  This means that the distinguishing advantage is equal to the total variation distance between the probability distributions of $s$ and $o$, or equivalently between $o$ and $\tilde{o}$.
  Because, conditioned on $o, \tilde{o} \in \mathbb R$ or $o, \tilde{o} = \mathsf{Abort}$, the distributions of $o$ and $\tilde{o}$ are the same, the total variation distance reduces to the difference in the $\mathsf{Abort}$ probability:
  \begin{equation}\label{eq:distinguishing_advantage_vboe}
    \begin{split}
       & \left|\operatorname{Pr}\left[o = \mathsf{Abort}] - \operatorname{Pr}[\tilde{o} = \mathsf{Abort}\right]\right|                            \\
       & \quad = \operatorname{Pr}\left[\tilde{o} \neq \mathsf{Abort} \wedge |\tilde{o} -\operatorname{Tr}\left[\rho O\right]| > \epsilon \right] \\
       & \quad = \operatorname{Pr}\left[\mathsf{Acc} \wedge \mathsf{Bad} \right].
    \end{split}
  \end{equation}

  We here analyze the averaged computation and test statistics directly.
  Under QOTP and blindness, the Server's arbitrary attack can be represented as a convex mixture of global Pauli deviations across all rounds, while the Server remains ignorant of the random partition into computation and test rounds.
  It is therefore sufficient to condition on an arbitrary fixed component of this mixture.
  Namely, we condition on Pauli deviations $\left(\mathsf{P}_{1},\ldots,\mathsf{P}_{N}\right)$, where the sequence may be chosen adversarially and need not be identical across rounds.
  The remaining randomness comes from the Client's random partition, trap-pattern choices, PEC sampling, and measurement outcomes.

  Conditioned on such arbitrary sequence of  Pauli deviations $(\mathsf{P}_{1},\ldots,\mathsf{P}_{N})$, and for each realized random partition $[N]=\mathrm{S}_{\mathsf{C}}\sqcup\mathrm{S}_{\mathsf{T}}$, we formally define the empirical distributions of the deviations on the computation and test rounds as
  \begin{equation}\label{eq:define_p_c_p_t_security}
    \begin{split}
      \vec{p}_{\mathsf{C}}
      = \frac{1}{N_{c}}\sum_{i\in\mathrm{S}_{\mathsf{C}}}\mathbf{e}_{\mathsf{P}_{i}},
      \quad
      \vec{p}_{\mathsf{T}}
      = \frac{1}{N_{t}}\sum_{i\in\mathrm{S}_{\mathsf{T}}}\mathbf{e}_{\mathsf{P}_{i}}.
    \end{split}
  \end{equation}
  Here, $\mathbf{e}_{\mathsf{P}_{i}}$ denotes the unit vector supported on $\mathsf{P}_{i}$.
  Later on, we show that the gap between $\vec{p}_{\mathsf{C}}$ and $\vec{p}_{\mathsf{T}}$ becomes exponentially small as $N_{c}$ and $N_{t}$ scale up, thanks to randomly mixing the locations of test and computation rounds.

  Using $\vec{p}_{\mathsf{T}}$, we also introduce
  \begin{equation}\label{eq:Delta_T}
    \begin{split}
      \Delta_{\mathsf{T}}
       & = \mathbb{E}[\tilde{\Delta}\mid \vec{p}_{\mathsf{T}}]
      = \mathbb{E}[\mathrm{D}_{\Lambda}(\tilde{p}_{\mathsf{T}}, \vec{p}_{\mathrm{PEC}}) \mid \vec{p}_{\mathsf{T}}],
    \end{split}
  \end{equation}
  where the expectation is over the randomness of the test-round trap patterns and their decoded outputs, conditioned on the averaged empirical test-round deviation $\vec{p}_{\mathsf{T}}$.
  Similarly, we define $\mu_{\mathsf{C}} = \mathbb{E}\left[\tilde{o}_{\mathsf{C}}\mid \vec{p}_{\mathsf{C}}\right]$, where the expectation is over the PEC sampling and measurement outcomes in the computation rounds, conditioned on $\vec{p}_{\mathsf{C}}$.
  .

  For the bias of the output $\tilde{o}_{\mathsf{C}}$ processed by computation rounds, the triangle inequality yields
  \begin{equation}\label{eq:ineq_triangle_computation}
    \begin{split}
      \left|\tilde{o}_{\mathsf{C}} - \operatorname{Tr}[\rho O]\right|
       & \leq \left|\tilde{o}_{\mathsf{C}} - \mu_{\mathsf{C}}\right|
      + \left|\mu_{\mathsf{C}} - \operatorname{Tr}[\rho O]\right|.
    \end{split}
  \end{equation}
  For the second term on the right-hand side, the inequality in Lemma~\ref{lemma:link_bias_test} gives
    \begin{equation}\label{eq:bias_by_p_c_security}
      \begin{split}
        \left|\mu_{\mathsf{C}} - \operatorname{Tr}[\rho O]\right|
         & \leq \left\|\vec{q}\right\|_{1} \left(\Delta_{\mathsf{T}}
        + \mathrm{D}_{\Lambda}\left(\vec{p}_{\mathsf{C}},\vec{p}_{\mathsf{T}}\right)\right).
      \end{split}
    \end{equation}

  We next relate the acceptance condition to the computation bias.
  Let $\gamma_{c},\gamma_{t},\gamma_{s}>0$.
  Suppose that the following three fluctuation bounds hold:
  \begin{equation}\label{eq:conditions_gamma_security_corrected}
    \begin{split}
      \left|\tilde{o}_{\mathsf{C}}-\mu_{\mathsf{C}}\right|
       & \leq \gamma_{c}, \\
      \left|\tilde{\Delta}-\Delta_{\mathsf{T}}\right|
       & \leq \gamma_{t}, \\
      \mathrm{D}_{\Lambda}\left(\vec{p}_{\mathsf{C}},\vec{p}_{\mathsf{T}}\right)
       & \leq \gamma_{s}.
    \end{split}
  \end{equation}
  When the event $\mathsf{Acc}$ occurs, we have $\tilde{\Delta}\leq\epsilon_{t}$.
  Together with the second inequality in Eq.~\eqref{eq:conditions_gamma_security_corrected}, this implies
  \begin{equation}\label{eq:delta_t_bound_from_acceptance_security}
    \begin{split}
      \Delta_{\mathsf{T}}
       & \leq \tilde{\Delta}+|\tilde{\Delta}-\Delta_{\mathsf{T}}|
      \leq \epsilon_{t}+\gamma_{t}.
    \end{split}
  \end{equation}
  Combining Eqs.~\eqref{eq:ineq_triangle_computation}--\eqref{eq:delta_t_bound_from_acceptance_security}, we obtain
  \begin{equation}\label{eq:good_from_acceptance_security_corrected}
    \begin{split}
      \left|\tilde{o}_{\mathsf{C}}-\operatorname{Tr}[\rho O]\right|
       & \leq \gamma_{c}
      + \left\|\vec{q}\right\|_{1}\left(\epsilon_{t}+\gamma_{t}+\gamma_{s}\right).
    \end{split}
  \end{equation}
  Therefore, if $\gamma_{c} + \left\|\vec{q}\right\|_{1}\left(\epsilon_{t}+\gamma_{t}+\gamma_{s}\right)\leq \epsilon$, then $\mathsf{Acc}$ together with the three bounds in Eq.~\eqref{eq:conditions_gamma_security_corrected} implies $\mathsf{Good}$.
  By contraposition, we have
  \begin{equation}\label{eq:acc_bad_subset_three_fluctuations}
    \begin{split}
      \{\mathsf{Acc}\wedge\mathsf{Bad}\}
       & \subseteq \left\{\left|\tilde{o}_{\mathsf{C}}-\mu_{\mathsf{C}}\right|>\gamma_{c}\right\}                         \\
       & \quad \cup \{|\tilde{\Delta}-\Delta_{\mathsf{T}}|>\gamma_{t}\}                                                   \\
       & \quad \cup \left\{\mathrm{D}_{\Lambda}\left(\vec{p}_{\mathsf{C}},\vec{p}_{\mathsf{T}}\right)>\gamma_{s}\right\}.
    \end{split}
  \end{equation}
  Applying the union bound gives
  \begin{equation}\label{eq:union_bound_security_corrected}
    \begin{split}
      \operatorname{Pr}[\mathsf{Acc}\wedge\mathsf{Bad}]
       & \leq \operatorname{Pr}\left[\left|\tilde{o}_{\mathsf{C}}-\mu_{\mathsf{C}}\right|>\gamma_{c}\right]                           \\
       & \quad + \operatorname{Pr}[|\tilde{\Delta}-\Delta_{\mathsf{T}}|>\gamma_{t}]                                                   \\
       & \quad + \operatorname{Pr}\left[\mathrm{D}_{\Lambda}\left(\vec{p}_{\mathsf{C}},\vec{p}_{\mathsf{T}}\right)>\gamma_{s}\right].
    \end{split}
  \end{equation}

  We now bound the three terms on the right-hand side separately.
  First, conditioned on the fixed Pauli sequence and on $\mathrm{S}_{\mathsf{C}}$, each PEC-weighted computation-round output is bounded in $\left[-\left\|\vec{q}\right\|_{1},\left\|\vec{q}\right\|_{1}\right]$.
  Applying Hoeffding's inequality to the computation statistic yields
  \begin{equation}\label{eq:hoeffding_comp_security_corrected}
    \begin{split}
      \operatorname{Pr}\left[\left|\tilde{o}_{\mathsf{C}}-\mu_{\mathsf{C}}\right|>\gamma_{c}\right]
       & \leq 2\exp\left(-\frac{\gamma_{c}^{2}}{2\left\|\vec{q}\right\|_{1}^{2}}
      N_{c}\right).
    \end{split}
  \end{equation}
  The bound is uniform over the conditioned Pauli sequence and partition, and therefore also holds unconditionally.

  Second, conditioned on the fixed Pauli sequence and on $\mathrm{S}_{\mathsf{T}}$, the same bounded-difference argument as in the correctness proof applies to $\tilde{\Delta}$.
  Changing a single test-round pair changes $\tilde{\Delta}$ by at most $4|\Lambda|/N_{t}$.
  Hence, McDiarmid's inequality gives
  \begin{equation}\label{eq:mcdiarmid_test_security_corrected}
    \begin{split}
      \operatorname{Pr}\left[|\tilde{\Delta}-\Delta_{\mathsf{T}}|>\gamma_{t}\right]
       & \leq 2\exp\left(-\frac{\gamma_{t}^{2}}{8|\Lambda|^{2}}
      N_{t}\right).
    \end{split}
  \end{equation}

  Third, we bound the mismatch between the averaged deviations observed in the computation and test rounds for each deviation $\mathsf{P}\in\Lambda$.
  To begin with, we define the total averaged frequency $\bar{X}^{(\mathsf{P})}$ for each $\mathsf{P}\in\Lambda$ and each round $i\in[N]$ by
  \begin{equation}
    \begin{split}
      \bar{X}^{(\mathsf{P})}
       & := \frac{1}{N}\sum_{i=1}^{N}X_{i}^{(\mathsf{P})},\quad \text{where}~~
      X_{i}^{(\mathsf{P})}
      := \begin{cases}
           1 & \text{if} ~~ \mathsf{P}_{i}=\mathsf{P}, \\
           0 & \text{otherwise}.                       \\
         \end{cases}
    \end{split}
  \end{equation}
  Using $X_{i}^{(\mathsf{P})}$, each element of the empirical deviation distributions $\vec{p}_{\mathsf{C}}$ and $\vec{p}_{\mathsf{T}}$ can be expressed as
  \begin{equation}
    \begin{split}
      \vec{p}_{\mathsf{C}}(\mathsf{P})
      = \frac{1}{N_{c}} \sum_{i\in\mathrm{S}_{\mathsf{C}}}
      X_{i}^{(\mathsf{P})}, \quad \vec{p}_{\mathsf{T}}(\mathsf{P}) = \frac{1}{N_{t}} \sum_{i\in\mathrm{S}_{\mathsf{T}}} X_{i}^{(\mathsf{P})}.
    \end{split}
  \end{equation}
  Therefore, the gap between $\vec{p}_{\mathsf{C}}$ and $\vec{p}_{\mathsf{T}}$ for each $\mathsf{P}$ becomes
  \begin{equation}
    \begin{split}
      \vec{p}_{\mathsf{C}}(\mathsf{P}) - \vec{p}_{\mathsf{T}}(\mathsf{P})
       & = \frac{N}{N_{c}}\left(\bar{X}^{(\mathsf{P})}-\vec{p}_{\mathsf{T}}(\mathsf{P})\right).
    \end{split}
  \end{equation}
  This implies that the gap between $\vec{p}_{\mathsf{C}}$ and $\vec{p}_{\mathsf{T}}$ is dependent on how distant the test sample average is from the total averaged frequency.
  Since $\vec{p}_{\mathsf{T}}(\mathsf{P})$ is the average of $N_{t}$ samples drawn uniformly without replacement from the total $N$ empirical samples, applying Serfling's inequality to the deviation of $\vec{p}_{\mathsf{T}}(\mathsf{P})$ from $\bar{X}^{(\mathsf{P})}$ yields
  \begin{equation}\label{eq:serfling_single_pauli_security}
    \begin{split}
       & \operatorname{Pr}\left[\left|\vec{p}_{\mathsf{C}}(\mathsf{P})-\vec{p}_{\mathsf{T}}(\mathsf{P})\right|>\frac{\gamma_{s}}{2|\Lambda|}\right] \\
       & \quad= \operatorname{Pr}\left[\left|\vec{p}_{\mathsf{T}}(\mathsf{P}) - \bar{X}^{(\mathsf{P})}\right|>\frac{\gamma_{s}
      N_{c}}{2N|\Lambda|}\right]                                                                                                                    \\ &\quad\leq 2\exp\left(-\frac{N_{t} N_{c}^{2}}{N\left(N_{c}+1\right)}\frac{\gamma_{s}^{2}}{4|\Lambda|^{2}}\right) \\  & \quad\leq 2\exp\left(-\frac{1}{8} \frac{N_{c} N_{t}}{N}\frac{\gamma_{s}^{2}}{|\Lambda|^{2}}\right).
      \\
    \end{split}
  \end{equation}
  Since $\displaystyle \mathrm{D}_{\Lambda}\left(\vec{p}_{\mathsf{C}},\vec{p}_{\mathsf{T}}\right) \leq 2\sum_{\mathsf{P}\in\Lambda}\left|\vec{p}_{\mathsf{C}}(\mathsf{P})-\vec{p}_{\mathsf{T}}(\mathsf{P})\right|$, the union bound over $\mathsf{P}\in\Lambda$ gives
  \begin{equation}\label{eq:serfling_split_security_corrected}
    \begin{split}
      \operatorname{Pr}\left[\mathrm{D}_{\Lambda}\left(\vec{p}_{\mathsf{C}},\vec{p}_{\mathsf{T}}\right)>\gamma_{s}\right]
       & \leq \sum_{\mathsf{P}\in\Lambda}\operatorname{Pr}\left[\left|\vec{p}_{\mathsf{C}}(\mathsf{P})-\vec{p}_{\mathsf{T}}(\mathsf{P})\right|>\frac{\gamma_{s}}{2|\Lambda|}\right] \\
       & \leq 2|\Lambda| \exp\left(-\frac{1}{8} \frac{N_{c}
        N_{t}}{N}\frac{\gamma_{s}^{2}}{|\Lambda|^{2}}\right).
    \end{split}
  \end{equation}

  Combining Eqs.~\eqref{eq:union_bound_security_corrected}--\eqref{eq:serfling_split_security_corrected} and choosing parameters
  \begin{equation}\label{eq:gamma_choice_security_corrected}
    \begin{split}
      \gamma_{c}
       & = \frac{\epsilon}{4},
      \quad
      \epsilon_{t}
      = \gamma_{t}
      = \gamma_{s}
      = \frac{\epsilon}{4\left\|\vec{q}\right\|_{1}},
    \end{split}
  \end{equation}
  so that they satisfy $\gamma_{c} + \left\|\vec{q}\right\|_{1}\left(\epsilon_{t}+\gamma_{t}+\gamma_{s}\right)\leq \epsilon$, we finally obtain
  \begin{equation}\label{eq:final_bound_acc_bad_security_corrected}
    \begin{split}
      \operatorname{Pr}[\mathsf{Acc}\wedge\mathsf{Bad}]
       & \leq 2\exp\left(-\frac{\epsilon^{2}}{32\left\|\vec{q}\right\|_{1}^{2}}
      N_{c}\right)                                                              \\ &\quad + 2\exp\left(-\frac{\epsilon^{2}}{128\left\|\vec{q}\right\|_{1}^{2}|\Lambda|^{2}}N_{t}\right) \\ &\quad + 2|\Lambda|\exp\left(-\frac{\epsilon^{2}}{128\left\|\vec{q}\right\|_{1}^{2}|\Lambda|^{2}}\frac{N_{c}N_{t}}{N}\right).
    \end{split}
  \end{equation}
  For a polynomially bounded $|\Lambda|$, this upper bound is negligible in $N_{c}$ and $N_{t}$ whenever $N_{t}/|\Lambda|^{2}$ grows super-logarithmically.

\end{proof}

\section{Noise-robustness of VBPEC\label{sec:robustness}}

Having proved the security of VBPEC, we now analyze its noise robustness.
This section clarifies two distinct notions of robustness.
VBPEC retains passive structural robustness introduced in \cite{Leichtle2021Verifying} and exemplified for observable estimation tasks in~\cite{Yang2025Verifiable}.
For instance, when applied to VBOE, such structural robustness allows small deviations to be tolerated as long as no more than a constant fraction (eg.
less than 1/4) of test rounds  are failed.
In such case, the protocol still accepts and produces a correct estimate with high probability.
In the case of VBPEC, similar deviations would make PEC no longer exact, yet the protocol would still accept and produce an accurate estimate with high probability.
This is because the mismatch between the actual deviation and expected noise $D_{\mathrm{mis}}=\mathrm{D}_{\Lambda}(\vec{p},\vec{p}_{\mathrm{PEC}})$ remains sufficiently small.

However, this robustness is passive, since the protocol does not suppress the noise affecting the computation rounds.
VBPEC provides a stronger, active form of robustness in the noise-matched regime.
Here, we call an honest execution noise-matched if the actual noise affecting the computation rounds coincides with the reference noise model used for PEC.
When the Server's honest deviation follows the PEC noise model $\mathcal{E}_{\mathrm{PEC}}$, VBPEC can perfectly invert this noise in the computation rounds.
Thus, honest matched noise is transformed by PEC from a source of $\mathsf{Abort}$ into a sampling overhead that is controlled by $\left\|\vec{q}\right\|_{1}$.

\subsection{Robustness under honest mismatched-noise }

We analyze the robustness of VBPEC when the Server is honest, but the actual noise does not exactly match the noise model used for PEC.
We assume that each round is independently affected by the same actual stochastic Pauli channel with averaged deviation distribution $\vec{p} \in \mathbb{R}_{\geq 0}^{|\Omega_{V}|}$.
We quantify the honest noise mismatch by
\begin{equation}
  \begin{split}
    D_{\mathrm{mis}} := \left\|\vec{p} - \vec{p}_{\mathrm{PEC}}\right\|_{1} = \mathrm{D}_{\Lambda}(\vec{p}, \vec{p}_{\mathrm{PEC}}).
  \end{split}
\end{equation}
Here, $D_{\mathrm{mis}}$ measures the mismatch between the actual deviation and the noise model.

The detailed concentration analysis is given in Lemma~\ref{lemma:noise_robustness_honest_mismatch_vbpec} in Appendix~\ref{appendix:proof_noise-robustness}.
Setting $\epsilon_{t} = \epsilon / 4$ in Lemma~\ref{lemma:noise_robustness_honest_mismatch_vbpec}, one can make the probability $\operatorname{Pr}\left[\mathsf{Acc}\wedge\mathsf{Good}\right]$ exponentially close to one if
\begin{equation}
  \begin{split}
    D_{\mathrm{mis}} < \min\left\{\frac{\epsilon}{\left\|\vec{q}\right\|_{1}}, \frac{\epsilon}{4}\right\}, \quad N_{t} \geq \frac{16|\Lambda|^{2}}{(\epsilon/4 - D_{\mathrm{mis}})^{2}}.
  \end{split}
\end{equation}

Thus, in such regime, VBPEC would still accept the computation with high probability and, because security holds, produce a result that is within $\epsilon$ of the correct value as requested by the idea resource the protocol implements.
In short, VBPEC works on imperfectly characterized devices.

\subsection{Robustness under honest matched-noise}

We now evaluate the probability that VBPEC accepts the correct outcome when the Server's deviation follows exactly the PEC noise model $\mathcal{E}_{\mathrm{PEC}}$.
In this setup, the VBPEC estimator is unbiased for the ideal expectation value, and the acceptance of test rounds is affected only by finite-shot fluctuations.
Choosing $\epsilon_{t} = \epsilon / (4\left\|\vec{q}\right\|_{1})$ in Eq.~\eqref{eq:condition_honest_pec_shots} and Eq.~\eqref{eq:bound_honest_vbpec_reject}, we obtain
\begin{equation}
  \begin{split}
    \operatorname{Pr}[\mathsf{Acc}\wedge\mathsf{Good}]
     & = \operatorname{Pr}[\mathsf{Acc}] - \operatorname{Pr}[\mathsf{Acc}\wedge\mathsf{Bad}]                                            \\
     & = 1 - \operatorname{Pr}[\mathsf{Rej}] - \operatorname{Pr}[\mathsf{Acc}\wedge\mathsf{Bad}]                                        \\
     & \geq 1 - 2\exp\left(- \frac{\epsilon^{2}}{512\left\|\vec{q}\right\|_{1}^{2}|\Lambda|^{2}}N_{t}\right)                            \\
     & \quad\quad - 2\exp\left(- \frac{\epsilon^{2}}{32\left\|\vec{q}\right\|_{1}^{2}}N_{c}\right)                                      \\
     & \quad\quad - 2\exp\left(- \frac{\epsilon^{2}}{128\left\|\vec{q}\right\|_{1}^{2}|\Lambda|^{2}}N_{t}\right)                        \\
     & \quad\quad - 2|\Lambda|\exp\left(-\frac{\epsilon^{2}}{128\left\|\vec{q}\right\|_{1}^{2}|\Lambda|^{2}}\frac{N_{c}N_{t}}{N}\right)
  \end{split}
\end{equation}
conditioned on
\begin{equation}
  \begin{split}
    N_{t} \geq \frac {256|\Lambda|^{2}\left\|\vec{q}\right\|_{1}^{2}} {\epsilon^{2}},
  \end{split}
\end{equation}
which is obtained by assigning $\epsilon_{t} = \epsilon / (4\left\|\vec{q}\right\|_{1})$ in Eq.~\eqref{eq:condition_honest_pec_shots}.
This implies that, by taking a sufficient shot count $N = N_{c} + N_{t}$, the probability $\operatorname{Pr}[\mathsf{Acc}\wedge\mathsf{Good}]$ by VBPEC can be made exponentially close to one.

From Corollary~\ref{corollary:noise-robustness_honest_VBOE}, the VBOE protocol requires the condition $2p<\epsilon_{t} = \epsilon/4$, or equivalently $p<\epsilon/8$, for the non-identity probability $p$ of the actual noise channel to make $\operatorname{Pr}[\mathsf{Acc}\wedge\mathsf{Good}]$ exponentially close to one.
When $2p>\epsilon/4$, VBOE accepts the outcome with negligibly small probability, as shown in Lemma~\ref{lemma:probability_acceptance_honest_VBOE} in Appendix~\ref{appendix:proof_noise-robustness}.

To illustrate the impact of this robustness ability, one can use an extreme example where the noise only affects the output qubit but with very high probability.
That is $\mathsf{I}$ would be applied with probability $1/2 + 1/100$ whereas $\mathsf{Z}$ would be applied with probability $1/2 - 1/100$.
Clearly this noise can be well error mitigated using PEC but passive robustness such as the one exhibited by VBOE would be inapplicable as the number of failed test rounds under VBOE would exceed the 25\% maximum value after which any protocol would always reject (see~\cite{Kapourniotis2024Unifying} for the derivation of this 25\% bound).
Here, on the contrary, with an appropriate choice of parameters, the VBPEC would accept with high probability and provide the correct result.

\section{Discussion\label{sec:discussion}}

\begin{table*}[htbp]
  \centering
  \renewcommand{\arraystretch}{2.0}
  \begin{tabular}{ l|c|c|c }
    \hline\hline
                     & RVBQC~\cite{Leichtle2021Verifying}   & VBOE~\cite{Yang2025Verifiable}       & \textbf{VBPEC (this work)}           \\
    \hline
    Computation      & (randomised) discrete output         & \textbf{continuous-valued estimator} & \textbf{continuous-valued estimator} \\
    \hline
    Test             & count of failed deterministic rounds & count of failed deterministic rounds & \textbf{continuous-valued estimator} \\
    \hline\hline
    Noise-robustness & inherent robustness only             & inherent robustness only             & \makecell{inherent robustness,       \\ \textbf{active error suppression}} \\
    \hline\hline
  \end{tabular}
  \caption{
    Comparison of the round outputs and noise-robustness mechanisms in RVBQC, VBOE, and VBPEC.
    Unlike RVBQC and VBOE, whose test rounds yield deterministic pass/fail outcomes, VBPEC uses continuous-valued test estimators and therefore requires a concentration-based security proof that controls deviations from their mean values.
    Adding PEC allows the VBPEC protocol to incorporate active error suppression in addition to inherent robustness.
  }
  \label{tab:comparison_round_output_types}
\end{table*}

We have introduced Verifiable Blind Probabilistic Error Cancellation (VBPEC), the first protocol that brings QEM within the scope of composable cryptographic security.
Our main result, Theorem~\ref{theorem:vbpec}, shows that VBPEC constructs the ideal SDPEC resource with exponentially negligible error under polynomial-time execution.
It thereby provides an end-to-end correctness guarantee for error-suppressed delegated quantum computation, closing a longstanding gap between cryptographic verification and the practical QEM techniques needed for near-term and pre-fault-tolerant devices.

Moreover, VBPEC enables active error suppression using QEM within the cryptographic verification protocol.
This goes beyond its inherent noise robustness, which is limited by the acceptance threshold relative to the physical noise level.
It thereby makes the verification protocol substantially more practical by increasing the acceptance probability without introducing additional quantum-space overhead or compromising the underlying security guarantee.

\subsection{Advantages in QEM}

Our first contribution is to bridge QEM and cryptographic verification for the first time, identifying their technical compatibility and establishing mutually beneficial connections.
QEM techniques are usually designed under modelling, calibration, and hardware-specific assumptions, making their end-to-end execution difficult to certify.
This issue is particularly acute for PEC, whose accuracy depends directly on the consistency between the reference noise model and the actual device noise.
VBPEC addresses this problem by formalising PEC as an ideal cryptographic resource and proving that it can be constructed securely from a completely untrusted quantum server.

By exploiting the MBQC structure, VBPEC also provides several advantages for implementing PEC itself.
In MBQC, PEC corrections are realized as classical Pauli-frame updates, so cancellation does not require additional noisy quantum gates.
Additionally, although QEM methods can be designed for individual physical operations in MBQC~\cite{Hartung2024Real-time, Koh2026Readout}, VBPEC avoids the need to treat such components separately.
The QOTP randomization inherent in VBQC-type protocols effectively Pauli-twirls the server’s behavior, reducing arbitrary deviations, from the client’s perspective, to stochastic Pauli deviations.
Hence, as long as the client’s preparation noise is independent of the QOTP parameters, deviations arising from gates, measurements, or other server-side operations can be reliably mitigated by PEC within the MBQC framework.

\subsection{Advantages in verification protocols}

Our second contribution is to extend the post-processing of general trap patterns, revealing that their trap-output statistics encode the Server's effective deviations through a Walsh--Hadamard transform.
This enables the Client to benchmark the Server's harmful deviations, rather than merely certifying that the deviation level remains below a prescribed acceptance threshold.
In other words, while the general trap construction of~\cite{Kapourniotis2024Unifying} can be viewed as a family of stabilizer tests, our observation broadens its role into stabilizer benchmarking.
This perspective also clarifies why naively including PEC in the test rounds themselves is not a valid shortcut: as shown in Appendix~\ref{appendix:PEC_on_test}, doing so would undermine the security structure on which the verification proof relies.

VBPEC also contributes a new proof technique for trap-based verification by introducing statistical test criteria in place of deterministic pass/fail tests.
As shown in Table~\ref{tab:comparison_round_output_types}, both computation and test rounds produce statistical estimators rather than deterministic outputs.
Since individual test rounds in VBPEC no longer provide pass/fail certificates, the usual security proof strategy based on counting failed tests no longer applies.
Instead, our fluctuation-based analysis simultaneously controls deviations of the mitigated computation estimator and of the trap-based noise-model estimators.
It then translates the latter into a bound on the residual estimation bias.
This provides a versatile tool for analysing verification protocols whose correctness is formulated at the level of expectation values rather than deterministic classical outcomes.

\subsection{Advantages in practicality}

A central practical contribution of VBPEC is that it substantially improves the viability of verifiable delegated computation under realistic noisy conditions.
Conventional trap-based protocols tolerate physical noise only below a prescribed threshold.
Once the noise exceeds this threshold, they reject with high probability by design, irrespective of the number of samples collected.
Such threshold-based oversensitivity constitutes a major obstacle to experimental implementation on real hardware, in both off-chip and on-chip settings~\cite{Gustiani2025On, Polacchi2025Experimental, Drmota2024Verifiable}.
By contrast, once VBPEC verifies that the PEC noise model is sufficiently accurate, deviations consistent with that model can be securely mitigated rather than automatically treated as grounds for rejection.
Consequently, VBPEC can accept the mitigated estimator with high probability, at the standard additional sampling cost inherent to QEM.

This contribution is particularly relevant given the growing interest in verifiable quantum advantage based on observable estimation~\cite{Abanin2025Observation, Mi2025A}.
As claims of quantum advantage move toward increasingly complex and practically relevant tasks, cryptographic verification is expected to play a central role in establishing that the reported outcomes were produced correctly, even when the computation is delegated to an untrusted quantum device.
At the same time, realistic demonstrations on near-term hardware are likely to depend critically on QEM, making the reliable execution of QEM itself an essential component of any secure quantum-advantage experiment.
VBPEC addresses both requirements simultaneously.
By embedding PEC directly into a cryptographically secure verification protocol, VBPEC enables secure, reliable mitigation of hardware noise while substantially strengthening the protocol’s robustness to experimentally realistic noise levels.
VBPEC therefore provides a concrete route toward cryptographically secure demonstrations of quantum advantage before the advent of the fully fault-tolerant regime, while also opening the door to the secure integration of other device-level noise-mitigation techniques.

\subsection{Future work}

Although Theorem~\ref{theorem:vbpec} is formulated under the assumption that $|\Lambda|$ scales polynomially with $|V|$, conventional trap constructions~\cite{Fitzsimons2017Unconditionally, Leichtle2021Verifying, Kapourniotis2024Unifying, Yang2025Verifiable} immediately imply that one can efficiently verify whether the deviation is sufficiently close to the identity channel regardless of the shape of non-identity deviations.
We here conjecture that, for classes of noise models admitting a description in terms of only polynomially many parameters, it may be possible to efficiently verify whether the actual deviation is consistent with the given noise model.
Such a verification procedure would be particularly valuable for sparse Pauli--Lindblad models~\cite{VandenBerg2023Probabilistic}, which constitute a representative example of this class and are of significant practical interest for verifying the characterised noise model on current quantum hardware.

This work also suggests several directions towards a general theory of verifiable QEM.
The first direction is to extend the analysis to probabilistic error amplification (PEA)~\cite{Mari2021Extending, Kim2023Evidence}, which, as observed in Section~\ref{sec:vbpec}, can be implemented blindly in the same spirit as UBPEC.
The second direction is to make the protocol agnostic to the underlying noise model.
Rather than assuming that the Client is given an accurate reference model for PEC, one would like to sample, learn, and verify the relevant noise parameters on the fly during the protocol execution, as envisioned in~\cite{Xie2026Noise}.
This would strengthen both the practical security and the practical applicability of VBPEC by replacing model trust with statistical certification.

More broadly, these directions point towards a cryptographic framework for securely implementing a wider class of QEM methods~\cite{Czarnik2021Error, Yang2022Efficient, Yang2024Quantum, Liu2024Virtual} and other device-efficient quantum techniques~\cite{Peng2020Simulating, Yuan2021Quantum, Harada2025Density, Yang2025Resource}.
In particular, formalising more general QEM ansatzes, expressed for instance as linear combinations or polynomial transformations of density matrices, as secure ideal resources would provide a unified cryptographic understanding of hybrid quantum-classical procedures.

It is also important to adapt VBPEC to other computational models and hardware regimes.
Optimising VBPEC by incorporating techniques that reduce the amount of quantum communication while preserving composable security~\cite{Sater2026Composable} would be a key step towards more practical near-term implementations.
Combining verifiable QEM with secure fault-tolerant implementations~\cite{Kapourniotis2025Plugging} may further enable secure execution on early fault-tolerant quantum devices.
Together, these extensions would clarify how verifiable error suppression can be deployed across the transition from noisy intermediate-scale devices to early fault-tolerant architectures.

Finally, the test-round construction developed here can also be viewed as a form of verified noise benchmarking, and may therefore serve as a useful primitive for certifying learning-based tasks, including protocols based on classical shadows~\cite{Huang2020Predicting, Seif2023Shadow, Jnane2024Quantum}.
Such directions may ultimately help identify, from a cryptographic perspective, which quantum tasks can be made securely verifiable and how this relates to the boundary of classical simulability.

\begin{acknowledgments}
  B.Y.
  acknowledges the insightful and fruitful discussions with Alireza Seif from IBM Quantum and Dominik Leichtle from the University of Edinburgh.
  All authors received funding from the ANR research grants ANR-21-CE47-0014 (SecNISQ) and ANR-22-PNCQ-0002 (HQI).
\end{acknowledgments}

\appendix

\section{Concentration inequalities\label{appendix:concentration_inequalities}}

\begin{lemma}[Hoeffding's inequality]
  \label{lemma:ineq_hoeffding}
  Let $X_{1},\dots,X_{n}$ be $n\in\mathbb{N}$ independent random variables such that $a_{i} \leq X_{i} \leq b_{i}$ almost surely for each $i\in[n]$.
  Let $\displaystyle S_{n} = \sum_{i=1}^{n} X_{i}$.
  Then for any $t > 0$,
  \begin{equation}
    \begin{split}
      \operatorname{Pr}[S_n - \mathbb{E}[S_{n}] \leq -t]
       & \leq
      \exp\left(
      -\frac{2 t^{2}}{\sum_{i=1}^{n} (b_{i} - a_{i})^{2}}
      \right), \\
      \operatorname{Pr}[S_n - \mathbb{E}[S_{n}] \geq t]
       & \leq
      \exp\left(
      -\frac{2 t^{2}}{\sum_{i=1}^{n} (b_{i} - a_{i})^{2}}
      \right).
    \end{split}
  \end{equation}
\end{lemma}

\begin{lemma}[McDiarmid's inequality~\cite{McDiarmid1989On}]
  \label{lemma:ineq_mcdiarmid}
  Let $X_{1},\ldots,X_{n}$ be independent random variables, and let $f=f(X_{1},\ldots,X_{n})$ be a real-valued function of them.
  Assume that there exist constants $c_{1},\ldots,c_{n}\geq 0$ such that, for every $i\in [n]$, replacing only the $i$-th input changes the value of $f$ by at most $c_{i}$, namely,
  \begin{equation}
    \begin{split}
      \left|
      f(x_{1},\ldots,x_{i},\ldots,x_{n})
      -
      f(x_{1},\ldots,x_{i}^{\prime},\ldots,x_{n})
      \right|
      \leq c_{i}
    \end{split}
  \end{equation}
  for all possible values $x_{1},\ldots,x_{n},x_{i}^{\prime}$.
  Then, for any $\epsilon>0$, it holds that
  \begin{equation}
    \begin{split}
      \operatorname{Pr}\left[f-\mathbb{E}[f]\leq -\epsilon\right]
       & \leq\exp\left(-\frac{2\epsilon^{2}}{\sum_{i=1}^{n} c_{i}^{2}}\right), \\
      \operatorname{Pr}\left[f-\mathbb{E}[f]\geq \epsilon\right]
       & \leq\exp\left(-\frac{2\epsilon^{2}}{\sum_{i=1}^{n} c_{i}^{2}}\right).
    \end{split}
  \end{equation}
\end{lemma}

\begin{lemma}[Serfling's inequality for sampling without replacement]
  \label{lemma:ineq_serfling}
  Let $x_{1},\dots,x_{N}$ be a finite population of real numbers, and define $\displaystyle \mu := \frac{1}{N}\sum_{i=1}^{N} x_{i}$, where $x_{i}\in[a,b]$.
  Let $X_{1},\dots,X_{n}$ be drawn uniformly without replacement from $\{x_{1},\dots,x_{N}\}$, and let $\displaystyle \bar X_{n} = \frac{1}{n}\sum_{i=1}^{n} X_{i}$.
  Then for any $t>0$,
  \begin{equation}
    \begin{split}
      \operatorname{Pr}\left[\bar X_{n} - \mu \leq -t\right]
       & \le
      \exp\left(
      -\frac{2 n t^{2}}{\left(1-\frac{n-1}{N}\right)(b-a)^{2}}
      \right), \\
      \operatorname{Pr}\left[\bar X_{n} - \mu \geq t\right]
       & \le
      \exp\left(
      -\frac{2 n t^{2}}{\left(1-\frac{n-1}{N}\right)(b-a)^{2}}
      \right).
    \end{split}
  \end{equation}
\end{lemma}

\section{Abstract cryptography (AC)\label{appendix:ac}}

Abstract cryptography is a cryptographic framework designed to be top-down and axiomatic to analyze the security of a protocol in an arbitrarily adversarial environment.
In contrast to conventional game-based security, which analyzes each specific adversarial scenario, the abstract cryptography framework provides universal composable security.
Composably secure protocols within the abstract cryptography framework will maintain their security when composed in parallel or in series, ensuring modular security for the entire combined protocol.

The abstract cryptography framework consists of abstract systems with well-distinguished and labelled interfaces for transmitting information to other systems.
Systems are classified into resources, converters, filters, and distinguishers.
The abstract cryptography framework aims to construct a new secure resource $\pi\mathcal{R}$ from an available resource $\mathcal{R}$ and a protocol $\pi$ by showing the security of $\pi$.
Here, the resource $\mathcal{R}$ is an abstract system with an index set of interfaces, $\mathcal{I}$, that mediate transcripts.
The protocol $\pi=\{\pi_{i}\}_{i\in\mathcal{I}}$ is a set of converters $\pi_{i}$ indexed by $\mathcal{I}$, where each converter is a two-interface system mediating between the resource and an external party.

A protocol $\pi$ is proved to be secure by showing the statistical indistinguishability between the constructed resource $\pi\mathcal{R}$ and the ideal resource $\mathcal{S}$, i.e. any distinguisher cannot distinguish with high probability the two resources $\pi\mathcal{R}$ and $\mathcal{S}$.
In concrete terms, the distinguisher is an abstract system that interacts with a resource and attempts to determine whether it is connected to a real or an ideal resource.
It may send inputs, receive outputs, and exploit any observable behavior in order to distinguish the two resources.
Ultimately, the distinguisher must output a single bit indicating its guess: for instance, outputting $1$ if the distinguisher believes it is interacting with the constructed resource $\pi\mathcal{R}$ and $0$ otherwise.
The formal definition of statistical indistinguishability between two resources can be stated as follows.

\begin{definition}[Statistical Indistinguishability of Resources]
  \label{definition:Statistical_Indistinguishability_of_Resources}
  Let $\epsilon>0$, and let $\mathcal{R}_{1}$ and $\mathcal{R}_{2}$ be two resources with the same input and output interfaces.
  The resources are $\epsilon$-statistically-indistinguishable if, for any unbounded distinguisher $\mathcal{D}$, the following holds:
  \begin{equation}\label{eq:distinguishing_advantage}
    \begin{split}
      \left|\operatorname{Pr}\left[\mathcal{D}\left(\mathcal{R}_{1}\right)=1\right] - \operatorname{Pr}\left[\mathcal{D}\left(\mathcal{R}_{2}\right)=1\right]\right| \leq \epsilon,
    \end{split}
  \end{equation}
  which is denoted by $\mathcal{R}_{1}\approx_{\epsilon}\mathcal{R}_{2}$, and $\epsilon$ is referred to as distinguishing advantage.
\end{definition}

Here, the distinguishing advantage $\epsilon$ quantifies how much better a distinguisher can perform than random guessing.
If two resources are completely indistinguishable, the success probability is $\displaystyle \frac{1}{2}$ (the same as random guessing), yielding $\epsilon=0$.
Otherwise, the distinguishing advantage is $\epsilon$, the distinguisher can succeed with probability $\displaystyle\frac{1}{2}+\epsilon$.

\begin{figure}[htbp]
  \centering
  \subfloat[correctness: $\pi\mathcal{R}\approx_{\epsilon}\mathcal{S}$\label{fig:ac_correctness}]{
    \includegraphics[width=\linewidth]{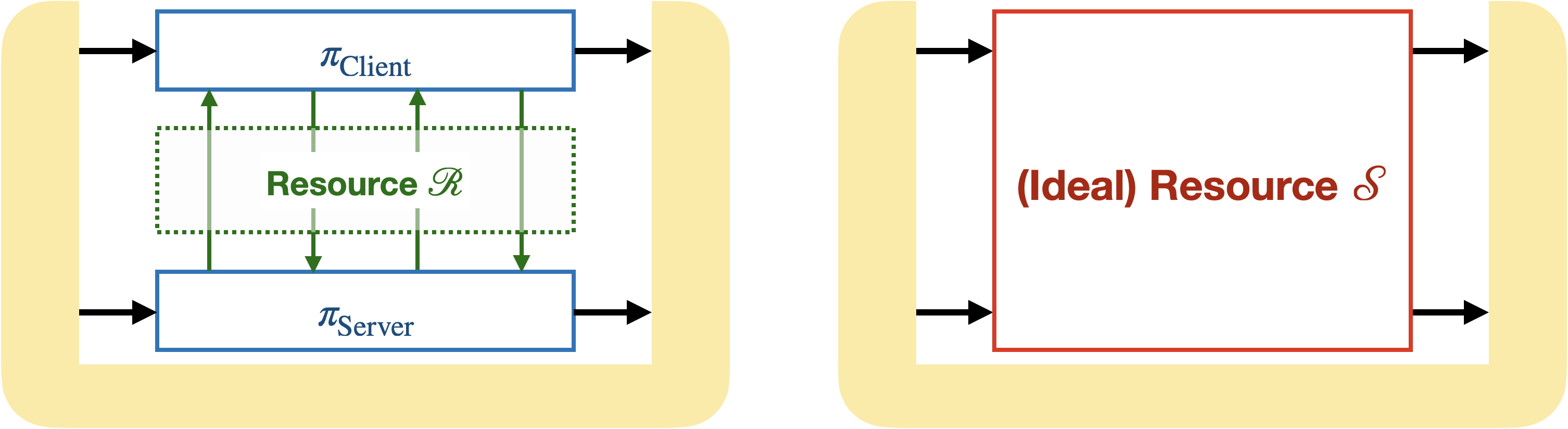}
  }
  \hfill
  \subfloat[security: $\pi_{\mathrm{Client}}\mathcal{R}\approx_{\epsilon}\mathcal{S}\sigma$\label{fig:ac_security}]{
    \includegraphics[width=\linewidth]{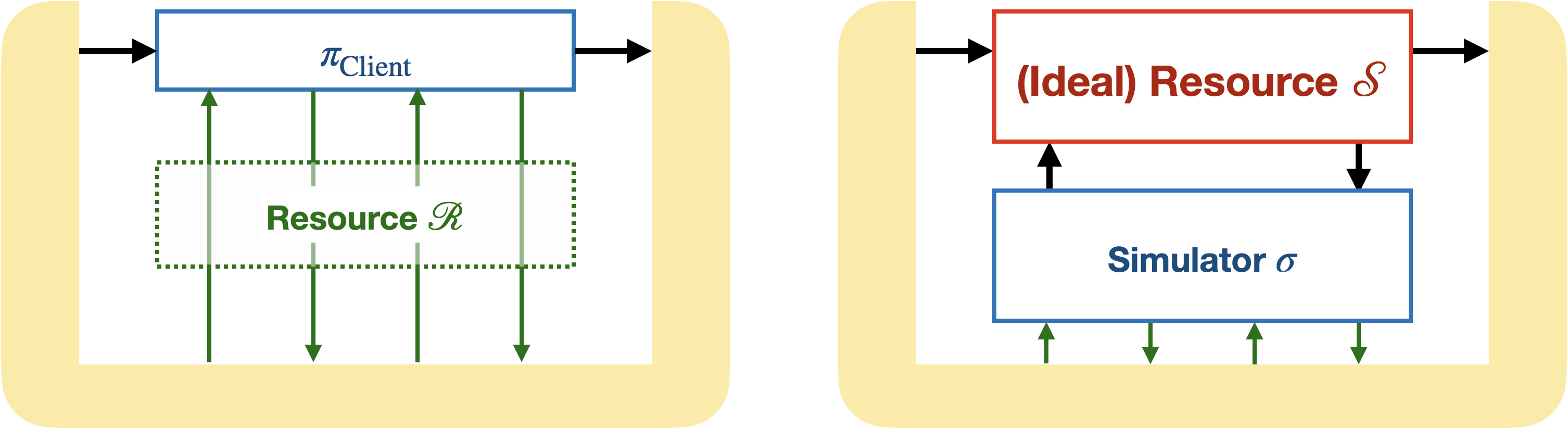}
  }
  \caption{The schematic illustrations of correctness and security are depicted in (a) and (b), respectively.
    On the basis of a secure resource $\mathcal{R}$ as an established channel between the Client and the Server, the protocol $\pi=(\pi_{\mathrm{Client}}, \pi_{\mathrm{Server}})$ constructs a new resource $\pi\mathcal{R}=\pi_{\mathrm{Client}}\mathcal{R}\pi_{\mathrm{Server}}$.
    The transcripts are written as arrows, and the yellow object is the distinguisher that manages the input and output transcripts between the resource and the protocol of interest.
  }
  \label{fig:abstract_cryptography}
\end{figure}

\begin{protocol*}[htbp]
  \caption{\raggedright Universal Blind Quantum Computation (UBQC) (Broadbent et al.~\cite{Broadbent2009Universal})}
  \label{protocol:ubqc}
  \begin{algorithmic}[0]
    \STATE \textbf{Public Information:}
    Public information $(\mathfrak{C}, G, f_{G})$.
    \STATE \textbf{Inputs from Client:}
    The target computation $\mathsf{C}\in \mathfrak{C}$ that produces $\rho$ by measurement angles $\{\phi_{v}\}_{v\in V}$ and allows to measure $\operatorname{Tr}\left[\rho O\right]$.
    \STATE \textbf{Protocol:}
    \begin{enumerate}

      \item The Client generates secret parameters:
            \begin{enumerate}
              \item ($\mathsf{X}$ randomization) The Client chooses a random bit $\mathrm{a}_{v}^{\mathrm{init}}\in\{0,1\}$ for $v\in \mathrm{I}$ and sets $\mathrm{a}_{v}^{\mathrm{init}}=0$ for $v\in V\setminus \mathrm{I}$.
                    The Client also computes $\displaystyle \mathrm{a}_{v}^{\mathrm{prop}} = \bigoplus_{j\in N_{G}(v)}\mathrm{a}_{j}^{\mathrm{init}} \in\{0,1\}$ for all $v\in V$.
              \item ($\mathsf{Z}$ randomization) The Client chooses a random bit $\mathrm{r}_{v}\in\{0,1\}$ for all $v\in V$.
              \item (randomization for blindness) The Client chooses a random $\theta_{v}\in\Theta$ for all $v \in V$.
            \end{enumerate}

      \item The Client prepares and sends to the Server all single qubits corresponding to $v\in V$.
            For $v\in \mathrm{I}$, the Client sequentially sends each qubit in $\displaystyle \left(\bigotimes_{v\in\mathrm{I}}\mathsf{Rz}_{v}(\theta_{v})\mathsf{X}_{v}^{\mathrm{a}_{v}^{\mathrm{init}}}\right)\left[\rho_{\mathrm{init}}\right]$.
            For $v\in V\setminus\mathrm{I}$, the Client sends $\displaystyle |+_{\theta_{v}}\rangle$.

      \item The Server applies the entangling operation $\displaystyle \mathsf{G}=\prod_{(u,v)\in E}\mathsf{CZ}_{u,v}$ to prepare the resource state according to the graph $G$.

      \item For each $v\in V$, the Client and Server interactively perform the MBQC process.
            Once the Client receives the measurement outcome $\mathrm{b}_{j}\in\{0,1\}$ for all $j\in S_{X,v}\cup S_{Z,v}$, where $S_{X,v} = f_{G}^{-1}\left(v\right), S_{Z,v} = \left\{j~|~v\in N_{G}\left(f_{G}\left(j\right)\right)\right\}$, the Client computes the adaptive angle update $\phi_{v}^{\prime}$.
            The Client then computes the measurement angle $\delta_{v}$:
            \begin{equation}\label{eq:delta_UBQC}
              \begin{split}
                s_{X,v} = \bigoplus_{j \in S_{X,v}} \mathrm{b}_{j} \oplus \mathrm{r}_{j}, &
                \quad s_{Z,v} = \bigoplus_{j \in S_{Z,v}} \mathrm{b}_{j} \oplus \mathrm{r}_{j}, \\
                \phi_{v}^{\prime} = (-1)^{s_{X,v}} \phi_{v}+s_{Z,v} \pi,                  &
                \quad \delta_{v} = \left(-1\right)^{\mathrm{a}_{v}^{\mathrm{init}}} \phi_{v}^{\prime} + \theta_{v}+\left(\mathrm{r}_{v} + \mathrm{a}_{v}^{\mathrm{prop}}\right) \pi.
              \end{split}
            \end{equation}
            Note that $s_{X,v} = s_{Z,v} = 0$ for $v\in \mathrm{I}$.
            The Client sends to the Server the angle $\delta_{v}$ and the Server returns to the Client a bit $\mathrm{b}_{v}\in\{0,1\}$ as a measurement result of qubit $v$ with basis $\{|+_{\delta_{v}}\rangle, |-_{\delta_{v}}\rangle\}$.

      \item The Client performs the remaining classical post-processing upon the measurement results of all qubits $v\in V$ and returns $\tilde{y}\in\{-1,1\}$.

    \end{enumerate}
  \end{algorithmic}
\end{protocol*}

When constructing a resource $\pi\mathcal{R}$ from a resource $\mathcal{R}$ and a protocol $\pi$, the security of $\pi$ is then characterised by the indistinguishability between $\pi\mathcal{R}$ and the ideal resource $\mathcal{S}$.
Here, we restrict to the two-party case with an honest ``Client'' and a potentially malicious ``Server''.
The following definition defines how well the protocol $\pi$ constructs $\mathcal{S}$ from $\mathcal{R}$.
\begin{definition}[Construction of Resources]
  \label{definition:Construction_of_resources}
  Let $\epsilon>0$.
  We say that a two-party protocol $\pi$, between an honest Client and a potentially malicious Server, $\epsilon$-statistically-constructs resource $\mathcal{S}$ from resource $\mathcal{R}$ if,
  \begin{itemize}[left=5pt]
    \item it is correct: $\pi\mathcal{R}\approx_{\epsilon}\mathcal{S}$, i.e. when the Server is honest, the client-side outputs between $\pi\mathcal{R}$ and $\mathcal{S}$ are $\epsilon$-statistically indistinguishable;
    \item it is secure against the malicious Server, i.e. there exists a simulator $\sigma$ such that $\pi_{\mathrm{Client}}\mathcal{R}\approx_{\epsilon}\mathcal{S}\sigma$, where $\pi_{\mathrm{Client}}$ is $\pi$'s Client side protocol.
  \end{itemize}
\end{definition}

Intuitively, correctness ensures that the protocol behaves as intended when all parties are honest, while security guarantees that malicious behavior can be emulated in the ideal world by a simulator, thereby preserving composable security.
The existence of such a simulator implies that the use of $\pi\mathcal{R}$ with a malicious Server is still well-indistinguishable from using the ideal resource $\mathcal{S}$, which is designed to be secure.
The schematic illustrations of correctness and security in Definition~\ref{definition:Construction_of_resources} are presented in Fig.~\ref{fig:abstract_cryptography}.

Using the definitions above, we can state the following general composition theorem~\cite{Maurer2011Abstract} that guarantees the additive accumulation of distinguishing advantage when composing two statistically secure protocols.
\begin{theorem}[General Composition of Resources~\cite{Maurer2011Abstract}]
  Let $\mathcal{R}$, $\mathcal{S}$ and $\mathcal{T}$ be resources, $\alpha, \beta$ and $\mathsf{id}$ be protocols, where protocol $\mathsf{id}$ does not modify the resource it is applied to.
  Let $\circ$ and $|$ denote the sequential and parallel composition of protocols and resources, respectively.
  Then the following implications hold:
  \begin{itemize}[left=5pt]
    \item Sequential composability: \\
          if $\alpha \mathcal{R} \approx_{\epsilon_{\alpha}} \mathcal{S}$ and $\beta\mathcal{S} \approx_{\epsilon_{\beta}} \mathcal{T}$, then $\left(\beta \circ\alpha\right) \mathcal{R} \approx_{\epsilon_{\alpha}+\epsilon_{\beta}}\mathcal{T}$.
    \item Context insensitivity: \\
          if $\alpha \mathcal{R} \approx_{\epsilon_{\alpha}} \mathcal{S}$, then $\left(\alpha \mid \mathrm{id}\right)\left(\mathcal{R} \mid \mathcal{T}\right) \approx_{\epsilon_{\alpha}} \left(\mathcal{S} \mid \mathcal{T}\right)$.
  \end{itemize}
  Combining these two properties yields the composability of protocols.
\end{theorem}

\section{Blind delegated quantum computation\label{appendix:bdqc}}

\begin{protocol*}[htbp]
  \caption{\raggedright Universal Blind Observable Estimation (UBOE)}
  \label{protocol:uboe_whole}
  \begin{algorithmic}[0]
    \STATE \textbf{Public Information:}
    Public information $(\mathfrak{C}, G, f_{G}, N_{c})$.
    \STATE \textbf{Inputs from Client:}
    The target computation $\mathsf{C}\in \mathfrak{C}$ that estimates $\operatorname{Tr}\left[\rho O\right]$.
    \STATE \textbf{Protocol:}
    \begin{enumerate}
      \item The Client chooses the total number of rounds $N_{c}$.
      \item For every round indexed by $i \in [N_{c}]$, the Client delegates the target computation with the UBQC protocol (Protocol~\ref{protocol:ubqc}).
      \item Upon receiving and decoding the result of the computation round $i \in [N_{c}]$, the Client assigns the result to $\tilde{y}_{i}$.
            The Client then sets $\displaystyle \tilde{o} = \frac{1}{N_{c}}\sum_{i\in [N_{c}]} \tilde{y}_{i}$ and returns $\tilde{o}$ as the final result.
    \end{enumerate}
  \end{algorithmic}
\end{protocol*}

\begin{resource*}[htbp]
  \caption{\raggedright Blind Delegated Observable Estimation (BDOE)}
  \label{resource:bdoe}
  \begin{algorithmic}[0]
    \STATE \textbf{Public Information:}
    The computational class $\mathfrak{C}$, its associated measurement pattern $C$ that contains the graph $G=(V,E)$ and the measurement flow $f_{G}$, and the total number of rounds $N_{c}$.

    \STATE \textbf{Client's interface:}
    The target computation $\mathsf{C}\in \mathfrak{C}$ and its associated measurement angles $\{\phi_{v}\}_{v\in V}$ to produce the state $\rho$ to be measured by $O$.

    \STATE \textbf{Server's interface:}
    \begin{enumerate}
      \item The interface is filtered so that when $e=0$, the interface does not send any information nor take inputs.
      \item For $e=1$, the Resource receives a quantum state $\sigma$ and $F$, a list of instructions so that the resource produces $s \in \mathbb R$.
    \end{enumerate}

    \STATE \textbf{Processing by the Resource:}
    \begin{enumerate}
      \item If $e=0$, it sets $\displaystyle o = \frac{1}{N_{c}} \sum_{i = 1}^{N_{c}}y_{i}$ with $y_{i}=2Y_{i}-1$ where $Y_{i} \sim \mathcal{B}(p)$, sampled from a Bernoulli distribution $\mathcal{B}(p)$ with $p = \operatorname{Tr}\left[\rho |+_{O}\rangle\langle+_{O}|\right]$.
      \item If $e=1$, it computes $o$ using the transmitted state $\sigma$ and $F$.
      \item It forwards $o$ to the Client.
    \end{enumerate}
  \end{algorithmic}
\end{resource*}

Delegated quantum computation protocols~\cite{Broadbent2009Universal, Fitzsimons2017Unconditionally, Leichtle2021Verifying} allow clients with limited quantum capabilities, such as single-qubit state preparation and communication, to delegate tasks to a powerful server while retaining security guarantees such as blindness and verifiability.
The security of these protocols can be rigorously analyzed within the aforementioned abstract cryptography framework~\cite{Maurer2011Abstract}, which formalizes composable security in a modular way.
This modular framework enables composable security guarantees without requiring incremental and exhaustive proofs of the entire protocol whenever individual components are combined.
We here review the ``Universal Blind Quantum Computation (UBQC)'' protocol~\cite{Broadbent2009Universal} and introduce its variant for observable estimation, ``Universal Blind Observable Estimation(UBOE)''.

The UBQC protocol achieves perfect blindness, allowing the Client to delegate its quantum computation to the untrusted Server, provided the Client can prepare and send a sequence of single qubits via quantum communication.
The procedure of UBQC is based on the measurement-based quantum computation (MBQC) model grounded in the principle of gate teleportation~\cite{Gottesman1999Demonstrating, Raussendorf2001A, Knill2001A, Danos2007The, Briegel2009Measurement}.
In this model, the computation proceeds by first preparing a highly entangled resource state, typically a graph state, and then performing a sequence of adaptive single-qubit measurements in rotated bases.
The measurement outcomes determine subsequent measurement angles, enabling the realisation of arbitrary quantum operations.
More formally, the MBQC procedure is defined by the following measurement pattern.
\begin{definition}[Measurement Pattern]
  A pattern in the MBQC model is given by a graph $G = (V, E)$, input and output vertex sets $\mathrm{I}, \mathrm{O}\subseteq V$, a flow function $f_{G}$ which induces a partial order $\preceq_{G}$ of the qubits $V$, and a set of measurement angles $\left\{\phi_{v}\right\}_{v\in V}$ in the $\mathsf{X}$-$\mathsf{Y}$ plane of the Bloch sphere.
\end{definition}

Based on MBQC, the UBQC protocol is given as Protocol~\ref{protocol:ubqc}.
Based on UBQC, the UBOE protocol estimates the expectation value over a continuous domain by repeating this protocol $N_{c}$ times and averaging the outcomes in the Client's protocol.
The procedure of UBOE is given as Protocol~\ref{protocol:uboe_whole}.

The composable security of the UBOE protocol can be described by the words of the abstract cryptography framework in the same way as UBQC.
Here, UBOE is supposed to return potentially biased output while keeping the blindness of the computation up to the allowed leakage as public information $(\mathfrak{C}, G, f_{G}, N_{c})$.
This is formally defined as the ``Blind Delegated Observable Estimation (BDOE)'' resource as depicted in Resource~\ref{resource:bdoe}, which enables the server to influence the outcome by modelling a potential deviation, while leaking no information to the server beyond the prescribed nature of leakage.
Since UBQC, i.e. each round of UBOE, is shown to achieve perfect composable security~\cite{Dunjko2014Composable} and additional procedure of UBOE occurs only in the Client's classical processing, it follows trivially that the UBOE protocol also achieves perfect composable security.
That is, the UBOE protocol and the BDOE resource are perfectly indistinguishable.
Formally, the security of UBOE is stated as follows.

\begin{theorem}[Security of UBOE]\label{theorem:UBOE}
  The UBOE protocol perfectly constructs the BDOE resource, leaking only public information $\left(\mathfrak{C}, G, f_{G}, N_{c}\right)$.
\end{theorem}

\section{General trap patterns\label{appendix:general_trap_patterns}}

In the test rounds of VBPEC, we use general trap patterns introduced in~\cite{Kapourniotis2024Unifying}, which can be seen as stabilizer tests.
We review the construction of the general trap patterns and show that it can be seen as a Walsh--Hadamard transform of the (quasi-)probability distribution of the deviation.

First, let $\mathcal{S}_{G}$ be the set of stabilizer generators of the graph state associated with $G = (V, E)$ defined as $\mathcal{S}_{G} = \left\{\mathsf{S}_{v}\right\}_{v\in V}$, where
\begin{equation}
  \begin{split}
    \mathsf{S}_{v} = \mathsf{X}_{v}\otimes \left(\bigotimes_{j\in N_{G}(v)} \mathsf{Z}_{j}\right).
  \end{split}
\end{equation}
The stabilizer group of the graph state is then given by $\langle \mathcal{S}_{G}\rangle$.
For a subset $T\subseteq V$, the general trap pattern associated with $T$ verifies the stabilizer $\displaystyle \prod_{v\in T}\mathsf{S}_{v}$.

Since the trap qubits are measured with the $\{|+\rangle,|-\rangle\}$ basis at last, the traps without deviation should be $+1$ eigenstate of $\displaystyle \bigotimes_{v\in S} \mathsf{X}_{v}$.
By commuting this observable backwards through the graph-state entangling operation
\begin{equation}
  \begin{split}
    \mathsf{G} = \prod_{(i,j)\in E}\mathsf{CZ}_{i,j},
  \end{split}
\end{equation}
the corresponding stabilizer before applying $\mathsf{G}$ becomes
\begin{equation}\label{eq:stabilizer_tensor_product}
  \begin{split}
    \left(\bigotimes_{j\in T_{\mathrm{even}}} \mathsf{X}_{j}\right) \otimes \left(\bigotimes_{j\in T_{\mathrm{odd}}} \mathsf{Y}_{j}\right) \otimes \left(\bigotimes_{j\in N_{G}^{\mathrm{odd}}(T)} \mathsf{Z}_{j}\right),
  \end{split}
\end{equation}
where $T_{\mathrm{even}}$ (resp. $T_{\mathrm{odd}}$) are the qubits of even (resp. odd) degree within $T$, and $j \in N_{G}^{\mathrm{odd}}(T)$ means $j$ is in the odd neighbourhood of $T$.
Therefore, one of the tensor product states stabilized by Eq.~\eqref{eq:stabilizer_tensor_product} is given by
\begin{equation}\label{eq:stabilized_tensor_product}
  \begin{split}
    \left(\bigotimes_{j\in T_{\mathrm{even}}} |+\rangle\right) \otimes \left(\bigotimes_{j\in T_{\mathrm{odd}}} |+_{\frac{\pi}{2}}\rangle\right) \otimes \left(\bigotimes_{j\in N_{G}^{\mathrm{odd}}(T)} |0\rangle\right).
  \end{split}
\end{equation}
An illustrative example can be found in Fig.~\ref{fig:pattern_trap_general}.

\begin{figure}
  \centering
  \includegraphics[width=0.8\linewidth]{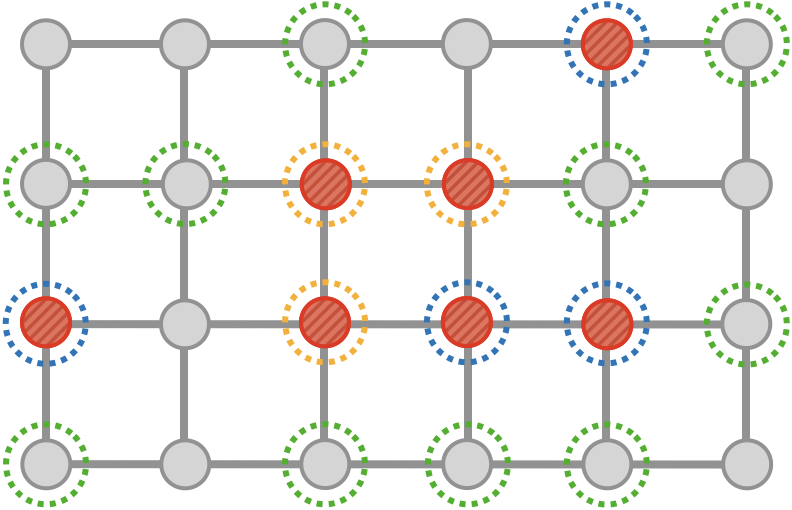}
  \caption{
    A schematic illustration of a specific general trap pattern.
    The red and grey vertices represent trap and dummy qubits, respectively.
    The yellow, blue, green circles describe the qubits in $T_{\mathrm{even}}$, $T_{\mathrm{odd}}$, and $N_{G}^{\mathrm{odd}}(T)$.
  }
  \label{fig:pattern_trap_general}
\end{figure}

Now consider a deviation $\mathsf{P}\in\{\mathsf{I},\mathsf{Z}\}^{|V|}$, and define its support as $S_{\mathsf{P}} = \{v\in V\mid \mathsf{P}(v) = \mathsf{Z}\}$.
Given the trap-qubit measurement outcomes $\mathrm{b}_{v}\in\{0,1\}$ for $v\in T$, the Client post-processes them by assigning the value $+1$ to the trap pattern $T\subseteq V$ when $\displaystyle \bigoplus_{v\in T}\mathrm{b}_{v}=0$, and the value $-1$ otherwise.
Since each $\mathsf{Z}$ deviation flips the corresponding $\mathsf{X}$-basis measurement outcome, this post-processing assigns the value $+1$ for a deviation $\mathsf{P}$ precisely when $|T\cap S_{\mathsf{P}}| = 0 \pmod{2}$, and assigns the value $-1$ otherwise.
This observation yields the following lemma.
\begin{lemma}[Lemma 8 of Kapourniotis et al.~\cite{Kapourniotis2024Unifying}]
  The general trap pattern $T\subseteq V$ chosen uniformly at random detects any non-identity deviation $\mathsf{P}\in\{\mathsf{I},\mathsf{Z}\}^{|V|}\setminus\{\mathsf{I}^{\otimes|V|}\}$ with probability $1/2$, i.e.
  \begin{equation}
    \begin{split}
      \underset{T\sim\mathcal{U}(\Omega_{V})}{\operatorname{Pr}}\left[|T\cap S_{\mathsf{P}}| = 1~(\mathrm{mod}~2)\right] = \frac{1}{2},
    \end{split}
  \end{equation}
  where $\Omega_{V}$ is the power set of $V$.
\end{lemma}
This property underlies the optimised test-round construction of~\cite{Kapourniotis2024Unifying}, where uniformly sampled general trap patterns are post-processed by averaging the obtained trap outputs $\{\tilde{y}_{i}\}_{i\in\mathrm{S}_{\mathsf{T}}}$, thereby estimating the total probability weight of non-identity deviations.

Based on the above construction, we further establish that general trap patterns can be used to extract finer information about the deviation by modifying the Client's post-processing.
The key observation is that $(-1)^{|T\cap S_{\mathsf{P}}|}$ is precisely the Walsh--Hadamard kernel between the trap pattern $T$ and the deviation support $S_{\mathsf{P}}$.
This implies that the Client can recover the deviation support $S_{\mathsf{P}}$ by applying the inverse Walsh--Hadamard transform to the outputs of trap patterns.

To make this precise, let $\vec{p}\in\mathbb{R}_{\geq0}^{|\Omega_{V}|}$ be the probability distribution of the Server's deviation over $\{\mathsf{I},\mathsf{Z}\}^{|V|}$.
For a fixed trap pattern $T\subseteq V$, the expected post-processed trap value is
\begin{equation}
  \begin{split}
    \hat{p}(T)
    := \mathbb{E}\left[\tilde{y}\mid T\right]
    = \sum_{\mathsf{P}\in\{\mathsf{I},\mathsf{Z}\}^{|V|}} (-1)^{\langle T,S_{\mathsf{P}}\rangle} \vec{p}(\mathsf{P}).
  \end{split}
\end{equation}
Therefore, the vector of expected trap outputs is exactly the Walsh--Hadamard transform of the deviation distribution:
\begin{equation}
  \begin{split}
    \hat{p}=H_{|V|}\vec{p}, \quad
    H_{|V|}(T,S_{\mathsf{P}})
    = (-1)^{\langle T, S_{\mathsf{P}}\rangle},
  \end{split}
\end{equation}
where $\langle T,S_{\mathsf{P}}\rangle:=|T\cap S_{\mathsf{P}}|\pmod{2}$.
Since $H_{|V|}^{2}=2^{|V|}\mathsf{I}$, the deviation distribution can be recovered from the expected trap-output vector as
\begin{equation}
  \begin{split}
    \vec{p}=\frac{1}{2^{|V|}}
    H_{|V|}\hat{p}.
  \end{split}
\end{equation}

To this end, the post-processing of general trap patterns introduced in~\cite{Kapourniotis2024Unifying} is the special case in which one extracts only the identity component of the deviation distribution.
This corresponds to applying the first row of the inverse Walsh--Hadamard transform.
Indeed, when $T$ is sampled uniformly at random, we have
\begin{equation}
  \begin{split}
    \mathbb{E}[\tilde{y}]
    = \frac{1}{2^{|V|}}\sum_{T\subseteq V}\hat{p}(T)
    = \vec{p}(\mathsf{I}^{\otimes |V|}).
  \end{split}
\end{equation}
Thus, if $T_{i}\subseteq V$ is sampled uniformly in each test round $i\in\mathrm{S}_{\mathsf{T}}$, the empirical average
\begin{equation}
  \begin{split}
    \tilde{p}_{\mathsf{T}}(\mathsf{I}^{\otimes |V|})
    := \frac{1}{N_{t}}\sum_{i\in\mathrm{S}_{\mathsf{T}}}\tilde{y}_{i}
  \end{split}
\end{equation}
is an unbiased estimator of the probability of the identity deviation.
Equivalently,
\begin{equation}
  \begin{split}
    1-\tilde{p}_{\mathsf{T}}(\mathsf{I}^{\otimes |V|})
    = \frac{1}{N_{t}}\sum_{i\in\mathrm{S}_{\mathsf{T}}}\left(1-\tilde{y}_{i}\right)
  \end{split}
\end{equation}
estimates the total probability weight of non-identity deviations.

VBPEC uses the same trap outputs, but applies a different classical post-processing.
The standard averaging procedure extracts only the identity component of the deviation distribution, whereas VBPEC uses the inverse Walsh--Hadamard structure to estimate selected non-identity components as well.
Specifically, for any $\mathsf{P}\in\{\mathsf{I},\mathsf{Z}\}^{|V|}$, define
\begin{equation}
  \begin{split}
    \tilde{p}_{\mathsf{T}}\left(\mathsf{P}\right)
    := \frac{1}{N_{t}}\sum_{i\in\mathrm{S}_{\mathsf{T}}} (-1)^{\langle T_{i},S_{\mathsf{P}}\rangle}\tilde{y}_{i},
  \end{split}
\end{equation}
where $T_{i}$ is the trap pattern used in the $i$-th test round.
Since $T_{i}$ is chosen uniformly at random from $T_{i}\subseteq V$, it holds that
\begin{equation}
  \begin{split}
    \mathbb{E}\left[ \tilde{p}_{\mathsf{T}}\left(\mathsf{P}\right) \right]
     & = \frac{1}{2^{|V|}}\sum_{T\subseteq V} (-1)^{\langle T,S_{\mathsf{P}}\rangle} \hat{p}(T)
    = \vec{p}\left(\mathsf{P}\right),
  \end{split}
\end{equation}
where the last equality is precisely the inverse Walsh--Hadamard relation.
Summarising above gives the following lemma.
\begin{lemma}[Unbiased estimation of Pauli deviation\label{lemma:Unbiased_estimation_of_Pauli_deviation}]
  Let $\vec{p}$ be the deviation distribution over $\Omega_{V}:=\{\mathsf{I},\mathsf{Z}\}^{|V|}$, and let $T_{i}\subseteq V$ be sampled uniformly at random in each test round $i\in\mathrm{S}_{\mathsf{T}}$.
  For any $\mathsf{P}\in\Omega_{V}$, $\tilde{p}_{\mathsf{T}}\left(\mathsf{P}\right)$ is an unbiased estimator of $\vec{p}(\mathsf{P})$.
\end{lemma}

Hence, by replacing the unsigned average with a Walsh--Hadamard signed average, the Client can estimate the probability of any chosen Pauli deviation, rather than only the identity component.
In VBPEC, this inverse Walsh--Hadamard post-processing is applied to the Pauli strings in the relevant support $\Lambda=\operatorname{supp}(\vec{p}_{\mathrm{PEC}})$, and the resulting estimates are compared with the PEC noise model.

\section{Upper-bounding computation bias by test bias in VBPEC\label{appendix:proof_vbpec}}

Here, we relate the expected bias of the computation to the expected test statistic that estimates the Server's deviation from the PEC noise model.

We first introduce and recall some notation that will be used repeatedly in the following lemmas and proofs.
Let $\Omega_{V} = \{0,1\}^{|V|}$ be the set of length-$|V|$ binary strings and denote the all-zero string by $0_{V}\in\Omega_{V}$.
We also denote by $\mathbf{e}_{0_{V}}\in\mathbb{R}^{|\Omega_{V}|}$ the unit vector at $0_{V}\in\mathbb{R}^{|\Omega_{V}|}$, and by $\mathbf{1}_{V}$ the vector where all elements take the value one.

Conditioned on a fixed sequence of Pauli deviations and a realized partition, let $\vec{p}_{\mathsf{C}}$ and $\vec{p}_{\mathsf{T}}$ be the empirical deviation distributions on the computation and test rounds defined in Eq.~\eqref{eq:define_p_c_p_t_security}, and let $\vec{q}\in\mathbb{R}^{|\Omega_{V}|}$ be the quasi-probability distribution for $\mathcal{E}_{\mathrm{PEC}}^{-1}$.
  We define the effective deviation of the mitigated computation rounds by PEC as $\vec{r}_{\mathsf{C}}:=\vec{q}\ast\vec{p}_{\mathsf{C}}\in\mathbb{R}^{|\Omega_{V}|}$, where $\ast$ denotes convolution over $\Omega_{V}$.

For $\vec{a}, \vec{b}\in\mathbb{R}^{|\Omega_{V}|}$ and $z\in\Omega_{V}$, this convolution is defined by
\begin{equation}
  \begin{split}
    (\vec{a}\ast\vec{b})(z) = \sum_{z^{\prime}\in\Omega_{V}} \vec{a}(z^{\prime})\vec{b}(z+z^{\prime})
  \end{split}
\end{equation}
with addition operation $+$ understood bitwise modulo $2$.

With the aforementioned notations, we bound the conditional expected estimation bias of the target computation by the conditional expected test statistic and the finite-sample mismatch between the computation and test rounds.
  Let $f_{\mathsf{C}}$ be a function performing the computational task $\mathsf{C}\in\mathfrak{C}$ through the MBQC process, taking a density matrix for the measurement pattern over graph $G=(V,E)$ and a projector $\Pi_{\mathsf{C}}^{(G)}$ over $G$, and returning an expectation value in $[-1,1]$.
  The conditional expected computation output $\mu_{\mathsf{C}}=\mathbb{E}[\tilde{o}_{\mathsf{C}}\mid\vec{p}_{\mathsf{C}}]$ can be expressed as
  \begin{equation}
    \begin{split}
      \mu_{\mathsf{C}}
       & = \sum_{\mathsf{P}\in\{\mathsf{I},\mathsf{Z}\}^{|V|}}\vec{r}_{\mathsf{C}}(S_{\mathsf{P}})f_{\mathsf{C}}\left(\mathsf{P}\rho_{\mathsf{C}}^{(G)}\mathsf{P}, \Pi_{\mathsf{C}}^{(G)}\right).
    \end{split}
  \end{equation}
  Through $f_{\mathsf{C}}$, we upper-bound the conditional expected estimation bias by the $1$-norm between the effective computation-round deviation $\vec{r}_{\mathsf{C}}$ after PEC and the noise-free distribution $\mathbf{e}_{0_{V}}$ as follows:

  \begin{equation}
    \begin{split}
       & \left|\mu_{\mathsf{C}} - \operatorname{Tr}[\rho O]\right|                                                                                                                                                                                                                        \\
       & = \Bigg|\sum_{\mathsf{P}\in\{\mathsf{I},\mathsf{Z}\}^{|V|}}\vec{r}_{\mathsf{C}}(S_{\mathsf{P}})f_{\mathsf{C}}\left(\mathsf{P}\rho_{\mathsf{C}}^{(G)}\mathsf{P}, \Pi_{\mathsf{C}}^{(G)}\right) - f_{\mathsf{C}}\left(\rho_{\mathsf{C}}^{(G)}, \Pi_{\mathsf{C}}^{(G)}\right)\Bigg| \\
       & = \Bigg|\sum_{\substack{\mathsf{P}\in\{\mathsf{I},\mathsf{Z}\}^{|V|}                                                                                                                                                                                                             \\ \mathsf{P}\neq \mathsf{I}^{\otimes|V|}}}\vec{r}_{\mathsf{C}}(S_{\mathsf{P}})f_{\mathsf{C}}\left(\mathsf{P}\rho_{\mathsf{C}}^{(G)}\mathsf{P}, \Pi_{\mathsf{C}}^{(G)}\right)\\
       & \quad\quad\quad\quad\quad\quad\quad\quad\quad\quad~+ \left(\vec{r}_{\mathsf{C}}(0_{V}) - 1\right) f_{\mathsf{C}}\left(\rho_{\mathsf{C}}^{(G)}, \Pi_{\mathsf{C}}^{(G)}\right)\Bigg|                                                                                               \\
       & \leq |\vec{r}_{\mathsf{C}}(0_{V}) - 1|  + \sum_{\substack{s\in\Omega_{V}                                                                                                                                                                                                         \\ s\neq 0_{V}}}|\vec{r}_{\mathsf{C}}(s)|\\
       & = \|\vec{r}_{\mathsf{C}} - \mathbf{e}_{0_{V}}\|_{1}.                                                                                                                                                                                                                             \\
    \end{split}
  \end{equation}

We next upper-bound the conditional expected computation bias by the sum of the conditional expected test statistic and the mismatch between the empirical computation- and test-round deviation distributions, as summarized in Lemma~\ref{lemma:link_bias_test}.
\begin{lemma}[Linking biases of test and computation\label{lemma:link_bias_test}]
  Using the notations above, it holds that
  
    \begin{equation}
      \begin{split}
        \left|\mu_{\mathsf{C}} - \operatorname{Tr}[\rho O]\right|
         & \leq \left\|\vec{q}\right\|_{1}
        \left(
        \Delta_{\mathsf{T}}
        + \mathrm{D}_{\Lambda}\left(\vec{p}_{\mathsf{C}},\vec{p}_{\mathsf{T}}\right)
        \right).
      \end{split}
    \end{equation}
  
\end{lemma}

\begin{proof}
  Recalling that $\vec{r}_{\mathsf{C}} = \vec{q}\ast\vec{p}_{\mathsf{C}}$ and $\mathbf{e}_{0_{V}} = \vec{q}\ast\vec{p}_{\mathrm{PEC}}$, we obtain
  
    \begin{equation}
      \begin{split}
        \vec{r}_{\mathsf{C}}-\mathbf{e}_{0_{V}}
         & = \vec{q}\ast\left(\vec{p}_{\mathsf{C}} - \vec{p}_{\mathrm{PEC}}\right).
      \end{split}
    \end{equation}
  
  Using the convolution inequality for the $1$-norm, namely $\|a \ast b\|_{1} \leq \|a\|_{1} \|b\|_{1}$, we have
  
    \begin{equation}\label{eq:r-e_v0}
      \begin{split}
        \left|\mu_{\mathsf{C}} - \operatorname{Tr}[\rho O]\right|
         & \leq \left\|\vec{r}_{\mathsf{C}}-\mathbf{e}_{0_{V}}\right\|_{1}                                  \\
         & \leq \left\|\vec{q}\right\|_{1} \left\|\vec{p}_{\mathsf{C}} - \vec{p}_{\mathrm{PEC}}\right\|_{1} \\
         & = \left\|\vec{q}\right\|_{1}
        \mathrm{D}_{\Lambda}\left(\vec{p}_{\mathsf{C}},\vec{p}_{\mathrm{PEC}}\right),
      \end{split}
    \end{equation}
  
  where the equality follows because $\vec{p}_{\mathsf{C}}$ and $\vec{p}_{\mathrm{PEC}}$ are probability distributions and $\vec{p}_{\mathrm{PEC}}$ is supported on $\Lambda$.

  We next relate the computation-round deviation distribution to the test statistic.
    Using the triangle inequality for $\mathrm{D}_{\Lambda}$ and Jensen's inequality $|\mathbb{E}[X]| \leq \mathbb{E}[|X|]$, we obtain
    \begin{equation}\label{eq:p_c_p_t_triangle_security}
      \begin{split}
        \mathrm{D}_{\Lambda}\left(\vec{p}_{\mathsf{C}},\vec{p}_{\mathrm{PEC}}\right)
         & \leq \mathrm{D}_{\Lambda}\left(\vec{p}_{\mathsf{T}},\vec{p}_{\mathrm{PEC}}\right)
        + \mathrm{D}_{\Lambda}\left(\vec{p}_{\mathsf{C}},\vec{p}_{\mathsf{T}}\right)                                              \\
         & = \mathrm{D}_{\Lambda}\left(\mathbb{E}[\tilde{p}_{\mathsf{T}} \mid \vec{p}_{\mathsf{T}}],\vec{p}_{\mathrm{PEC}}\right)
        + \mathrm{D}_{\Lambda}\left(\vec{p}_{\mathsf{C}},\vec{p}_{\mathsf{T}}\right)                                              \\
         & \leq \mathbb{E}[\mathrm{D}_{\Lambda}(\tilde{p}_{\mathsf{T}}, \vec{p}_{\mathrm{PEC}}) \mid \vec{p}_{\mathsf{T}}]
        + \mathrm{D}_{\Lambda}\left(\vec{p}_{\mathsf{C}},\vec{p}_{\mathsf{T}}\right)                                              \\
         & = \Delta_{\mathsf{T}}
        + \mathrm{D}_{\Lambda}\left(\vec{p}_{\mathsf{C}},\vec{p}_{\mathsf{T}}\right),
      \end{split}
    \end{equation}
    where the last equality comes from Eq.~\eqref{eq:Delta_T}, i.e., $\Delta_{\mathsf{T}} = \mathbb{E}[\mathrm{D}_{\Lambda}(\tilde{p}_{\mathsf{T}}, \vec{p}_{\mathrm{PEC}}) \mid \vec{p}_{\mathsf{T}}]$.
  
\end{proof}

\section{Proof of noise-robustness\label{appendix:proof_noise-robustness}}

\subsection{Noise-robustness of VBOE under honest noise}

We now analyze the noise robustness of VBOE in the honest-noise setting.
Since VBOE is a special case of VBPEC that omits PEC, its test statistic can also be described as the identity-reference special case of $\tilde{\Delta}$ in Eq.~\eqref{eq:vbpec_process_classical}, by replacing $\vec{p}_{\mathrm{PEC}}$ with $\vec{p}_{\mathrm{ref,VBOE}}:=\mathbf{e}_{0_{V}}$ and setting $\Lambda_{\mathrm{VBOE}}=\{\mathsf{I}^{\otimes|V|}\}$.
The VBOE test statistic $\tilde{\Delta}_{\mathrm{VBOE}}$ then become
\begin{equation}
  \begin{split}
    \tilde{\Delta}_{\mathrm{VBOE}}
     & := \sum_{\mathsf{P}\in\Lambda_{\mathrm{VBOE}}} \left|\mathbf{e}_{0_{V}}(\mathsf{P}) - \tilde{p}_{\mathsf{T}}(\mathsf{P})\right|
    \\
     & \quad + \left|1 - \sum_{\mathsf{P}\in\Lambda_{\mathrm{VBOE}}} \tilde{p}_{\mathsf{T}}(\mathsf{P})\right|                           \\
     & = \left(1-\tilde{p}_{\mathsf{T}}(\mathsf{I}^{\otimes|V|})\right) + \left(1-\tilde{p}_{\mathsf{T}}(\mathsf{I}^{\otimes|V|})\right) \\
     & = 2\left(1-\tilde{p}_{\mathsf{T}}(\mathsf{I}^{\otimes|V|})\right).
  \end{split}
\end{equation}

For comparison with matched-noise VBPEC, suppose that the Server is honest and that each round is independently affected by the actual noise channel $\mathcal{E}_{\mathrm{PEC}}$, with deviation distribution $\vec{p}$ and total non-identity probability $p:=1-\vec{p}(\mathsf{I}^{\otimes|V|})$.
For each test round, the output satisfies $\tilde{y}_{i}\in\{-1,1\}$.
Since a uniformly random general trap pattern detects any fixed non-identity deviation with probability $1/2$, we have $\operatorname{Pr}[\tilde{y}_{i}=-1]=p/2$ and $\operatorname{Pr}[\tilde{y}_{i}=1]=1-p/2$.
It follows that $\mathbb{E}[\tilde{y}_{i}]=1-p$ and hence $\mathbb{E}[\tilde{p}_{\mathsf{T}}(\mathsf{I}^{\otimes|V|})]=1-p$.
Therefore, the expected VBOE test statistic is
\begin{equation}
  \begin{split}
    \Delta_{\mathrm{VBOE}}
    := \mathbb{E}[\tilde{\Delta}_{\mathrm{VBOE}}]
     & = 2\left(1 - \mathbb{E}[\tilde{p}_{\mathsf{T}}(\mathsf{I}^{\otimes|V|})]\right) \\
     & = 2(1-(1-p))                                                                    \\
     & = 2p.
  \end{split}
\end{equation}

Since $\mathsf{Acc} \iff \tilde{\Delta}_{\mathrm{VBOE}} \leq \epsilon_{t}$, the robustness of VBOE is governed by the comparison between the threshold $\epsilon_{t}$ and the honest mean value $2p$.
The following lemma makes this precise.

\begin{lemma}[Acceptance probability of VBOE under honest noise]
  \label{lemma:probability_acceptance_honest_VBOE}
  Assume that the Server is honest and that each round is independently affected by $\mathcal{E}_{\mathrm{PEC}}$ with non-identity probability $p$.
  If $2p < \epsilon_{t}$, the protocol rejects the computation with negligibly small probability.
  If $2p > \epsilon_{t}$, the protocol accepts the computation with negligibly small probability.
\end{lemma}

\begin{proof}

  As noted above, $\mathsf{Acc} \iff \tilde{\Delta}_{\mathrm{VBOE}} \leq \epsilon_{t}$ and hence $\mathsf{Rej} \iff \tilde{\Delta}_{\mathrm{VBOE}} > \epsilon_{t}$.
  Moreover, $\tilde{p}_{\mathsf{T}}(\mathsf{I}^{\otimes|V|})=N_{t}^{-1}\sum_{i\in\mathrm{S}_{\mathsf{T}}}\tilde{y}_{i}$ is the average of independent random variables in $[-1,1]$, with expectation $1-p$.
  If $2p < \epsilon_{t}$, Hoeffding's inequality gives
  \begin{equation}
    \begin{split}
      \operatorname{Pr}[\mathsf{Acc}]
       & = 1 - \operatorname{Pr}[\mathsf{Rej}]                                                                                  \\
       & = 1 - \operatorname{Pr}\left[\tilde{p}_{\mathsf{T}}(\mathsf{I}^{\otimes|V|})-(1-p) < -\frac{\epsilon_{t}-2p}{2}\right] \\
       & \geq 1 - \exp\left( - \frac{\left(\epsilon_{t} - 2p\right)^{2}}{8}
      N_{t}\right).
    \end{split}
  \end{equation}
  If $2p > \epsilon_{t}$, Hoeffding's inequality gives
  \begin{equation}
    \begin{split}
      \operatorname{Pr}[\mathsf{Acc}]
       & = \operatorname{Pr}\left[\tilde{p}_{\mathsf{T}}(\mathsf{I}^{\otimes|V|})-(1-p) \geq \frac{2p-\epsilon_{t}}{2}\right] \\
       & \leq \exp\left( - \frac{\left(2p - \epsilon_{t}\right)^{2}}{8}
      N_{t}\right).
    \end{split}
  \end{equation}

\end{proof}

We next combine the above acceptance analysis with the estimation accuracy of the computation rounds.
Let $\Delta_{\mathsf{C}} := \left|\operatorname{Tr}[\mathcal{E}_{\mathrm{PEC}}(\rho)O] - \operatorname{Tr}[\rho O]\right|$ be the honest noisy bias of the target observable estimation.
Since we assume $\|O\|_{\infty} = 1$, $\Delta_{\mathsf{C}} \leq 2p$.
Then, the following lemma holds.

\begin{lemma}[Correctness of VBOE under honest noise]
  \label{lemma:correctness_honest_VBOE}
  Assume that the Server is honest and that each round is independently affected by $\mathcal{E}_{\mathrm{PEC}}$ with non-identity probability $p$.
  If $\Delta_{\mathsf{C}} < \epsilon$, then $\{\mathsf{Bad}\}$ occurs with exponentially small probability.
  If $\Delta_{\mathsf{C}} > \epsilon$, then $\{\mathsf{Good}\}$ occurs with exponentially small probability.
\end{lemma}

\begin{proof}
  Under honest execution, the computation-round outputs are independent random variables in $\{-1, 1\}$ with mean $\operatorname{Tr}[ \mathcal{E}_{\mathrm{PEC}}( \rho ) O ]$.
  Hence, if $\Delta_{\mathsf{C}} < \epsilon$, Hoeffding's inequality implies
  \begin{equation}
    \begin{split}
      \operatorname{Pr}[\mathsf{Good}]
       & = 1 - \operatorname{Pr}[\mathsf{Bad}]                                                                                                                                   \\
       & = 1 - \operatorname{Pr}\left[\left|\tilde{o}_{\mathsf{C}} - \operatorname{Tr}[\rho O ]\right| > \epsilon\right]                                                         \\
       & \geq 1 - \operatorname{Pr}\left[\left|\tilde{o}_{\mathsf{C}} - \operatorname{Tr}[ \mathcal{E}_{\mathrm{PEC}}( \rho ) O ]\right| > \epsilon - \Delta_{\mathsf{C}}\right] \\
       & \geq 1 - 2 \exp\left(- \frac{\left(\epsilon - \Delta_{\mathsf{C}}\right)^{2}}{2}
      N_{c}\right).
      \\
    \end{split}
  \end{equation}
  If $\Delta_{\mathsf{C}} > \epsilon$, so that $\mathsf{Good}$ holds, it is necessary that $\left|\tilde{o}_{\mathsf{C}} - \operatorname{Tr}[ \mathcal{E}_{\mathrm{PEC}}( \rho ) O ]\right| > \Delta_{\mathsf{C}} - \epsilon$.
  Using Hoeffding's inequality, we obtain
  \begin{equation}
    \begin{split}
      \operatorname{Pr}[\mathsf{Good}]
       & \leq \operatorname{Pr}\left[\left|\tilde{o}_{\mathsf{C}} - \operatorname{Tr}[ \mathcal{E}_{\mathrm{PEC}}( \rho ) O ]\right| > \Delta_{\mathsf{C}} - \epsilon\right] \\
       & \leq 2 \exp\left(- \frac{\left(\Delta_{\mathsf{C}} - \epsilon\right)^{2}}{2}
      N_{c}\right).
    \end{split}
  \end{equation}
  Note that since $\Delta_{\mathsf{C}} \leq 2p$, the regime $\Delta_{\mathsf{C}} > \epsilon$ can occur only if $2p > \epsilon$.
\end{proof}

From Lemma~\ref{lemma:probability_acceptance_honest_VBOE} and Lemma~\ref{lemma:correctness_honest_VBOE}, we obtain the following corollary on the probability that VBOE accepts the correct output.

\begin{corollary}[Noise-robustness of VBOE under honest noise]
  \label{corollary:noise-robustness_honest_VBOE}
  Assume that the Server is honest and that each round is independently affected by $\mathcal{E}_{\mathrm{PEC}}$ with non-identity probability $p$.
  If $2p < \epsilon_{t}$ and $\Delta_{\mathsf{C}} < \epsilon$, the protocol rejects the computation with exponentially small probability.
  If $2p > \epsilon_{t}$ or $\Delta_{\mathsf{C}} > \epsilon$, the protocol accepts the correct output with exponentially small probability.
\end{corollary}

\begin{proof}

  If $2p < \epsilon_{t}$ and $\Delta_{\mathsf{C}} < \epsilon$, applying the union bound yields
  \begin{equation}
    \begin{split}
      \operatorname{Pr}[ \mathsf{Acc} \wedge \mathsf{Good} ]
       & = 1 - \operatorname{Pr}[ \mathsf{Rej} \vee \mathsf{Bad} ]                      \\
       & \geq 1 - \operatorname{Pr}[ \mathsf{Rej} ] - \operatorname{Pr}[ \mathsf{Bad} ] \\
       & \geq1 - \exp\left(- \frac{\left(\epsilon_{t} - 2p\right)^{2}}{8}
      N_{t}\right)                                                                      \\ &\quad\quad- 2 \exp\left(- \frac{\left(\epsilon - \Delta_{\mathsf{C}}\right)^{2}}{2}N_{c}\right).
    \end{split}
  \end{equation}

  Otherwise, if $2p > \epsilon_{t}$, it follows directly from Lemma~\ref{lemma:probability_acceptance_honest_VBOE} that $\operatorname{Pr}[\mathsf{Acc}]$ is negligibly small.
  If $\Delta_{\mathsf{C}} > \epsilon$, it follows directly from Lemma~\ref{lemma:correctness_honest_VBOE} that $\operatorname{Pr}[\mathsf{Good}]$ is negligibly small.
  Since $\mathsf{Acc} \wedge \mathsf{Good} \subseteq \mathsf{Acc}$ and $\mathsf{Acc} \wedge \mathsf{Good} \subseteq \mathsf{Good}$, for both cases, $\operatorname{Pr}[\mathsf{Acc}\wedge\mathsf{Good}]$ is negligibly small.

\end{proof}

\subsection{Noise-robustness of VBPEC under honest noise mismatch}

We now analyze the noise-robustness of VBPEC when the Server is honest, but the actual noise does not exactly match the noise model used for PEC.
We assume that each round is independently affected by the same stochastic Pauli channel with deviation distribution $\vec{p}\in\mathbb{R}_{\geq 0}^{|\Omega_V|}$.

\begin{lemma}[Expected test output under honest noise mismatch]
  \label{lemma:expected_test_output_under_honest_noise_mismatch_vbpec}
  Define $D_{\mathrm{mis}}$ and $D_{\mathrm{stat}}$ as
  \begin{equation}
    \begin{split}
      D_{\mathrm{mis}} & := \mathrm{D}_{\Lambda}\left(\vec{p}, \vec{p}_{\mathrm{PEC}}\right) = \left\| \vec{p}-\vec{p}_{\mathrm{PEC}} \right\|_{1}, \\ D_{\mathrm{stat}} & := \frac{2|\Lambda|}{\sqrt{N_{t}}}.
    \end{split}
  \end{equation}
  Under the honest noise-mismatch setting, the expected test statistic satisfies
  \begin{equation}
    \begin{split}
      \Delta
      := \mathbb{E}[\tilde{\Delta}]
      \leq D_{\mathrm{mis}} + D_{\mathrm{stat}}.
    \end{split}
  \end{equation}
\end{lemma}

\begin{proof}
  First, by the triangle inequality,
  \begin{equation}
    \begin{split}
      \tilde{\Delta}
      = \mathrm{D}_{\Lambda}\left(\tilde{p}_{\mathsf{T}}, \vec{p}_{\mathrm{PEC}}\right)
       & \leq \mathrm{D}_{\Lambda}\left(\tilde{p}_{\mathsf{T}}, \vec{p}\right) + \mathrm{D}_{\Lambda}\left(\vec{p}, \vec{p}_{\mathrm{PEC}}\right) \\
       & = \mathrm{D}_{\Lambda}\left(\tilde{p}_{\mathsf{T}}, \vec{p}\right) + D_{\mathrm{mis}}.
    \end{split}
  \end{equation}
  Taking the expectation of both sides and using Eq.~\eqref{eq:inequality_triangle_D_Lambda}, we have
  \begin{equation}
    \begin{split}
      \mathbb{E}[\tilde{\Delta}]
       & \leq D_{\mathrm{mis}} + 2\sum_{\mathsf{P}\in\Lambda} \mathbb{E}\left[\left|\vec{p}(\mathsf{P}) - \tilde{p}_{\mathsf{T}}(\mathsf{P})\right|\right].
    \end{split}
  \end{equation}
  Recall that $\mathbb{E}[\tilde{p}_{\mathsf{T}}(\mathsf{P})] = \vec{p}(\mathsf{P})$ and $\operatorname{Var}[\tilde{p}_{\mathsf{T}}(\mathsf{P})] \leq 1 / N_{t}$ as discussed in the correctness proof.
  Similarly to Eq.~\eqref{eq:mean_abs_gap_P} in the correctness proof, we observe that
  \begin{equation}
    \begin{split}
      \mathbb{E}\left[\left|\vec{p}(\mathsf{P}) - \tilde{p}_{\mathsf{T}}(\mathsf{P})\right|\right]
      \leq \sqrt{\operatorname{Var}[\tilde{p}_{\mathsf{T}}(\mathsf{P})]}
      \leq \frac{1}{\sqrt{N_{t}}},
    \end{split}
  \end{equation}
  which further implies that
  \begin{equation}
    \begin{split}
      2\sum_{\mathsf{P}\in\Lambda}\mathbb{E}\left[\left|\vec{p}(\mathsf{P}) - \tilde{p}_{\mathsf{T}}(\mathsf{P})\right|\right]
      \leq \frac{2|\Lambda|}{\sqrt{N_{t}}}
      = D_{\mathrm{stat}}.
    \end{split}
  \end{equation}
  Therefore, we obtain $\mathbb{E}[\tilde{\Delta}] \leq D_{\mathrm{mis}} + D_{\mathrm{stat}}$.
\end{proof}

Lemma~\ref{lemma:expected_test_output_under_honest_noise_mismatch_vbpec} shows that the honest noise mismatch contributes a deterministic shift $D_{\mathrm{mis}}$ to the expected test statistic, while the finite number of test rounds contributes the statistical term $D_{\mathrm{stat}}$.
Therefore, the acceptance threshold is governed by the gap between $\epsilon_{t}$ and $D_{\mathrm{mis}} + D_{\mathrm{stat}}$.
On the other hand, the accuracy of the PEC-processed computation rounds is governed by the bias $\left\|\vec{q}\right\|_{1} D_{\mathrm{mis}}$.

\begin{lemma}[Noise-robustness of VBPEC under honest noise mismatch]
  \label{lemma:noise_robustness_honest_mismatch_vbpec}
  The probability $\operatorname{Pr}\left[\mathsf{Acc}\wedge\mathsf{Good}\right]$ is exponentially close to one, provided that
  \begin{equation}
    \begin{split}
      D_{\mathrm{mis}} < \min\left\{\frac{\epsilon}{\left\|\vec{q}\right\|_{1}}, \epsilon_{t}\right\}, \quad N_{t} \geq \frac{16|\Lambda|^{2}}{(\epsilon_{t}-D_{\mathrm{mis}})^{2}}.
    \end{split}
  \end{equation}
\end{lemma}
\begin{proof}
  We prove an explicit lower bound on $\operatorname{Pr}\left[\mathsf{Acc}\wedge\mathsf{Good}\right]$.
  By the union bound,
  \begin{equation}
    \begin{split}
      \operatorname{Pr}\left[\mathsf{Acc}\wedge\mathsf{Good}\right]
       & = 1-\operatorname{Pr}\left[\mathsf{Rej}\vee\mathsf{Bad}\right]                                \\
       & \geq 1-\operatorname{Pr}\left[\mathsf{Rej}\right]-\operatorname{Pr}\left[\mathsf{Bad}\right].
    \end{split}
  \end{equation}

  We first bound $\operatorname{Pr}\left[\mathsf{Bad}\right]$.
  To begin with, we see that the triangle inequality yields
  \begin{equation}
    \begin{split}
      \left|\tilde{o}_{\mathsf{C}} - \mathbb{E}\left[\tilde{o}_{\mathsf{C}}\right]\right|
      \geq \left|\tilde{o}_{\mathsf{C}} - \operatorname{Tr}\left[\rho O\right]\right| - \left|\mathbb{E}\left[\tilde{o}_{\mathsf{C}}\right] - \operatorname{Tr}\left[\rho O\right]\right|.
    \end{split}
  \end{equation}
  By Lemma~\ref{lemma:link_bias_test}, the bias of the PEC-processed computation output is bounded as
  \begin{equation}
    \begin{split}
      \left|\mathbb{E}\left[\tilde{o}_{\mathsf{C}}\right] - \operatorname{Tr}\left[\rho O\right]\right|
       & \leq \left\|\vec{q}\right\|_{1}\left\|\vec{p} - \vec{p}_{\mathrm{PEC}}\right\|_{1} \\
       & = \epsilon - \left(\epsilon - \left\|\vec{q}\right\|_{1}D_{\mathrm{mis}}\right)
    \end{split}
  \end{equation}
  This implies that when $\left|\tilde{o}_{\mathsf{C}} - \operatorname{Tr}\left[\rho O\right]\right| > \epsilon$ occurs, it holds that
  \begin{equation}
    \begin{split}
      \left|\tilde{o}_{\mathsf{C}}-\mathbb{E}\left[\tilde{o}_{\mathsf{C}}\right]\right|
      > \epsilon-\left\|\vec{q}\right\|_{1}D_{\mathrm{mis}}
      > 0,
    \end{split}
  \end{equation}
  along with the condition $D_{\mathrm{mis}} < \epsilon/\left\|\vec{q}\right\|_{1}$.
  Since each summand in $\tilde{o}_{\mathsf{C}}$ is bounded in $\left[-\left\|\vec{q}\right\|_{1},\left\|\vec{q}\right\|_{1}\right]$, by Hoeffding's inequality, we obtain
  \begin{equation}\label{eq:bad_honest_mismatch_vbpec}
    \begin{split}
      \operatorname{Pr}\left[\mathsf{Bad}\right]
      \leq 2\exp\left(- \frac{\left(\epsilon-\left\|\vec{q}\right\|_{1}
        D_{\mathrm{mis}}\right)^{2}}{2\left\|\vec{q}\right\|_{1}^{2}}N_{c}\right).
    \end{split}
  \end{equation}

  We next bound $\operatorname{Pr}\left[\mathsf{Rej}\right]$.
  By Lemma~\ref{lemma:expected_test_output_under_honest_noise_mismatch_vbpec}, we have $\Delta \leq D_{\mathrm{mis}}+D_{\mathrm{stat}}$.
  This implies that
  \begin{equation}
    \begin{split}
      \operatorname{Pr}\left[\mathsf{Rej}\right]
       & = \operatorname{Pr}\left[\tilde{\Delta}>\epsilon_{t}\right]
      \leq \operatorname{Pr}\left[\tilde{\Delta}-\Delta>\epsilon_{t}-D_{\mathrm{mis}}-D_{\mathrm{stat}}\right].
    \end{split}
  \end{equation}
  Since we have assumed $D_{\mathrm{mis}}<\epsilon_{t}$, if further $D_{\mathrm{stat}} \leq (\epsilon_{t}-D_{\mathrm{mis}}) / 2$, we obtain $\epsilon_{t}-D_{\mathrm{mis}}-D_{\mathrm{stat}}\geq 0$.
  This condition $D_{\mathrm{stat}} \leq (\epsilon_{t}-D_{\mathrm{mis}}) / 2$ is then equivalent to
  \begin{equation}
    \begin{split}
      N_{t} \geq \frac{16|\Lambda|^{2}}{(\epsilon_{t}-D_{\mathrm{mis}})^{2}}.
    \end{split}
  \end{equation}

  As in the correctness proof, changing a single test-round pair changes $\tilde{\Delta}$ by at most $4|\Lambda|/N_{t}$.
  Thus, applying McDiarmid's inequality with $c_{i}=4|\Lambda|/N_{t}$ gives
  \begin{equation}\label{eq:rej_honest_mismatch_vbpec}
    \begin{split}
      \operatorname{Pr}\left[\mathsf{Rej}\right]
      \leq 2\exp\left(- \frac{\left(\epsilon_{t}-D_{\mathrm{mis}}-D_{\mathrm{stat}}\right)^{2}}{8|\Lambda|^{2}}
      N_{t}\right).
    \end{split}
  \end{equation}
  Combining Eq.~\eqref{eq:bad_honest_mismatch_vbpec} and Eq.~\eqref{eq:rej_honest_mismatch_vbpec}, we obtain
  \begin{equation}
    \begin{split}
      \operatorname{Pr}\left[\mathsf{Acc}\wedge\mathsf{Good}\right]
       & \geq 1 - 2\exp\left(- \frac{\left(\epsilon-\left\|\vec{q}\right\|_{1}
      D_{\mathrm{mis}}\right)^{2}}{2\left\|\vec{q}\right\|_{1}^{2}}N_{c}\right) \\ &\quad\quad - 2\exp\left(- \frac{\left(\epsilon_{t}-D_{\mathrm{mis}}-D_{\mathrm{stat}}\right)^{2}}{8|\Lambda|^{2}}N_{t}\right).
    \end{split}
  \end{equation}
  Hence, $\operatorname{Pr}\left[\mathsf{Acc}\wedge\mathsf{Good}\right]$ is exponentially close to one.
\end{proof}

\section{Security issues with applying PEC to test rounds\label{appendix:PEC_on_test}}

\begin{figure*}
  \centering
  \subfloat[]{
    \includegraphics[width=0.32\linewidth]{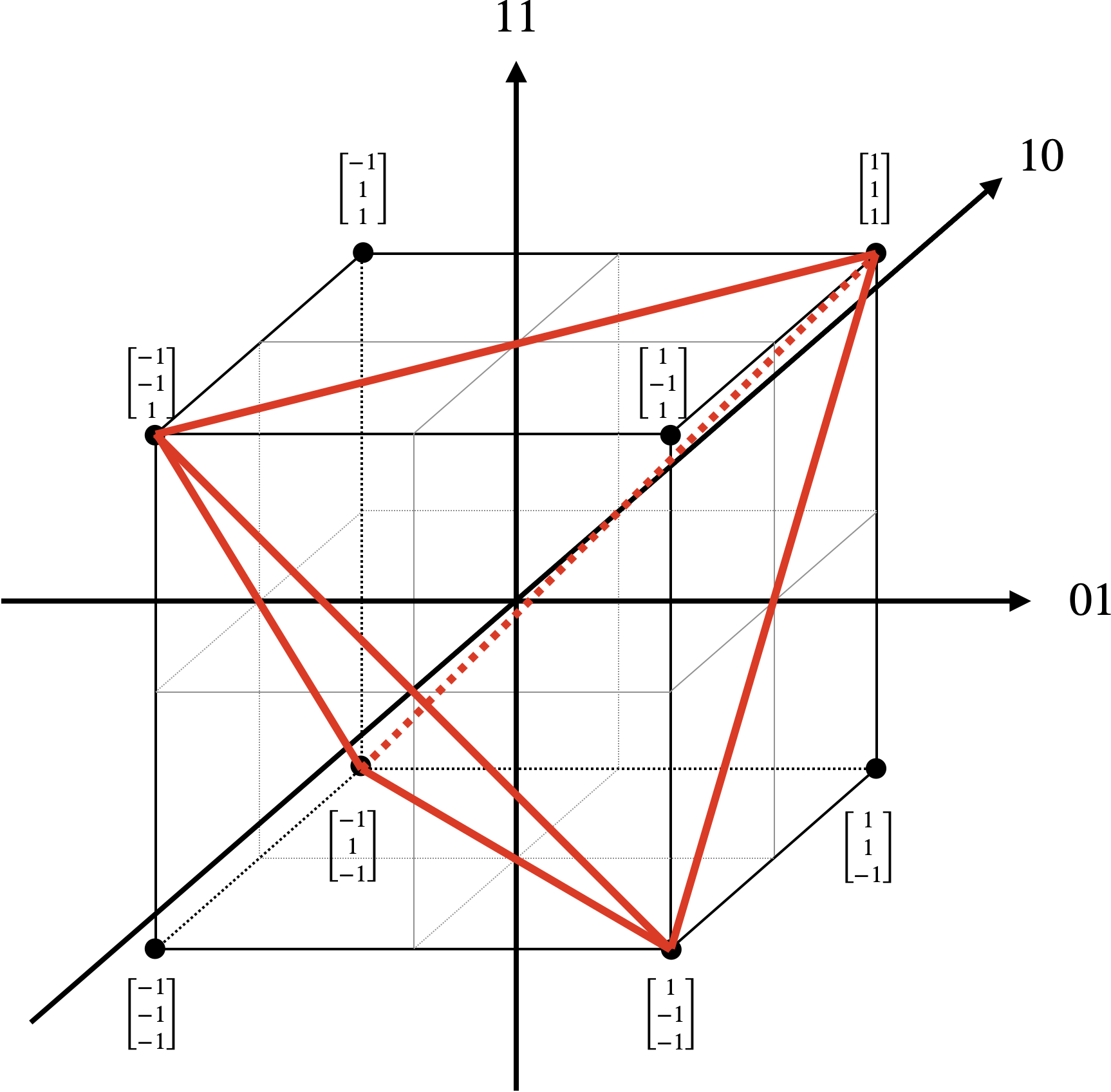}
  }
  \hfill
  \subfloat[]{
    \includegraphics[width=0.32\linewidth]{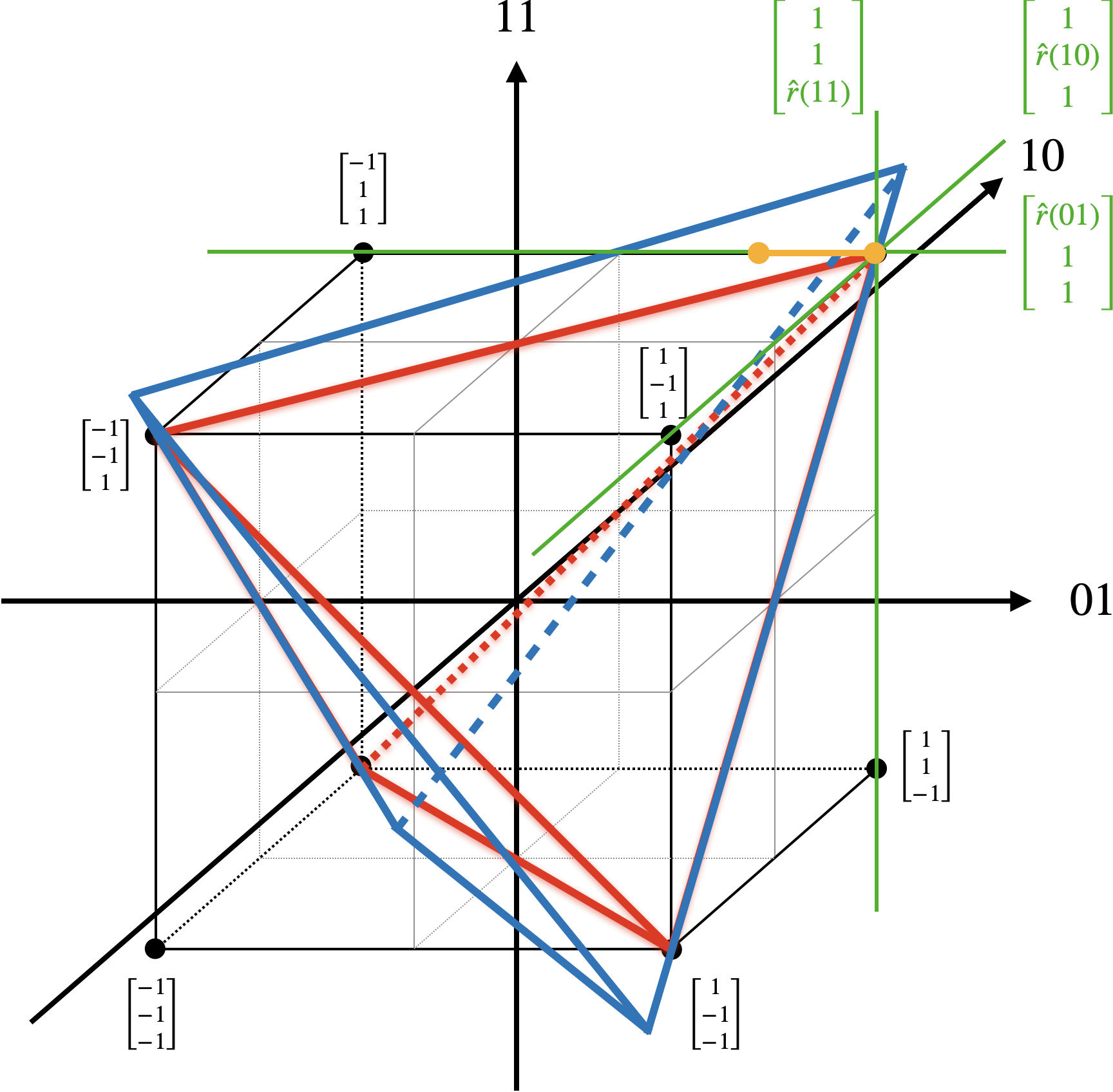}
  }
  \hfill
  \subfloat[]{
    \includegraphics[width=0.32\linewidth]{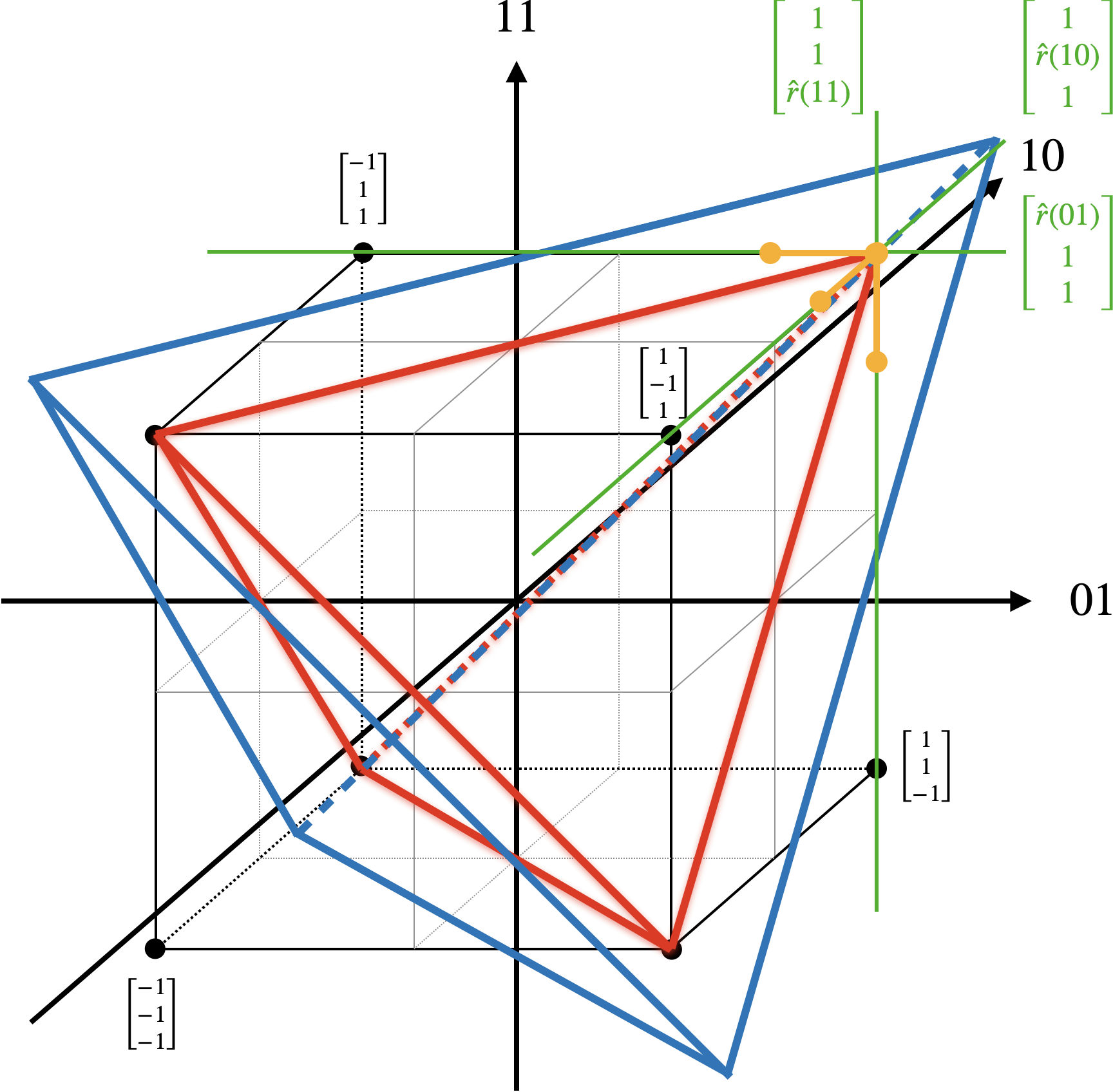}
  }
  \caption{
    Geometric action of PEC on two-qubit trap-response vectors in the Walsh--Hadamard representation.
    Since the component associated with the trivial pattern $00$ is fixed to one, we plot only the three non-trivial components $(01,10,11)$.
    The red tetrahedron is the set of $\hat{p}=H_{2}\vec{p}$ generated by all physical probability distributions $\vec{p}$, while each blue tetrahedron is the corresponding set of post-PEC vectors $\hat{r}=\hat{q}\odot\hat{p}$.
    The point $(1,1,1)$ represents the ideal identity response without any deviations detected.
    The green coordinate lines contain vectors for which exactly one non-trivial trap pattern differs from the ideal response; their yellow portions indicate intersections with the admissible post-PEC region and hence deviations that the Server can physically realize.
    (a) Before PEC, the physical region intersects these single-pattern lines only at $(1,1,1)$.
    (b) For PEC acting on the first qubit only, $\hat{q}=(1,1,1/(1-2p_{1}),1/(1-2p_{1}))^{T}$ stretches the $10$ and $11$ axes; the blue region has a non-trivial intersection only with the line associated with the trap pattern $01$.
    (c) For global two-qubit dephasing, $\hat{q}=(1,1/(1-p),1/(1-p),1/(1-p))^{T}$ stretches all three axes equally; the blue region intersects all three single-pattern lines, allowing the deviation to be localized to any non-trivial trap pattern.
  }
  \label{fig:polytopes_in_cube}
\end{figure*}

We remark in this section that applying PEC to the test rounds would undermine the protocol's verifiability.
One might consider the following natural alternative: instead of assuming that the relevant noise model admits an efficient description, the Client could apply PEC post-processing directly to the observed test statistics and compare the resulting estimate with the ideal identity response.
However, the inverse noise model used in PEC can distort the trap-response statistics in a way that a malicious Server can exploit to conceal its deviation.
We give explicit examples showing that, after PEC post-processing, the deviation from the ideal trap-response statistics can be concentrated on a single trap pattern among exponentially many possible trap patterns.
Consequently, detecting such an attack would require testing exponentially many trap patterns.

We first illustrate this phenomenon in a two-qubit example where PEC acts nontrivially only on one qubit.
In this local setting, as long as the non-trivial Pauli components in the noise model are localized on a constant-size subsystem, the Client can still efficiently identify the responding trap pattern by inspecting all trap patterns on that subsystem.
We then show that, once non-trivial deviations may occur on both qubits, the detectable effect can be localized to any trap pattern of the two-qubit system, so that all trap patterns must be checked.
Finally, we generalize this construction to a $|V|$-qubit system, where the number of trap patterns to be inspected grows as $2^{|V|}$.
This shows that applying PEC to the test rounds would destroy the efficiency of trap-based verification, particularly when the malicious Server knows the noise model.

\subsection{Applying PEC locally to the first qubit in a two-qubit system}

Let us first consider a two-qubit system and take the noise model $\mathcal{E}_{\mathrm{PEC}} = (1-p_{1})\mathsf{I}\mathsf{I} + p_{1}\mathsf{Z}\mathsf{I}$ as an illustrative example.
The probability distributions of $\mathcal{E}_{\mathrm{PEC}}$ and $\mathcal{E}_{\mathrm{PEC}}^{-1}$ can be expressed in vector form, respectively, as
\begin{equation}
  \begin{split}
    \vec{p}_{\mathrm{PEC}}
    = \left[\begin{array}{c}
                1 - p_{1} \\
                0         \\
                p_{1}     \\
                0         \\
              \end{array}\right]
    \begin{array}{l}
      \cdots~\mathsf{I}\mathsf{I} \\
      \cdots~\mathsf{I}\mathsf{Z} \\
      \cdots~\mathsf{Z}\mathsf{I} \\
      \cdots~\mathsf{Z}\mathsf{Z} \\
    \end{array}~, \quad
    \vec{q}
    = \left[\begin{array}{c}
                \frac{1-p_{1}}{1-2p_{1}} \\
                0                        \\
                \frac{-p_{1}}{1-2p_{1}}  \\
                0                        \\
              \end{array}\right].
  \end{split}
\end{equation}
Denoting the Server's actual deviation and its associated probability distribution by $\mathcal{E}$ and $\vec{p}$, the probability distribution of the effective deviation after PEC processing, $\mathcal{E}_{\mathrm{PEC}}^{-1}\circ\mathcal{E}$, corresponds to the convolution $\vec{r} = \vec{q}\ast\vec{p}$.

Let $H_{2}$ be the $4\times 4$ non-normalized Walsh--Hadamard transformation
\begin{equation}
  \begin{split}
    H_{2} :=
    \left[\begin{array}{rrrr}
              1 & 1  & 1  & 1  \\
              1 & -1 & 1  & -1 \\
              1 & 1  & -1 & -1 \\
              1 & -1 & -1 & 1
            \end{array}\right].
  \end{split}
\end{equation}
As discussed in Appendix~\ref{appendix:general_trap_patterns}, the sensitivity of each trap pattern to each Pauli deviation in $\{\mathsf{I},\mathsf{Z}\}^{|V|}$ can be understood through the Walsh--Hadamard transformation of the corresponding (quasi-)probability distribution.
Thus, defining $\hat{p}$ and $\hat{q}$ as the Walsh--Hadamard transforms of $\vec{p}$ and $\vec{q}$, we have
\begin{equation}
  \begin{split}
    \hat{p}
    = H_{2}\vec{p}
    = \left[\begin{array}{c}
                1           \\
                \hat{p}(01) \\
                \hat{p}(10) \\
                \hat{p}(11) \\
              \end{array}\right]
    \begin{array}{l}
      \cdots~00 \\
      \cdots~01 \\
      \cdots~10 \\
      \cdots~11 \\
    \end{array}~, \quad
    \hat{q}
    = H_{2}\vec{q}
    = \left[\begin{array}{c}
                1                  \\
                1                  \\
                \frac{1}{1-2p_{1}} \\
                \frac{1}{1-2p_{1}} \\
              \end{array}\right].
  \end{split}
\end{equation}
The first element, corresponding to the pattern $00$, is always $1$.
This makes it useful to visualize $\hat{p}$ by placing $(\hat{p}(01), \hat{p}(10), \hat{p}(11))$ in a three-dimensional coordinate system.
Since $\|\vec{p}\|_{1}=1$ and $\vec{p}\geq 0$, the set of all possible values of $\hat{p}$ forms a tetrahedron in this three-dimensional space, shown as the red region in Fig.~\ref{fig:polytopes_in_cube}(a).

Using the property of the Walsh--Hadamard transformation and convolution, we also see that

  \begin{equation}
    \begin{split}
      \hat{r}
      = H_{2}(\vec{q} \ast \vec{p})
      = (H_{2}\vec{q}) \odot (H_{2}\vec{p})
      = \hat{q}\odot\hat{p}
      = \left[\begin{array}{c}
                  1                            \\
                  \hat{p}(01)                  \\
                  \frac{\hat{p}(10)}{1-2p_{1}} \\
                  \frac{\hat{p}(11)}{1-2p_{1}} \\
                \end{array}\right],
    \end{split}
  \end{equation}

where $\odot$ denotes the Hadamard product, i.e., entry-wise multiplication.
This representation shows that $\hat{q}$ acts as an entry-wise scaling factor on $\hat{p}$.
Since $\hat{r}(00)=1$, $(\hat{r}(01), \hat{r}(10), \hat{r}(11))$ can also be represented in a three-dimensional coordinate system.
The set of admissible values is obtained from the tetrahedron for $\hat{p}$ by scaling each coordinate axis by the corresponding component of $\hat{q}$.
For the noise model $\mathcal{E}_{\mathrm{PEC}}$, the corresponding $\hat{q}$ stretches $\hat{p}(10)$ and $\hat{p}(11)$ along their respective coordinate axes by a factor of $1/(1-2p_{1})$.
This gives the blue tetrahedron shown in Fig.~\ref{fig:polytopes_in_cube}(b).

For $\hat{p}$ in the red tetrahedron, $(\hat{p}(01), \hat{p}(10), \hat{p}(11)) = (1, 1, 1)$ if and only if $\vec{p} = [1,0,0,0]^{T}$.
For any valid probability distribution $\vec{p}\neq [1,0,0,0]^{T}$, at least two elements among $\hat{p}(01)$, $\hat{p}(10)$, and $\hat{p}(11)$ become smaller than $1$ and are therefore sensitive to $\vec{p}$.

In fact, in this example, the Server can construct a deviation $\vec{p}$ such that only the trap pattern $01$ is sensitive, i.e., $\hat{r} = [1, \hat{r}(01), 1, 1]^{T}$.
Solving this $\hat{r}$ for $\vec{p}$ and $\hat{p}$, we obtain
\begin{equation}
  \begin{split}
    \hat{p}
    =\left[\begin{array}{c}
               1           \\
               \hat{r}(01) \\
               1 - 2p_{1}  \\
               1 - 2p_{1}
             \end{array}\right], \quad
    \vec{p}
    =\frac{1}{4}\left[\begin{array}{c}
                          3 + \hat{r}(01)- 4p_{1}  \\
                          1 - \hat{r}(01)          \\
                          \hat{r}(01) - 1 + 4p_{1} \\
                          1 - \hat{r}(01)
                        \end{array}\right].
  \end{split}
\end{equation}
Here, $\hat{p}$ corresponds to a non-trivial valid probability distribution when $\vec{p}(\mathsf{I}\mathsf{Z}) = \hat{r}(01)$ ranges over $1-4p_{1} \leq \vec{p}(\mathsf{I}\mathsf{Z}) \leq 1$.
This corresponds to the yellow line segment in Fig.~\ref{fig:polytopes_in_cube}(b).
Thus, the Server can hide its deviation so that only the trap pattern $01$ is sensitive to it.

We also verify that the Server cannot construct a non-trivial probability distribution such that only the trap pattern $11$ is sensitive, i.e., $\hat{r} = [1, 1, 1, \hat{r}(11)]^{T}$.
Solving this $\hat{r}$ for $\vec{p}$ and $\hat{p}$, we obtain
\begin{equation}
  \begin{split}
    \hat{p}
     & = \left[\begin{array}{c}
                   1          \\
                   1          \\
                   1 - 2p_{1} \\
                   (1 - 2p_{1}) \hat{r}(11)
                 \end{array}\right],                          \\
    \vec{p}
     & = \frac{1}{4}\left[\begin{array}{c}
                              3 - 2p_{1} + (1 - 2p_{1}) \hat{r}(11) \\
                              (1 - 2p_{1}) (1 - \hat{r}(11))        \\
                              1 + 2p_{1} - (1 - 2p_{1}) \hat{r}(11) \\
                              - (1 - 2p_{1}) (1 - \hat{r}(11))
                            \end{array}\right].
  \end{split}
\end{equation}
Here, we notice that $\vec{p}(\mathsf{I}\mathsf{Z}) = -\vec{p}(\mathsf{Z}\mathsf{Z})$.
Thus, $\hat{r}(11)$ can only take $\hat{r}(11) = 1$ for $\vec{p}$ to be a valid probability distribution, which is equivalent to the honest deviation $\vec{p}_{\mathrm{PEC}}$.
This implies that the Server cannot hide its deviation solely in $\hat{r}(11)$.
Since the trap patterns $10$ and $11$ are symmetric, the same argument shows that there is no non-trivial probability distribution that attacks only $10$.

Visually, the sets for which exactly one trap is triggered are indicated by the green line segments.
The red tetrahedron shares only the point $(1,1,1)$ with these segments, whereas the blue tetrahedron intersects the line $(\hat{r}(01), 1,1)$ over the range $1-4p_{1} \leq \hat{r}(01) \leq 1$.
On the other hand, the lines $(1, \hat{r}(10), 1)$ and $(1,1, \hat{r}(11))$ meet the blue tetrahedron only at $(1,1,1)$.
Hence, $\hat{r}(01)=1$ automatically implies $\hat{r}(10)=\hat{r}(11)=1$.

Therefore, the Client would need to examine the trap patterns on this local subsystem to identify the one that responds to the deviation.
This example already shows that PEC post-processing can localize the detectable effect of a deviation.
However, since the non-trivial noise is confined to a known local subsystem, this by itself does not yet imply an exponential obstruction: all relevant local trap patterns can still be inspected directly.

\subsection{Applying PEC for global dephasing in a two-qubit system}

Next, we consider the global two-qubit dephasing noise model
\begin{equation}
  \begin{split}
    \mathcal{E}_{\mathrm{PEC}}
     & = \left(1-\frac{3p}{4}\right)\mathsf{I}\mathsf{I}
    + \frac{p}{4}\left(
    \mathsf{Z}\mathsf{I}+\mathsf{I}\mathsf{Z}+\mathsf{Z}\mathsf{Z}
    \right).
  \end{split}
\end{equation}
Here, $0<p<1$, where the upper bound ensures that the inverse map exists.
The probability distributions of $\mathcal{E}_{\mathrm{PEC}}$ and $\mathcal{E}_{\mathrm{PEC}}^{-1}$ can be expressed, respectively, as
\begin{equation}
  \begin{split}
    \vec{p}_{\mathrm{PEC}}
     & = \left[\begin{array}{c}
                   1-3p/4 \\
                   p/4    \\
                   p/4    \\
                   p/4    \\
                 \end{array}\right]
    \begin{array}{l}
      \cdots~\mathsf{I}\mathsf{I} \\
      \cdots~\mathsf{I}\mathsf{Z} \\
      \cdots~\mathsf{Z}\mathsf{I} \\
      \cdots~\mathsf{Z}\mathsf{Z} \\
    \end{array}~,   \\
    \vec{q}
     & = \frac{1}{4(1-p)}
    \left[\begin{array}{c}
              4-p \\
              -p  \\
              -p  \\
              -p  \\
            \end{array}\right].
  \end{split}
\end{equation}
Then, the Walsh--Hadamard transforms of the noise model $\vec{p}_{\mathrm{PEC}}$ and its inverse $\vec{q}$ become
\begin{equation}
  \begin{split}
    \hat{p}_{\mathrm{PEC}}
     & = H_{2}\vec{p}_{\mathrm{PEC}}
    = \left[\begin{array}{c}
                1   \\
                1-p \\
                1-p \\
                1-p \\
              \end{array}\right]
    \begin{array}{l}
      \cdots~00 \\
      \cdots~01 \\
      \cdots~10 \\
      \cdots~11 \\
    \end{array}~,                  \\
    \hat{q}
     & = H_{2}\vec{q}
    = \left[\begin{array}{c}
                1             \\
                \frac{1}{1-p} \\
                \frac{1}{1-p} \\
                \frac{1}{1-p} \\
              \end{array}\right].
  \end{split}
\end{equation}

Using the property of the Walsh--Hadamard transformation and convolution, we also see that
\begin{equation}
  \begin{split}
    \hat{r}
    = \hat{q}\odot\hat{p}
    = \left[\begin{array}{c}
                1                       \\
                \frac{\hat{p}(01)}{1-p} \\
                \frac{\hat{p}(10)}{1-p} \\
                \frac{\hat{p}(11)}{1-p} \\
              \end{array}\right],
  \end{split}
\end{equation}
Thus, $\hat{q}$ acts as an entry-wise scaling factor on $\hat{p}$ along all axes.
This gives the blue tetrahedron shown in Fig.~\ref{fig:polytopes_in_cube}(c).

In this case, the Server can localize its deviation along any axis.
For example, solving $\hat{r} = [1, \hat{r}(01), 1, 1]^{T}$ for $\vec{p}$ and $\hat{p}$, we obtain
\begin{equation}
  \begin{split}
    \hat{p}
     & = \hat{p}_{\mathrm{PEC}}
    + \left[\begin{array}{c}
                0                      \\
                (1-p)(\hat{r}(01) - 1) \\
                0                      \\
                0
              \end{array}\right], \\
    \vec{p}
     & = \vec{p}_{\mathrm{PEC}}
    + \frac{(1-p)(\hat{r}(01) - 1)}{4}
    \left[\begin{array}{r}
              1  \\
              -1 \\
              1  \\
              -1 \\
            \end{array}\right].
  \end{split}
\end{equation}
Here, $\vec{p}\geq 0$ is satisfied, for example, throughout the non-trivial interval
\begin{equation}
  \begin{split}
    1-\frac{p}{1-p} \leq \hat{r}(01) \leq 1.
  \end{split}
\end{equation}
By symmetry of the global dephasing model, the Server can likewise localize its deviation to $\hat{r}(10)$ if
\begin{equation}
  \begin{split}
    1-\frac{p}{1-p} \leq \hat{r}(10) \leq 1,
  \end{split}
\end{equation}
and, by solving $\hat{r} = [1, 1, 1, \hat{r}(11)]^{T}$, to $\hat{r}(11)$ if
\begin{equation}
  \begin{split}
    1-\frac{p}{1-p} \leq \hat{r}(11) \leq 1.
  \end{split}
\end{equation}
These conditions correspond to the yellow line segments within the blue tetrahedron in Fig.~\ref{fig:polytopes_in_cube}(c).
Consequently, the Client cannot certify the absence of Server's deviation by checking only a fixed local subset of trap patterns.
Instead, in this two-qubit example, all non-trivial trap patterns need to be inspected.

\subsection{Applying PEC for global dephasing in a \texorpdfstring{$|V|$}{|V|}-qubit system}

We now generalize the global dephasing model above to a $|V|$-qubit system.
Let $d:=2^{|V|}$, and identify each $u\in\{0,1\}^{|V|}$ with the Pauli operation
\begin{equation}
  \begin{split}
    \mathsf{Z}_{u}:=\bigotimes_{i=1}^{|V|}\mathsf{Z}_{u_{i}}.
  \end{split}
\end{equation}
Since $\mathsf{X}$ components are harmless in MBQC, grouping the Pauli components of a global dephasing channel by their effective $\mathsf{Z}$ components gives the following global dephasing model:

  \begin{equation}
    \begin{split}
      \mathcal{E}_{\mathrm{PEC}}
       & = (1-p)\mathsf{I}^{\otimes |V|}
      + \frac{p}{d}\sum_{u\in\{0,1\}^{|V|}}\mathsf{Z}_{u}          \\
       & = \left(1-\frac{(d-1)p}{d}\right)\mathsf{I}^{\otimes |V|}
      + \frac{p}{d}\sum_{u\in\{0,1\}^{|V|}\setminus\{0_{V}\}}\mathsf{Z}_{u},
    \end{split}
  \end{equation}

where $0<p<1$.
Its probability distribution is
\begin{equation}
  \begin{split}
    \vec{p}_{\mathrm{PEC}}(u)
     & = \begin{cases}
           \displaystyle 1-\frac{(d-1)p}{d} & \text{if}~~u=0_{V}, \\
           \displaystyle \frac{p}{d}        & \text{otherwise}.
         \end{cases}
  \end{split}
\end{equation}
For the non-normalized Walsh--Hadamard transformation $H_{|V|}$ with entries $(H_{|V|})_{v,u}=(-1)^{\langle\mathrm{set}(v), \mathrm{set}(u)\rangle}$, its transform is
\begin{equation}
  \begin{split}
    \hat{p}_{\mathrm{PEC}}(v)
     & = \begin{cases}
           1   & \text{if}~~v=0_{V}, \\
           1-p & \text{otherwise}.
         \end{cases}
  \end{split}
\end{equation}
Accordingly, the Walsh--Hadamard transform of the quasi-probability distribution $\vec{q}$ for the inverse channel is
\begin{equation}
  \begin{split}
    \hat{q}(v)
     & = \begin{cases}
           1                     & \text{if}~~v=0_{V}, \\
           \displaystyle 1/(1-p) & \text{otherwise},
         \end{cases}
  \end{split}
\end{equation}
and applying the inverse Walsh--Hadamard transformation gives
\begin{equation}
  \begin{split}
    \vec{q}(u)
     & = \begin{cases}
           \displaystyle\frac{d-p}{d(1-p)} & \text{if}~~u=0_{V}, \\[10pt]
           \displaystyle-\frac{p}{d(1-p)}  & \text{otherwise}.
         \end{cases}
  \end{split}
\end{equation}

Choosing an arbitrary non-trivial trap pattern $w\in\{0,1\}^{|V|}\setminus\{0_{V}\}$, suppose that the Server chooses its deviation to be sensitive only at $w$, namely
\begin{equation}
  \begin{split}
    \hat{r}(v)
     & = \begin{cases}
           1-\eta & \text{if}~~v=w,   \\
           1      & \text{otherwise}.
         \end{cases}
  \end{split}
\end{equation}
Since $\hat{r}=\hat{q}\odot\hat{p}$ and $\hat{q}\odot\hat{p}_{\mathrm{PEC}}=\mathbf{1}_{V}$, the required Walsh--Hadamard transform of the Server's deviation is
\begin{equation}
  \begin{split}
    \hat{p}
     & = \hat{p}_{\mathrm{PEC}}-(1-p)\eta\mathbf{e}_{w}.
  \end{split}
\end{equation}
Applying the inverse Walsh--Hadamard transformation, we obtain
\begin{equation}
  \begin{split}
    \vec{p}(u)
     & = \vec{p}_{\mathrm{PEC}}(u)
    -\frac{(1-p)\eta}{d}(-1)^{u\cdot w}.
  \end{split}
\end{equation}
For $|V|\geq 2$, every non-zero $w$ has a non-zero $u$ satisfying $u\cdot w=0$.
The corresponding component of $\vec{p}$ is $[p-(1-p)\eta]/d$, while all other non-negativity constraints are weaker.
Therefore, $\vec{p}$ is a valid non-trivial probability distribution whenever
\begin{equation}\label{eq:condition_insensitivity_V}
  \begin{split}
    0<\eta\leq\min\left\{1,\frac{p}{1-p}\right\}.
  \end{split}
\end{equation}
Hence, for the global $|V|$-qubit dephasing model, the Server can localize its deviation to an arbitrary non-trivial trap pattern.
Detecting such a deviation in the worst case requires checking all $2^{|V|}-1$ non-trivial trap patterns.
This incurs exponential overhead and explains why applying PEC directly to the test rounds is incompatible with an efficient trap-based security guarantee.


\bibliography{main}


\end{document}